%% file: paper.tex
\documentclass[12pt]{article}
\input{preamble/preamble}


\input{preamble/preamble_tikz}

\title{%
\textbf{Restricted Search Space Graph MCMC via Birth-Death Processes}\thanks{We thank Piotr Zwiernik and Ricardo Baptista for their constructive suggestions that have improved the paper. Funding support for the work was provided by NSERC of Canada.}
}

\author{%
Morris Greenberg$^{1}$\thanks{Corresponding author: \texttt{morris.greenberg@mail.utoronto.ca}}, 
Kieran R Campbell$^{1,2,3,4,5}$, 
Radu Craiu$^{1,4}$\\
\small\itshape
$^{1}$Department of Statistical Sciences, University of Toronto, Toronto, ON, Canada\\
\small\itshape
$^{2}$Lunenfeld-Tanenbaum Research Institute, Sinai Health System, Toronto, ON, Canada\\
\small\itshape
$^{3}$Department of Molecular Genetics, University of Toronto, Toronto, ON, Canada\\
\small\itshape
$^{4}$Vector Institute, Toronto, ON, Canada\\
\small\itshape
$^{5}$Ontario Institute of Cancer Research, Toronto, ON, Canada}

\date{}
\begin{document}

\maketitle
\bigskip

\begin{abstract}
Inferring directed acyclic graphs (DAGs) from data via Markov chain Monte Carlo (MCMC) is computationally challenging in moderate-to-high dimensional settings because their discrete sampling space grows super-exponentially with the number of nodes. To address scalability, several recent MCMC-based graph inference methods restrict the search space to a subset of edges, at the cost of introducing error into the inference procedure.

In this work, we derive sharp lower and upper bounds on the total variation distance between the unrestricted posterior distribution and the posterior distribution induced by a state-of-the-art restricted search space MCMC method. These bounds characterize regimes in which the approximation error is negligible and regimes in which it is not. In order to reduce the error, we propose a flexible transdimensional MCMC sampler which allows the search space to expand or contract dynamically as the chain progresses. The sampler is defined by birth-and-death rates that induce a prior distribution on the set of search spaces, rather than assume a fixed restricted search space throughout. We outline an efficient implementation of the proposed algorithm and demonstrate its finite-sample performance through simulation studies.

\end{abstract}


\section{Introduction}
Understanding the dependence structure of a set of random variables is important in multivariate statistical analysis. \textit{Graphical models} can be used to represent statistical models with conditional independence structures, and \textit{Bayesian networks} are the special cases which can be represented by a directed acyclic graph (DAG). In a Bayesian network, each vertex in the DAG represents a unique variable, and a directed edge between two vertices corresponds to conditional dependence. When data are generated from a multivariate normal distribution, the DAG dictates the sparsity pattern of the precision matrix. In addition, the acyclicity of the DAG implies possible causal relationships between variables, although it is unclear to what extent Bayesian networks facilitate causal discovery \citep{dawid_beware_2010}. In most problems, the DAG structure is not known a priori, which motivates the development of \textit{structure discovery} methods that infer the Bayesian network from data.

When inferring a Bayesian network, data $D$ are used to estimate the DAG $G$ and the set of parameters $\theta$. Although it is straightforward to either evaluate the likelihood or generate data (that is, compute or sample from $P(D | G, \theta)$) or infer the set of parameters ($P(\theta | G, D)$) given a known structure, inferring $G$ and its associated posterior distribution $\pi(G | D)$ is more challenging. This is mainly because (i) it is computationally infeasible to consider all DAGs that span a set of $p$ variables unless $p$ is small, and (ii) \textit{ score equivalence} of DAGs can lead to identifiability problems \citep{chickering_optimal_structure_2002}. As a result, to estimate $P(G | D)$, exact Markov chain Monte Carlo (MCMC) methods can be employed for small $p$, but approximate methods (such as restricting the set of edges considered in the search space, often called \textit{hybrid} samplers) are typically used for larger networks. These techniques greatly improve scalability, but can also introduce systematic error in posterior estimates, an issue that is minimally quantified in the existing literature.

This paper contribution is multivalent. First, we establish lower and upper error bounds on posterior estimates for hybrid MCMC sampling methods on the topological order space. Second, we develop a transdimensional MCMC sampler that gives practitioners direct ability to control the fidelity-scalability tradeoff, and outperforms competing methods on large classes of graphs by avoiding the error from fixing the restricted search space. Thirdly, we implement a number of computational swindles that reduces the computational run time by orders of magnitude to allow for scalable transdimensional sampling. 

We introduce notation and review previous work in Section \ref{sec:background}. In Section \ref{sec:error_bound}, we derive error bounds for existing fixed-space samplers and discuss how transdimensional samplers can improve them. In Section \ref{sec:brood_intro}, we propose \textbf{B}irth-death processes \textbf{R}estricted \textbf{O}ver \textbf{O}rder \textbf{D}istributions (BROOD), a novel algorithm that efficiently targets transdimensional posteriors. We compare numerically the performance of BROOD with that of other algorithms in Section \ref{sec:results_sim}, and conclude in Section \ref{sec:discussion}.

\section{Background}\label{sec:background}
\subsection{Graphs and Bayesian Networks}

A directed acyclic graph (DAG) $G= (V, E_G)$ consists of a set of vertices $V$ (also called nodes) and directed edges $E_G$ such that no directed cycles exist (i.e., no node can be both an ancestor and a descendant of another). An edge exists from node $i$ to node $j$ is denoted as $i \rightarrow j$, where $i$ is a \textit{parent} of $j$. We denote the set of parents of node $i$ as $pa_{G}(i)$ or $pa(i)$ when the graph is clear from context.

One popular application of DAGs is in probabilistic graphical modeling, where DAGs represent minimal factorizations of joint distributions into conditional and marginal components. Specifically, in a Bayesian network $\mathcal{B}= (G, \theta)$, each vertex $i \in V$ corresponds to a random variable $X_i$, $E_G$ encodes the dependence structure among the variables, and $\theta$ denotes the parameters specifying the functional forms of the conditional distributions. We denote the random variables associated with the parents of $i$ in $G$ as $\mathbf{Pa}_{G}(i)$ (or simply $\mathbf{Pa}(i)$), and write the observed data generated from a $p$-dimensional Bayesian network as $D=(X_1,...,X_p)$. Given $\theta$ and $G$, the joint distribution $P(X_1,...,X_p)$ factorizes according to the \textit{Markov property} where each variable $X_i$ is conditionally independent of its non-descendants given $ \mathbf{Pa}(i)$. Consequentially, the joint distribution is expressed as:
\begin{align}
  P(X_1,...,X_p)  = \prod_{i=1}^{p} P(X_i | \mathbf{Pa}(i)).
\end{align}

Every DAG $G$ is consistent with at least one \textit{topological order} (also called \textit{topological ordering}, \textit{order}, or \textit{ordering}), a permutation of $\{1,...,p\}$ that exhibits the ancestral relations of $G$. Formally, a topological order $\prec \in \mathbb{S}^p$ (the symmetric group of degree $p$) is compatible with $G$ if for every edge $i \rightarrow j$ in $G$, we have $\prec_{[i]} < \prec_{[j]}$. Furthermore, $i$ is a \textit{predecessor} of $j$ in $\prec$, or $i \in pr_{\prec}(j)$ for short. A four-node DAG is illustrated in figure \ref{fig:dags_orders_a} with two example compatible orders in \ref{fig:dags_orders_b}. We define $\mathcal{G}_{\prec} := \{G: G \in \mathbb{G}_p, pa_{G}(i) \subseteq pr_{\prec}(i) \text{ for all } i\}$ (where $\mathbb{G}_p$ is the set of $p$-node DAGs) as the set of DAGs consistent with an order $\prec$. For a given order $\prec$, we define $\mathbf{Pa}_{\prec}(i) = \cup_{\{G:G \in \mathcal{G}_{\prec}\}} \mathbf{Pa}_{G}(i)$.

For the purpose of Bayesian inference, we treat the graph and the topological order as random variables denoted by $\bm{G}, \bm{O}$, respectively, which take values in the spaces $\mathbb{G}_p, \mathbb{S}^p$. We use $G, \prec$ to refer to specific realizations of these variables.

\begin{figure}[t]
\centering

\begin{subfigure}{0.45\textwidth}
\centering
\begin{tikzpicture}[
    node/.style={circle, draw, thick, minimum size=7mm},
    ->, >=Stealth, thick
  ]
  \node[node, draw=purple] (1) at (1,2) {1};
  \node[node, draw=purple] (2) at (2,0) {2};
  \node[node, draw=purple] (3) at (0,0) {3};
  \node[node, draw=purple] (4) at (3,2) {4};

  \draw (1) -- (2);
  \draw (1) -- (3);
  \draw (4) -- (2);

  \node at (0,3.5) {(a)};
\end{tikzpicture}
\caption{}
\label{fig:dags_orders_a}
\end{subfigure}
\hfill
\begin{subfigure}{0.45\textwidth}
\centering
\begin{tikzpicture}[
    node/.style={circle, draw, thick, minimum size=7mm},
    ->, >=Stealth, thick
  ]
  \node[node, draw=purple] (2) at (0,0) {2};
  \node[node, draw=purple] (4) at (1.5,0) {4};
  \node[node, draw=purple] (3) at (3,0) {3};
  \node[node, draw=purple] (1) at (4.5,0) {1};

  \draw (4) -- (2);
  \draw (1) -- (3);
  \draw (1) .. controls +(up:10mm) and +(up:10mm) .. (2);

  \node at (0,3.5) {(b)};

  \node[node, draw=purple] (3b) at (0,2) {3};
  \node[node, draw=purple] (2b) at (1.5,2) {2};
  \node[node, draw=purple] (1b) at (3,2) {1};
  \node[node, draw=purple] (4b) at (4.5,2) {4};

  \draw (4b) .. controls +(up:10mm) and +(up:10mm) .. (2b);
  \draw (1b) -- (2b);
  \draw (1b) .. controls +(down:10mm) and +(down:10mm) .. (3b);

\end{tikzpicture}
\caption{}
\label{fig:dags_orders_b}
\end{subfigure}

\caption{(a) Example DAG with 4 nodes (b) 2 example topological order representations of the DAG in (a). In both topological orders, parents are always to the right of children. }\label{fig:dags_orders}
\end{figure}
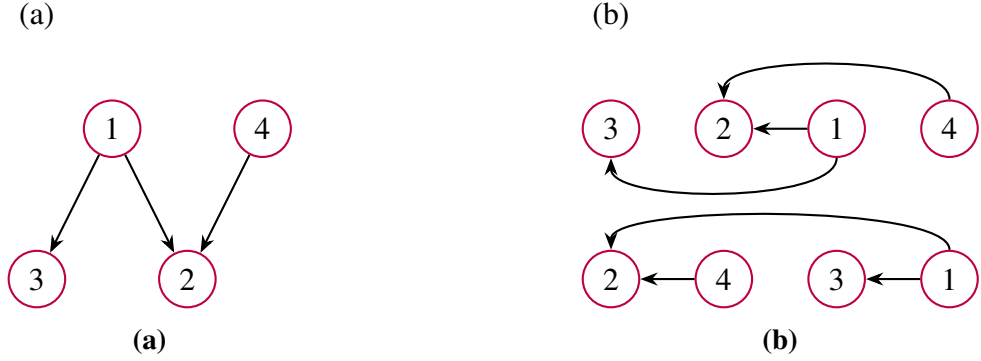

\subsection{Graph Inference Methods}
There are three main approaches to structure inference: \textit{constraint-based}, \textit{score-based}, and \textit{hybrid} methods. We summarize each below; for a complete review, see \cite{kitson_survey_2023}.

\subsubsection{Constraint-Based Methods}
Constraint-based methods infer the graph structure by estimating conditional dependencies via hypothesis testing. Since DAG edges correspond to conditional dependencies between nodes, these methods estimate a \textit{skeleton} (the undirected edge structure) by conducting a complete set of independence tests and retaining the edges that correspond to null hypotheses that are rejected. Edge directions are then added in a post-processing step to infer the DAG. 

Constraint-based methods are consistent estimators of the true DAG under mild assumptions (\citealt{kalisch_estimating_2007}), but can be unstable. Because tests are often correlated across sets of random variables, slight misspecification can cause many false negatives, resulting in true edge omissions (\citealt{uhler_geometry_2013}). In addition, the variable ordering used during testing can significantly influence the estimated graph (\citealt{colombo_order-independent_2014}).

\subsubsection{Score-Based Methods}
Score-based methods (sometimes called score-and-search-based methods) infer structure by evaluating an objective function (\textit{score}) that quantifies the relative merits of each graph. The score is then used to estimate a single best graph by searching for a maximum or to estimate a distribution over structures by sampling graphs proportionally to their scores. Unique score-based approaches vary in (1) the choice of score and (2) the strategy used to explore the space. We discuss desirable properties for score functions in \ref{subsec:score_funcs} and discuss sampling techniques in \ref{subsec:dag_mcmc_methods}.

When computationally feasible, score-based methods tend to outperform constraint-based ones, as they directly work with an objective tied to model quality. However, they can scale poorly due to the exponential size of the graph space with $p$ nodes. As a result, score-based methods are typically used in sparse graph problems or problems with few variables.

\subsubsection{Hybrid Methods}
Hybrid methods aim to combine the strengths of constraint-based and score-based approaches to graph inference. A common approach is to use a computationally efficient constraint-based approach to build a reduced search space, and then use a deeper score-based search within this restricted space. We describe a recent MCMC variant of this in Section \ref{subsubsec:hybrid_mcmc}.

\subsection{Score Functions}\label{subsec:score_funcs}
We define a graph score $Q(G, D)$ to be any objective function used to compare graphs. Scores are typically proportional to either (1) a model selection criterion (such as Bayesian information criterion (BIC) or Akaike information criterion (AIC)) or (2) to a posterior probability. We define a posterior score as the special case in which the score can be mapped to the posterior probability of the graph via a strictly monotone transformation $g$, i.e.,
\begin{align}
    Q(G, D) = g(\pi_{DAG}(G | D)).
\end{align}
Examples of nontrivial $g$ include the unnormalized posterior $P(D | \bm{G}=G)\pi(G)$ (where $g(x) = xP(D)$) or its log form ($g(x) = \log{x} +\log{ P(D)}$). Posterior scores therefore provide a direct Bayesian interpretation of the scoring function. We denote $(Q, \pi_{DAG})$ as a coupled posterior score and its associated posterior distribution, noting that any monotone transformation $Q'$ of $Q$ yields the same $\pi_{DAG}$.

While we introduce scores here in the graph space, it is often convenient (see Section \ref{subsubsec:straight_mcmc}) to conduct the search in reparameterized spaces to achieve efficiency gains. In those settings, we will analogously define scores and posterior scores over the corresponding space.

\subsubsection{Computationally Efficient Scores}\label{subsubsec:eff_scores}
To ensure computational tractability, most score-based methods utilize \textit{decomposable} scores. A score $Q(G, D)$ is decomposable iff $Q(G, D) = \prod_{i=1}^p S(X_i, \mathbf{Pa}(i) | D)$, where $S$ depends only on a node and its parents (sometimes, this is written in logarithmic scale where the product becomes a sum). These scores are especially useful in local update procedures, as only the terms associated with the affected node are reevaluated (whereas nondecomposable scores require rescoring the entire graph for each node-wise update). 

\cite{friedman_being_2003} found sufficient conditions for decomposable Bayesian network posterior scores. Notable examples include the BDe score for multinomial data and the BGe score for multivariate Gaussian data \citep{heckerman_learning_1995}.

Another approach for modular scoring is to use priors that induce \textit{structural Markov laws} \citep{byrne_structural_2015}, graph distributions with strong conditional independence properties. \cite{ben-david_high_2015} used these laws to create the DAG-Wishart distribution, a flexible generalization of the BGe score that still yields decomposable scores.

\subsubsection{Identifiable Scores within Equivalence Classes}\label{subsubsec:identifiable}
Another consideration in design of scores is allowing for identifiability. Two DAGs belong to the same \textit{Markov equivalence class} if they encode the same set of independence relationships \citep{drton_structure_2017}; when the true DAG belongs to a nontrivial equivalence class, structure learning from observational data alone becomes unidentifiable. Many score functions, including the BDe and BGe scores, are \textit{score equivalent}, meaning they assign identical scores to all graphs in the equivalence class \citep{geiger_parameter_2002}, thus preserving this ambiguity.

Recent methods achieve identifiability by leveraging the unique constraints of topological orders, though often at the cost of computational or model complexity. For instance, the topological order is identifiable under structural equation models with equal error variance \citep{pmlr-v84-ghoshal18a, chen_causal_2019}. While \cite{chang_orderbased_2023} utilized this to design an identifiable empirical Bayes score, the resulting function is non-decomposable and requires approximate sampling. Alternatively, the DAG-Wishart distribution \citep{ben-david_high_2015} maintains decomposability but achieves identifiability through order-dependent priors. This approach penalizes Markov equivalent structures differently based on their directionality, which risks model misspecification if the assumed ordering is incompatible with the underlying data.

\subsection{DAG MCMC Methods}\label{subsec:dag_mcmc_methods}
As the number of nodes $p$ in a graph grows, brute-force search methods become infeasible for 2 reasons: (1) the number of graphs in the search space grows super-exponentially, and (2) many scores introduce intractable terms, such as posteriors in Gaussian graph scoring setups. As a result, MCMC techniques are used to score and search the space.

\subsubsection{MCMC Setups}\label{subsubsec:straight_mcmc}

MCMC algorithms for graphs can vary in 3 key ways: (1) the score function used, (2) the state space used, and (3) the move proposal setup (both in neighborhood choice for proposals, and the proposal scheme itself).

As outlined in \ref{subsec:score_funcs}, choice of score function is motivated by either (1) minimizing the computational cost of score updates or (2) encoding information (such as identifiability) about unique graphs. Typically, MCMC techniques use posterior scores derived from graph priors that invoke the desired score properties.

The simplest state space for graph inference is the set of DAGs. \textit{Structure MCMC} (\citealt{madigan_1995_BayesianGM}, \citealt{giudici_improving_2003}) is a random-walk Metropolis-Hastings (MH) on this space, proposing moves by adding, deleting, or reversing one edge from the current DAG while ensuring acyclicity. More sophisticated moves on the DAG space include block-Gibbs sampling \citep{goudie_gibbs_2016}, large-scale edge reversals \citep{grzegorczyk_improving_2008}, adaptive neighborhood sets (\citealt{liang_adaptive_2022}, \citealt{caron_structure_2024}), and transdimensional methods (\citealt{giudici_decomposable_1999}, \citealt{mohammadi_bayesian_2015}, \citealt{dobra_loglinear_2018}), the last of which we discuss further in Section \ref{subsubsec:transdim_mcmc}.

Alternatively, grouping related DAGs as an element of the state space can improve mixing. For instance, sampling over equivalence classes (\citealt{he_reversible_2013}, \citealt{castelletti_learning_2018}) avoids transitions between indistinguishable graphs under score equivalent scores. Another approach is to group by topological orders (\citealt{friedman_being_2003}, \citealt{kuipers_efficient_2021}) or partial orders (\citealt{kuipers_partition_2017}). Since order identifiability implies graph identifiability, samplers on the order space can distinguish between otherwise equivalent DAGs. 

\textit{Order MCMC} is the class of graph MCMC techniques that operate on the topological order space. Beyond enabling larger move proposals than DAG-based samplers and distinguishing between equivalent graphs, the order space offers two further advantages. First, while it is NP-hard to find an optimal scoring graph (\citealt{chickering_learning_1996}, \citealt{chickering_large-sample_2004}), the problem is only $O(p^K)$ given a known order $\prec$ where $K$ is the \textit{sparsity} - maximal parent set size. This reduces structure learning to a variable selection task over the $\binom{p}{2}$ possible edges, instead of a simultaneous identification and selection task (\citealt{buntine_theory_1991}, \citealt{cooper_bayesian_1992}). Second, order scores efficiently factorize over nodes when using decomposable graph scores. If $(Q(G, D), \pi_{DAG}(G | D))$ is a posterior graph score and distribution pair, the marginal posterior distribution over orders $\pi_{Ord}$ (with associated score $R$) is:
\begin{align}
\pi_{Ord}(\prec | D) \propto R(\prec | D) &\propto \sum_{G \in \mathcal{G}_{\prec}} \pi_{DAG}(G | D).
\end{align}

For decomposable graph scores, this admits the factorized form,
\begin{align}
    R(\prec | D) &\propto \sum_{G:G \in \mathcal{G}_{\prec}}\prod_{i=1}^p S(X_i, \mathbf{Pa}_{G}(i) | D) = \prod_{i=1}^p \sum_{\mathbf{Pa} \in \mathbf{Pa}_\prec(i)} S(X_i, \mathbf{Pa} | D).
\end{align}
After sampling $\prec$ via MH, graphs can then be sampled from $\pi_{DAG}(G | \prec, D)$.

Note that order MCMC estimates are biased  towards graphs with more compatible orderings (namely, sparser graphs). This motivates substituting orders with partial orders, though both the search and scoring steps become more computationally expensive.

\subsubsection{Hybrid MCMC Setups}\label{subsubsec:hybrid_mcmc}

Recent work (\citealt{kuipers_efficient_2021}, \citealt{viinikka_towards_2020}) further improves the efficiency and scalability of order MCMC by restricting the search space to a subset of all possible edges, estimated via constraint-based methods in a hybrid setup. While in theory, a restricted search space could be any arbitrary subset of $\mathbb{G}_p$, recent literature focuses on restrictions induced by a directed graph $\mathcal{H} = (V, E_{\mathcal{H}}) \in \mathbb{D}_p$, where  $\mathbb{D}_p$ is the set of $p$-node directed graphs. The set of admissible DAGs is then,
\begin{align}
    \mathcal{G}_{\mathcal{H}} &:= \{G: G = (V, E_{G}) \in \mathbb{G}_p, E_{G} \subseteq E_{\mathcal{H}}\}.
\end{align}
This edge-based restriction is computationally advantageous because it allows the order score to remain decomposable over the nodes. Let $H$ be the corresponding adjacency matrix (where $H_{ji} = 1$ if $j \rightarrow i \in \mathcal{H}$) and $\mathbf{h}^i = \{X_j: H_{ji}=1\}$. The resulting restricted decomposable order score function is,
\begin{align}
    R_{\mathcal{H}}(\prec | D) &\propto \prod_{i=1}^p \sum_{\substack{\mathbf{Pa} \subseteq \mathbf{h}^i\\ \mathbf{Pa} \in \mathbf{Pa}_\prec(i)}} S(X_i, \mathbf{Pa} | D).
\end{align}
Typical order MCMC follows, just using $R_{\mathcal{H}}$ instead of $R$. Figure \ref{fig:example_res_space} outlines order MCMC procedure visually on a 4-node graph problem, where \ref{fig:example_res_space_a} shows the posterior probabilities for all 4-node DAGs, and some graphs are omitted by removing an edge from the space in \ref{fig:example_res_space_b}; the ensuing topological order posterior changes as a result, shown in \ref{fig:example_res_space_c}.

\begin{figure}[htb]
    \centering
    \begin{minipage}{0.52\textwidth}
        \centering
        \begin{subfigure}{\textwidth}
            \centering
            \includegraphics[width=1\linewidth]{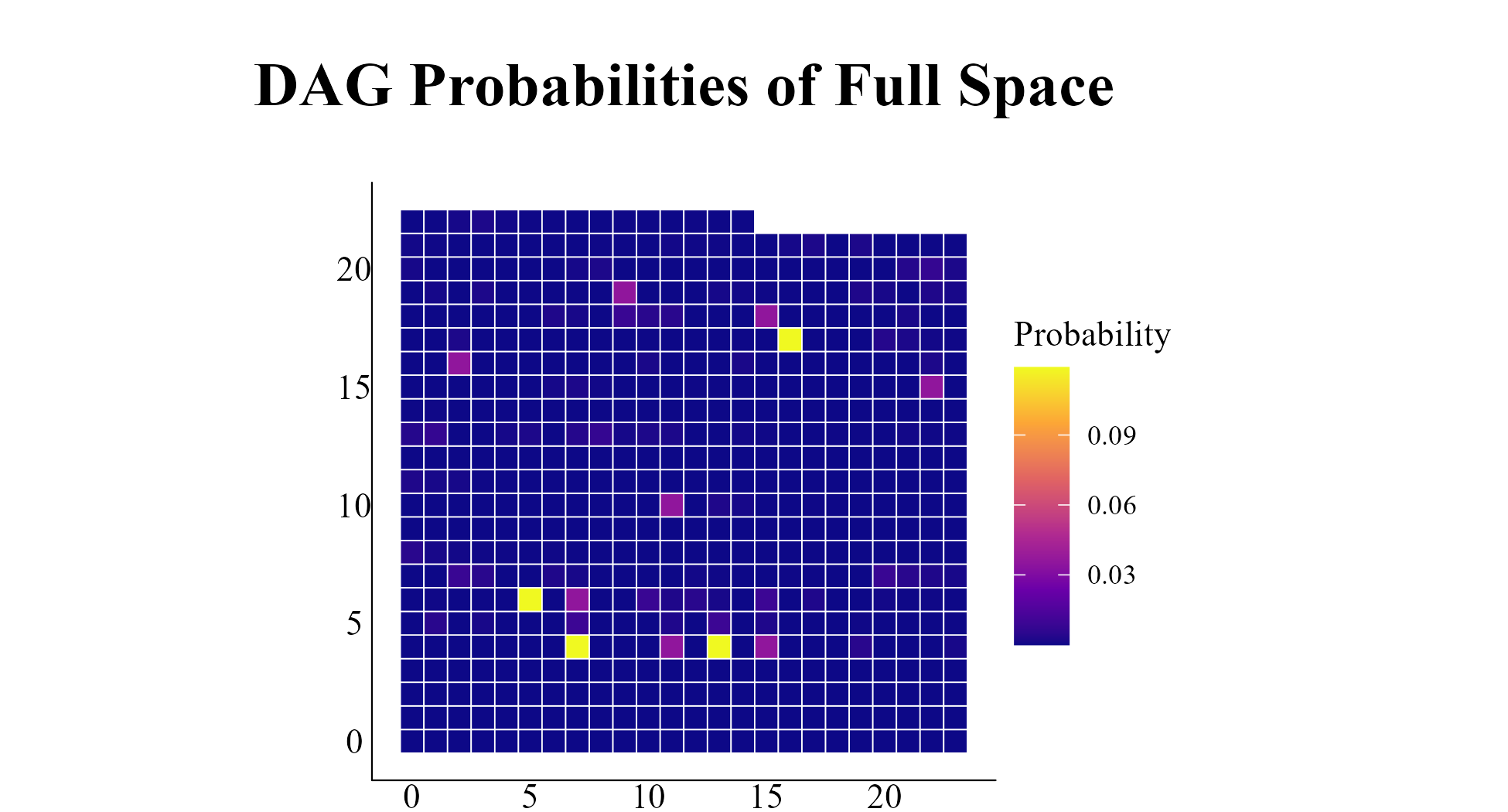}
            \caption{}
            \label{fig:example_res_space_a}
        \end{subfigure}
        
        \vspace{1cm} 
        
        \begin{subfigure}{\textwidth}
            \centering
            \includegraphics[width=1\linewidth]{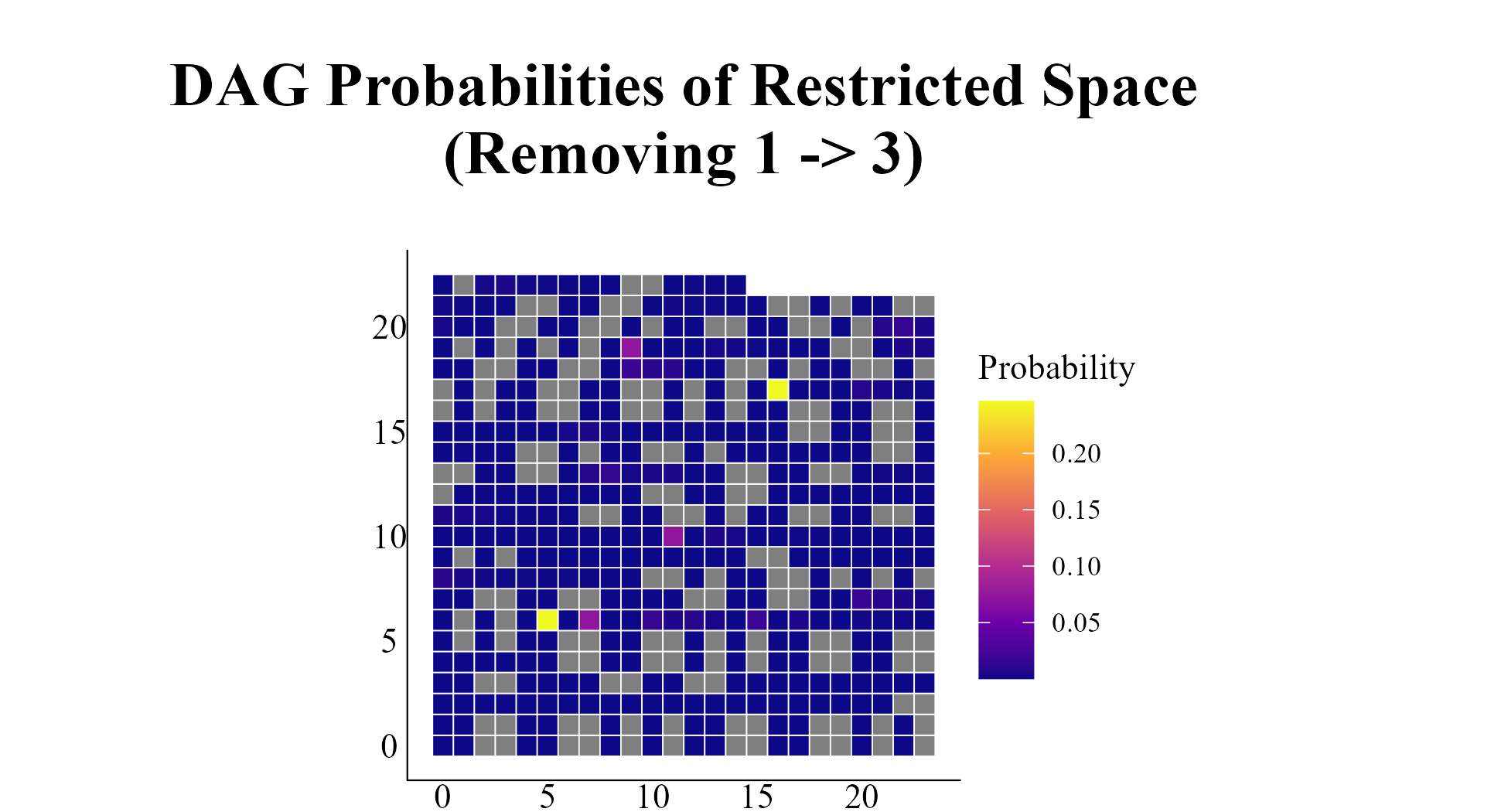}
            \caption{}
            \label{fig:example_res_space_b}
        \end{subfigure}
    \end{minipage}
    \hfill 
    \begin{minipage}{0.46\textwidth}
        \centering
        \begin{subfigure}{\textwidth}
            \centering
            \includegraphics[width=1\linewidth, height=9cm]{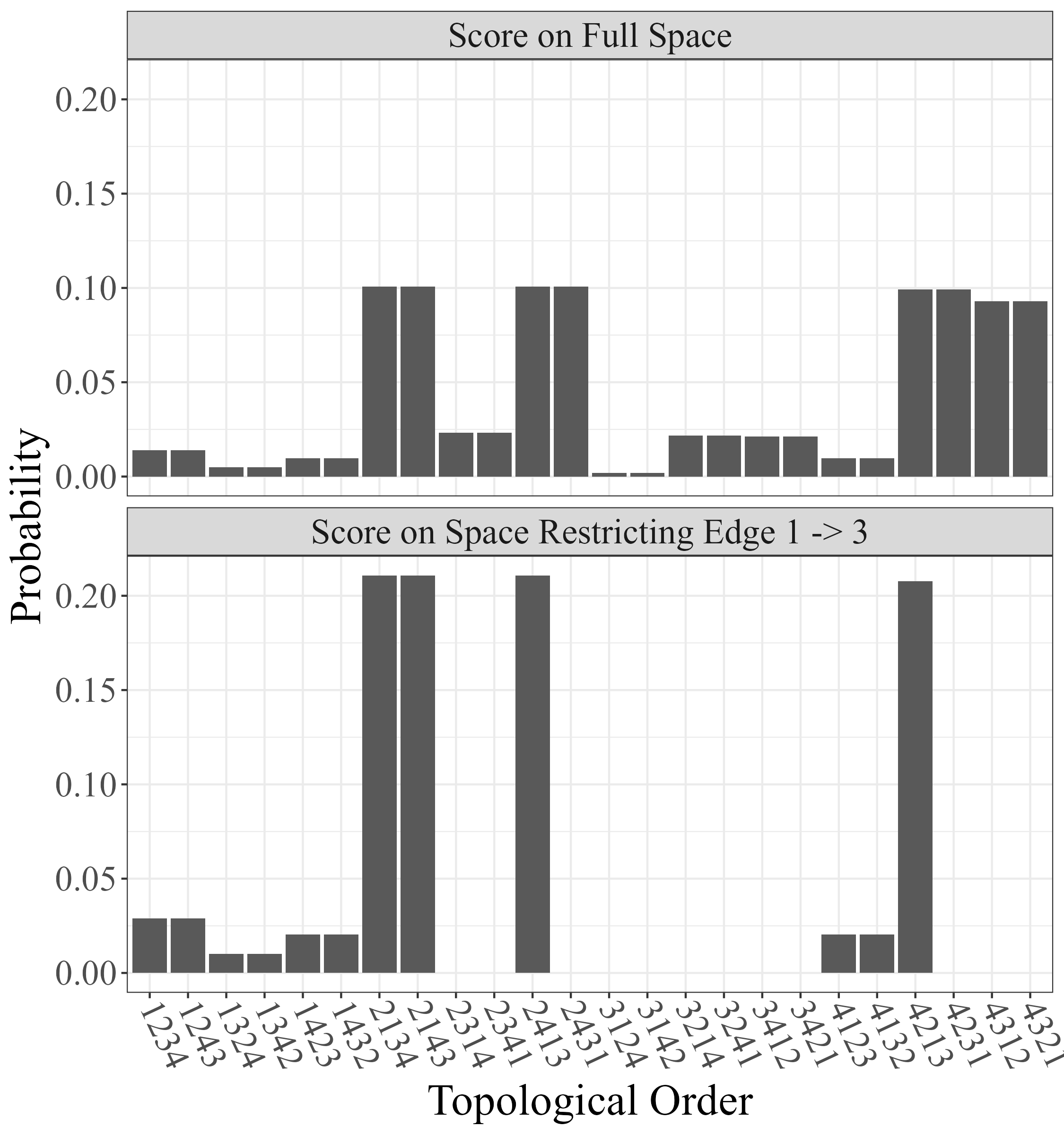}
            \caption{}
            \label{fig:example_res_space_c}
        \end{subfigure}
    \end{minipage}

    \caption{4-node example of restricted search space: (a) Posterior probabilities of the DAGs without imposing any restrictions. (b) Posterior probabilities of the graphs where the edge $1 \to 3$ is not included in the restricted search space. Excluded DAGs are shaded in gray. (c) Resulting restricted order posterior probability mass functions for both search spaces.}
    \label{fig:example_res_space}
\end{figure}

Constraining scores to be on the restricted space reduces the cost of MCMC steps from $O(p^K)$ to $O(Kp)$ (see \ref{app:score_table_bkgrd} for implementation details). However, this comes at the cost of omitting low-probability regions of the graph posterior in the space of order scores, meaning that $R_{\mathcal{H}}$ will not correspond to the same target as $R$, as shown in figure \ref{fig:example_res_space_c}. We note that, to our knowledge, there is a major gap in the literature concerning the error incurred when considering a restricted search space. In this paper we close this gap by establishing upper and lower bounds on the error, as reported in Section \ref{sec:error_bound}.

\subsubsection{transdimensional MCMC Setups}\label{subsubsec:transdim_mcmc}
Standard MCMC methods and their hybrid extensions fix the support throughout the chain to ensure ergodicity (and in the latter, restrict this fixed support to a manageable subset at the cost of exactness). In contract, \textit{transdimensional MCMC} methods are algorithms which permit transitions between states of differing dimensionality.

The reversible jump MCMC (RJMCMC) method \citep{green_reversible_1995} generalizes the MH algorithm to allow for dimension switches. This is done by introducing auxiliary variables to dimension-match the proposal with the current state of the chain, and including a Jacobian term in the acceptance probability to account for the change in dimension. \cite{giudici_decomposable_1999} extended this approach to Gaussian Graphical Models (GGMs). The two large drawbacks of RJMCMC include (1) the slow mixing of the chain when the acceptance probabilities are low, and (2) the large computational cost of calculating the Jacobians \citep{stephens_bayesian_2000}.

An alternate approach is to use a birth-death MCMC (BDMCMC) process which leverages a continuous-time Markov process (\cite{preston_spatial_1977}) rather than a discrete-time Markov chain. In BDMCMC, waiting times between jumps to higher dimensions (births) and lower dimensions (deaths) are taken to be random variables based on specific birth and death rates. The birth and death rates also define the transition kernel between spaces of different dimensions and are chosen such that the stationary distribution of the point process is the desired target. Extensions of the BDMCMC samplers for GGMs and general graphical models were given by \cite{mohammadi_bayesian_2015} and \cite{dobra_loglinear_2018}, respectively.

We will take inspiration from BDMCMC to propose a flexible hybrid transdimensional sampler in Section \ref{subsec:mixture_kernel}.

\section{Target Error Bounds in Restricted Search Space Methods}\label{sec:error_bound}

We first quantify the error introduced from the hybrid order MCMC sampler (as highlighted in figure \ref{fig:example_res_space}) that restricts the search space. These results confirm our intuition that the error introduced from hybrid order MCMC techniques could be quite large in some settings, and we characterize in what ways this could occur. We will use these results to motivate the design of our novel MCMC method in Section \ref{sec:brood_intro}.

For the rest of this paper, we consider order MCMC algorithms defined on restricted search spaces. For a fixed directed graph $\mathcal{H} \in \mathbb{D}_p$, the state space is 
$(\prec, G) \in \{(\prec, G): \prec \in \mathbb{S}^p, G \in \mathcal{G}_{\mathcal{H}} \cap \mathcal{G}_{\prec}\}$, with corresponding approximate order score $R_{\mathcal{H}}$. When $\mathcal{H}$ itself is updated during sampling, we operate on $(\mathcal{H}, \prec, G) \in \{(\mathcal{H}, \prec, G): \mathcal{H} \in \mathbb{D}_p, \prec \in \mathbb{S}^p, G \in \mathcal{G}_{\mathcal{H}} \cap \mathcal{G}_{\prec}\}$. 

For a fixed restricted space $\mathcal{H} \in \mathbb{D}_p$, we define the restricted order posterior $\pi_{Ord}^{(\mathcal{H})}$ as the distribution over orders induced by $\mathcal{G}_{\mathcal{H}}$. Additionally, for notational convenience, we write $\pi_{DAG}(\mathcal{H} | D)$ to denote the posterior mass assigned to DAGs in $\mathcal{G}_{\mathcal{H}}$, that is,
\begin{align*}
    \pi_{DAG}(\mathcal{H} | D) := \sum_{G \in \mathcal{G}_{\mathcal{H}}} \pi_{DAG}(G | D),
\end{align*}
with equivalent analogues for conditional or joint distributions involving $\mathcal{H}$. 

To our knowledge, restricted search spaces have only been applied to the graph space and not to higher-level spaces such as equivalence classes or orders, in the sense that restrictions occur by omitting specific graphs (and not specific equivalence classes or orders) from the space based on a subset of edges. So, our results have been derived for order scores with restrictions to the set of admissible graphs. Note that since the empty graph is consistent with all possible orders, such restrictions do not change the support of order MCMC, but they do alter the \textit{effective} support by changing the corresponding score and induced posterior distribution.

\subsection{Lower and Upper Bound}

The following theorem gives us upper and lower bounds on the error introduced in the restricted search space MCMC methods:

\begin{theorem}\label{theorem:tv_bound}
Let $\mathcal{H} \in \mathbb{D}_p$ be a directed graph with corresponding edge set $E_{\mathcal{H}}$ and admissible DAG set $\mathcal{G}_\mathcal{H}$. Let $\pi_{DAG}, \pi_{Ord}$ denote the posterior distributions over DAGs and orders, respectively, using the full DAG space $\mathbb{G}_p$, and let $(R_{\mathcal{H}}, \pi_{Ord}^{(\mathcal{H})})$ denote the perturbed score-distribution pair, where $\pi_{Ord}^{(\mathcal{H})}$ is the posterior over orders restricted to $\mathcal{G}_\mathcal{H}$. Define $\varepsilon_{\mathcal{H}} := 1-\pi_{DAG}(\mathcal{H} | D)$ as the posterior mass assigned to DAGs outside of $\mathcal{H}$. 

Then, for some constant $c \in [0,1]$,
\begin{align}
    1-\sqrt{1-c\varepsilon_{\mathcal{H}}} &\leq D_{TV}(\pi_{Ord}^{(\mathcal{H})}, \pi_{Ord}) \leq \min\left(\sqrt{2-2\sqrt{1-c\varepsilon_{\mathcal{H}}}}, 1\right)\label{eq:lower_upper_bound}
\end{align}

where $D_{\mathrm{TV}}$ denotes the total variation distance between distributions.
\end{theorem}

\begin{proof}
    See Section \ref{app:tv_bound}.
\end{proof}

We plot the lower and upper bounds from equation \eqref{eq:lower_upper_bound} for when $c=0.1, 0.5, 0.9$ in Figure \ref{fig:tv_bound}, which illustrates that for large $c$, the target will be far away from the posterior when omitting even a small portion of it (we describe the intuition of $c$ in \ref{subsubsec:understand_c}). A consequence is that if the bounds are large for a restricted search space, no Markov chain operating on it, no matter how efficient, could satisfy \textit{uniform ergodicity} \citep{roberts_general_2004}, a consequential property of an MCMC chain's convergence to a stationary distribution, $\pi$. 

We show how theorem \ref{theorem:tv_bound} connects to convergence analysis in the following examples:

\begin{remark}
    If the prior over graphs $\pi_{DAG}(G)$ is supported entirely over $\mathcal{H}$, i.e., $\pi_{DAG}(G) = 0$ if $G \notin \mathcal{G}_\mathcal{H}$, then the posterior over orders restricted to $\mathcal{H}$ coincides with the full posterior over orders: $\pi_{Ord}^{(\mathcal{H})}(\prec | D) = \pi_{Ord}(\prec | D)$ for all $\prec \in \mathbb{S}^p$, and hence $D_{TV}(\pi_{Ord}^{(\mathcal{H})}(\cdot | D), \pi_{Ord}(\cdot | D)) = 0$ in Theorem \ref{theorem:tv_bound}.
\end{remark}

\begin{remark}\label{rem:strong_consistency}
    $D_{TV}(\pi_{Ord}^{(\mathcal{H})}(\cdot | D), \pi_{Ord}(\cdot | D))$ can converge to 0 in Theorem \ref{theorem:tv_bound} under strongly consistent posteriors. A graph posterior distribution $\pi_{DAG}$ is strongly consistent when
    \begin{align*}
        \lim_{n \to \infty} Pr(|\pi_{DAG}(G_{\text{true}} | D) - 1 | \geq \epsilon) &= 0 \quad \text{ for all } \epsilon > 0,
    \end{align*}
    where $G_{\text{true}}$ is the true underlying graph and $n$ is the number of observations in $D$. If $\pi_{DAG}$ is strongly consistent and if $\mathcal{H}$ is estimated such that $\lim_{n \to \infty} Pr(\hat{G}_{MAP_n} \in \mathcal{G}_{\hat{\mathcal{H}}_n}) = 1$, then $D_{TV}(\pi_{Ord}^{(\mathcal{H})}, \pi_{Ord})$ converges to 0. However, posteriors that have been shown to be strongly consistent in previous research assume a known topological ordering (\cite{cao_posterior_2019}, \cite{lee_minimax_2019}), so results on strongly consistent graph posteriors need to be extended to order scores.
\end{remark}

\begin{remark}\label{rem:noisy_mcmc}
    There is a growing body of literature on \textit{noisy MCMC} \citep{alquier_noisy_2016} in which the transition kernel, $P$, of a Markov chain that has the desired stationary distribution $\pi_P$ but is too computationally expensive to run is replaced with a ``perturbed'' kernel $\tilde{P}$ that is computationally cheaper to run. The modified transition kernel defines a chain that has a different stationary distribution $\pi_{\tilde {P}}$. The noisy MCMC methodology is needed to characterize conditions under which the difference between $\pi_P$ and $\pi_{\tilde{P}}$  is not too large. A common setup is when $P$ satisfies a \textit{contraction} property, that is
    \begin{align*}
        D_{TV}(P(x, A), P(y, A)) \leq (1-\alpha)\mathbbm{1}[x \neq y] \text{     } \forall x,y \in \Omega,\forall A \in \sigma(\Omega), \text{ and } 0 < \alpha \leq 1
    \end{align*}
    If the perturbed kernel $\tilde{P}$ satisfies
    \begin{align*}
        D_{TV}(P(x, A), \tilde{P}(x, A)) \leq \epsilon \text{     } \forall x \in \Omega, \forall A \in \sigma(\Omega) \text{ for some } 0 \leq \epsilon < \infty
    \end{align*}
    then $D_{TV}(\pi_P, \pi_{\tilde{P}}) \leq \frac{\epsilon}{\alpha}$ \citep{rudolf_perturbations_2024}. 
    
    In contrast, Theorem \ref{theorem:tv_bound} directly bounds the total variation distance between the targets themselves (rather than between the stationary distributions of (perturbed) Markov kernels). When Markov chains targeting $\pi_{Ord}$ and $\pi_{Ord}^{(\mathcal{H})}$ exist and are ergodic, it suggests that for any pair of Markov chains $(P, \tilde{P})$ where $P$ and $\tilde{P}$ explore the full and restricted space, respectively, the inequality
    
    \begin{align*}
        1-\sqrt{1-c\varepsilon_{\mathcal{H}}} \leq \frac{\epsilon}{\alpha}
    \end{align*} 
    
provides a practical lower bound on the gap between $\pi_P$ and $\pi_{\tilde{P}}$ (even when we cannot neatly evaluate $\epsilon$). This highlights how the quality of $\mathcal{H}$ directly influences the accuracy of any approximate kernel (as expressed through $\epsilon$) for order MCMC, and formalizes our intuition on the practical limitations of using a restricted support set.
\end{remark}

Even when there are theoretical guarantees for convergence, issues may arise in practice. \cite{viinikka_towards_2020} analyzed the restricted search space problem empirically and found that, even when including the most probable 60\% of possible parents for each node, the mean probability of including the true parents per node was only 90\% for a data set of size $n=200$ for a GDM with $p=20$ nodes. $1-\varepsilon_{\mathcal{H}}$ is approximately (though not identical) to $90\%^{20} = 12\%$, suggesting that even including a large portion of possible nodes in a restricted search space when $p$ is small could lead to large $\varepsilon_{\mathcal{H}}$. Perhaps even more importantly, it suggests a curse of dimensionality: even a slight error introduced per node yields large implications when increasing $p$.

\begin{figure}
    \centering
    \includegraphics[width=.95\linewidth]{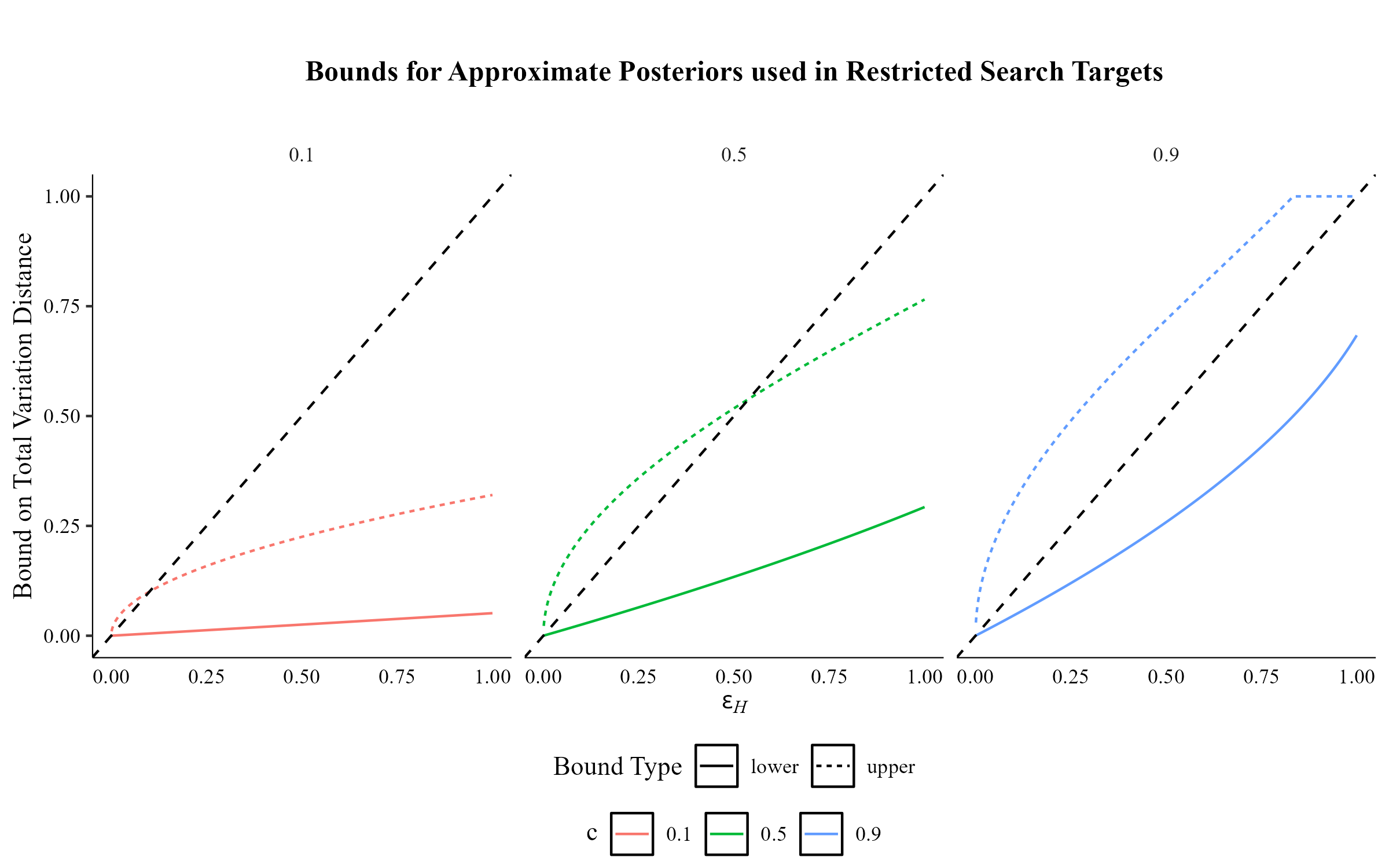}
    \caption{
    Upper and Lower bound of the Total Variation Distance for different thresholds across 3 different constants. Dashed line denotes $y=x$.}
    \label{fig:tv_bound}
\end{figure}

\subsubsection{Understanding \texorpdfstring{$c$}{TEXT}}\label{subsubsec:understand_c}
In Theorem \ref{theorem:tv_bound}, the lower and upper bounds rely on 2 constants $\varepsilon_{\mathcal{H}}$ and $c$. $\varepsilon_{\mathcal{H}}$ is intuitive to understand and is a potential link between consistency and Theorem \ref{theorem:tv_bound}. Lemmas \ref{theorem:min_c}-\ref{theorem:max_c} provide some insight into $c$. For these, for a restricted representation directed graph $\mathcal{H}\in\mathbb{D}_p$ with corresponding DAG set $\mathcal{G}_{\mathcal{H}} \subseteq \mathbb{G}_p$, we define the complementary DAG set as 
\begin{align*}
    \mathcal{G}_{\neg\mathcal{H}} := \mathbb{G}_p \setminus \mathcal{G}_{\mathcal{H}}.
\end{align*}
Here $\neg\mathcal{H}$ serves only as a symbolic index for this complementary DAG set (so $\neg\mathcal{H}$ does not correspond to a single element of $\mathbb{D}_p$). 
We denote the posterior distribution over orders restricted to this omitted set as $\pi_{Ord}^{(\neg\mathcal{H})}$.

\begin{lemma}\label{theorem:min_c}
   Let $\mathcal{H}\in\mathbb{D}_p$ be a directed-graph representation with associated DAG set $\mathcal{G}_{\mathcal{H}}\subseteq\mathbb{G}_p$, and let $\mathcal{G}_{\neg\mathcal{H}}=\mathbb{G}_p\setminus\mathcal{G}_{\mathcal{H}}$. Define $\pi_{Ord}^{(\mathcal{H})}, \pi_{Ord}^{(\neg\mathcal{H})}$ as the posterior distributions over orders by restricting the full posterior to $\mathcal{G}_{\mathcal{H}}$ and $\mathcal{G}_{\neg\mathcal{H}}$, respectively. Then $c$ assumes the minimal value $c=0$ if and only if $\pi_{Ord}^{(\mathcal{H})}(\prec | D)= \pi_{Ord}^{(\neg \mathcal{H})}(\prec | D)$ $\forall \prec \in \mathbb{S}^p$.
\end{lemma}
\begin{proof}
    See Section \ref{app:min_c}.
\end{proof}

\begin{lemma}\label{theorem:max_c}
    Let $\mathcal{H}$, $\neg\mathcal{H}$, $\mathcal{G}_{\mathcal{H}}$,
    $\mathcal{G}_{\neg \mathcal{H}}$, $\pi_{Ord}^{(\mathcal{H})}, \pi_{Ord}^{(\neg\mathcal{H})}$ be as in lemma \ref{theorem:min_c}. Then $c$ is maximal (that is, $c=1$) if and only if the probability mass functions $\pi_{Ord}^{(\mathcal{H})}(\cdot | D)$ and $\pi_{Ord}^{(\neg \mathcal{H})}(\cdot | D)$ have disjoint supports.
\end{lemma}
\begin{proof}
    See Section \ref{app:max_c}.
\end{proof}

The above lemmas provide intuition that $c$ is related to the degree of disagreement between the restricted posterior and the posterior restricted to the region in the sample space that is being omitted. In practice, it is difficult to assess the degree of disagreement between $\pi_{Ord}^{(\mathcal{H})}$ and $\pi_{Ord}^{(\neg\mathcal{H})}$, as the construction omits $\pi_{Ord}^{(\neg\mathcal{H})}$. This implies that if it is not guaranteed or highly probable that $\lim_{n \to \infty}\varepsilon_{\mathcal{H}_n} = 0$ when constructing $\mathcal{H}_n$, the lower bound for noisy MCMC can be hard to quantify and can be large. This insight motivates the development of approaches that do not use a fixed restricted search space.

\subsection{Adaptive MCMC Approaches for Reducing Target Error and Prior Considerations in a Hierarchical Setup}
The error produced by hybrid algorithms for DAG sampling arises because the restricted search space $\mathcal{H}$ is fixed, and the chain samples from a restricted target $\pi_{Ord}^{(\mathcal{H})}$. An alternate strategy is to adapt $\mathcal{H}$ as the Markov chain progresses, so that the chain does not exclusively sample from a single restricted target with a fixed error but a mixture that could collectively lower the error.

To formalize this, we can treat the restricted search space as a random variable $\bm{H}$ taking values in $\mathbb{D}_p$. The following hierarchical model then characterizes hybrid algorithms with adaptation of the search space:
\begin{align}
    \bm{H} &\sim P_{\bm{H}}(\cdot)\label{eq:nested_model_h}\\
    \bm{G} &\sim P_{\bm{G} \mid \bm{H}}(\cdot \mid \mathcal{H})\label{eq:nested_model_graph}\\
    \bm{X}_1,...,\bm{X}_n &\overset{\text{iid}}{\sim} P_{\bm{X}\mid\Theta_{G}}(\cdot \mid \theta_{G})\label{eq:nested_model_data}
\end{align}
In this framework, standard hybrid order MCMC corresponds to choosing a point mass prior $P_{\bm{H}} = \delta_{\mathcal{H}^*}$, where $\mathcal{H}^*$ is the chosen restricted search space; that is, $P_{\bm{H}}(\mathcal{H}^*) = 1$ and $P_{\bm{H}}(\mathcal{H}') = 0$ for all $\mathcal{H}' \neq \mathcal{H}^*$. Sampling from the full space likewise corresponds to $P_{\bm{H}} = \delta_{K_p}$, where $K_p$ is the complete graph (i.e., its adjacency matrix $A$ satisfies $A_{ij}=1$ when $j \neq i$, and $A_{ij}=0$ when $j=i$). 

The original order MCMC paper \citep{friedman_being_2003} effectively uses another degenerate prior, but of a different kind: instead of restricting the allowable edges through a graph $\mathcal{H}$, it restricts the allowable \emph{parent sets} by imposing a fixed maximal in-degree constraint. This constraint defines a subset of $\mathbb{G}_p$ that cannot in general be represented by a single directed graph $\mathcal{H}$. However, we can extend the hierarchical formulation to encapsulate this case as well by including a parameter $\bm{d}$ that represents the degree. Then, the priors on $\bm{d}$ and $\bm{H}$ are also degenerate, where $P_{\bm{H}}(\mathcal{H}) = \delta_{K_p}$, and $P_{\bm{d}}(d) = \delta_{d^*}$, where $d^* \leq p$ is the prespecified degree constraint.

We instead propose a new class of priors that move beyond this fixed support, allowing for the potential to reduce fixed search space error. To accommodate flexible priors that allow for the restricted search space and the corresponding scores to change during sampling, we turn to transdimensional MCMC methods. Specifically, we propose to extend BDMCMC to $(\mathcal{H}, \prec, G)$ to reduce the error incurred when using a restricted order space. This proves to be difficult given that orders can map to multiple graphs and that graphs can map to multiple orders. To avoid this issue, in Section \ref{sec:brood_intro}, we propose building a Markov chain that uses a mixture transition kernel: one component samples $\mathcal{H} \mid \prec$ via a birth-death transition kernel and one component updates $\prec \mid \mathcal{H}$ according to the restricted search order transition kernel. The graphs can then be sampled directly from $G \mid \prec, \mathcal{H}$. This approach is validated using the methodology developed by \cite{geyer_simulation_1994}.

\section{Birth-Death Restricted Search MCMC}\label{sec:brood_intro}

We aim to enable transitions between restricted search spaces of different sizes, thereby limiting the fixed-support error. This requires a class of priors on $\bm{H}$ that (1) are non-degenerate and (2) allow for computationally feasible sampling of the posterior. We draw on transdimensional MCMC literature to inform this class, specifically the balance conditions underlying birth-death MCMC, and a related mixture kernel method.

\subsection{Birth-Death Detailed Balance and transdimensional Target Recovery}\label{subsec:bdmcmc}
To create a BDMCMC algorithm for restricted search space MCMC, we must define the birth and death rates. The standard BDMCMC literature defines birth and death rates based on satisfying detailed balance conditions on a transdimensional posterior measure $\mu$ (\cite{preston_spatial_1977}, \cite{dobra_loglinear_2018}). For a birth-death process to have stationary probability $\mu$ on measurable space $(\Omega, \mathcal{F})$,
\begin{align}
    \int_U \beta(z)d\mu_m(z) &= \int \delta(y) K_{\delta}^{(m+1)}(y, U) d\mu_{m+1}(y) \text{ for } m \geq 0, U \in \mathcal{F}_m\label{eq:det_bal_1} \\
    \int_V \delta(z) d\mu_{m+1}(z) &= \int \beta(x) K_{\beta}^{(m)}(x, V) d \mu_m(x) \text{ for } m \geq 0, V \in \mathcal{F}_{m+1}\label{eq:det_bal_2}
\end{align}
where $\Omega_m$, $\mathcal{F}_m$ denote the support and possible states of the process with dimension $m$, respectively, $\mu_m$ is $\mu$ restricted to $(\Omega_m, \mathcal{F}_m)$, $\beta, \delta$ are the birth and death rates and $K_\beta^m:\Omega_{m}\times \mathcal{F}_{m+1} \to \mathbb{R}^{+}, K_\delta^m:\Omega_{m} \times \mathcal{F}_{m-1} \to \mathbb{R}^{+}$ are transition kernels describing the jump processes from dimension $m$. To sample from $\mu$ by BDMCMC, a transition is proposed by adding or removing a point from the space according to the transition kernels defined by $\beta, \delta$ that satisfy \eqref{eq:det_bal_1} and \eqref{eq:det_bal_2}. All proposals are accepted, but each sample is weighted by the \textit{waiting time}, $\frac{1}{\beta(x) + \delta(x)}$.

Alternatively, the problem can be flipped: choose birth and death rates $\beta, \delta$ that define a valid birth-death process with a unique solution. Then, the limit of this jump process is $\mu$, the probability measure that satisfies \eqref{eq:det_bal_1} and \eqref{eq:det_bal_2}. \cite{preston_spatial_1977} found conditions on $\beta, \delta$ that allow $\mu$ to be recovered through a recursive formula. We use this approach to define a class of transdimensional distributions on $(\bm{H}, \bm{O})$ that can be sampled by a birth-death process for restricted search MCMC:

\begin{theorem}\label{theorem:distribution_class}
    Let $\mathcal{H} \in \mathbb{D}_p$, and $\prec \in \mathbb{S}^p$. Define $B_e(\mathcal{H}, \prec)$ and $D_e(\mathcal{H}, \prec)$ as follows:
    
    \begin{align}
        B_e(\mathcal{H}, \prec) &\propto \frac{\frac{1}{|\mathcal{G}_{\mathcal{H}^{+e}} |}\pi_{DAG}(\mathcal{H}^{+e} \mid \prec, D)}{\frac{1}{|\mathcal{G}_{\mathcal{H}}|} \pi_{DAG}(\mathcal{H} \mid \prec, D)} = \frac{1}{2}\frac{\pi_{DAG}(\mathcal{H}^{+e} \mid \prec, D)}{\pi_{DAG}(\mathcal{H} \mid \prec,D)}\label{eq:birth_rate}  \\
        D_e(\mathcal{H}, \prec) &\propto \frac{\frac{c^*}{|\mathcal{G}_{\mathcal{H}^{-e}}|}\pi_{DAG}(\mathcal{H}^{-e}\mid \prec, D)}{\frac{1}{|\mathcal{G}_{\mathcal{H}}|} \pi_{DAG}(\mathcal{H} \mid \prec, D)}=2c^*\frac{ \pi_{DAG}(\mathcal{H}^{-e} \mid \prec, D)}{\pi_{DAG}(\mathcal{H} \mid \prec, D)}\label{eq:death_rate}
    \end{align}
    where $c^* \in (0, 1]$. Let $\beta(\mathcal{H}, \prec) = \sum_{e \notin E_\mathcal{H}} B_e(\mathcal{H}, \prec)$, and $\delta(\mathcal{H}, \prec) = \sum_{e \in E_\mathcal{H}} D_e(\mathcal{H}, \prec)$. Then, a birth-death process defined by $\beta, \delta$ produces a unique solution whose limit is a tractable probability distribution.
\end{theorem}
\begin{proof}
    See Section \ref{app:distribution_class}.
\end{proof}

These rates are similar to those of \cite{mohammadi_bayesian_2015} and \cite{dobra_loglinear_2018}. Since those BDMCMC samplers operate directly on the graph space, the rates are defined by ratios of graph probabilities; the rates we define in Theorem \ref{theorem:distribution_class} are ratios of mean graph probability within the search spaces (conditional on the topological order). This elicits the following class of search space posteriors:

\begin{corollary}\label{theorem:posterior_form}
    Construct a birth-death Markov process $(\bm{H}, \bm{O}) \in \mathbb{D}_p \times \mathbb{S}^p$ with the birth and death rates specified in Theorem \ref{theorem:distribution_class}. Then, the marginal of $\mathcal{H} \in \mathbb{D}_p$ of the stationary distribution admits the following form:
    \begin{align}
        \mu_{Space}^{(c^*)}(\mathcal{H}) \propto (2c^{*})^{-|E_{\mathcal{H}}|}\sum_{\prec \in \mathbb{S}^p} \pi_{DAG}(\mathcal{H} | \prec, D).\label{eq:posterior_form}
    \end{align}
\end{corollary}
\begin{proof}
    See Section \ref{app:posterior_form}.
\end{proof}

The class of posteriors defined by the rates in Theorem \ref{theorem:distribution_class} induces a class of priors that assign mass to more than one value in the set of search spaces. When combined with the restricted order MCMC, this can allow for sampling on $(\bm{H}, \bm{O})$ without assuming a fixed $\mathcal{H}$.

\subsection{Mixture Kernels in Spatial Point Processes}\label{subsec:mixture_kernel}

\cite{geyer_simulation_1994} developed a framework that is a generalization of the BDMCMC sampler. They consider a Poisson process $\nu$ on $(\Omega, \mathcal{F})$, with the intensity measure $\lambda$ on $(\Lambda, \mathcal{B})$, and want to sample realizations of a point process with distribution $\pi$ on $(\Omega, \mathcal{F})$, and density $g = d\pi/d\nu$ with respect to the Poisson process. To this end, they define the following MH algorithm that is a mixture of 2 transition kernels:
\begin{align}
    Q_\ell(F \mid x) &= (1-\ell)Q_0(F \mid x) + \ell Q_1(F \mid x) \quad 0 \leq \ell \leq 1\label{eq:mixture_kern}
\end{align}
where $Q_0$ is a standard MH procedure concentrated on $H_m = \Omega_m \cap H$, where $H=\{g > 0\}$, and $Q_1$ is concentrated on $H_{m-1} \cup H_m \cup H_{m+1}$. $Q_1$ is constructed such that with probability $q(x)$, a new point $\xi$ is generated from some random density $\mathsf{b}(x, \xi)$ with respect to $\lambda(d\xi)$, and with probability $1-q(x)$, a random point $\eta \in x$ is removed with some probability $\mathsf{D}(x/\eta; \eta)$ (or if $m=0$, we do nothing). For $x \in H_m$, $Q_1$ is defined as:
\begin{align*}
    Q_1(F_{m+1} \mid x) &= q(x) \int_{x \cup \xi \in F_{m+1}} \mathsf{b}(x, \xi) A_1(x \cup \xi \mid x)\lambda(d\xi)\\
    Q_1(F_m \mid x) &= \mathbbm{1}[x \in F_m]\Big\{q(x)\int \mathsf{b}(x, \xi) [1 - A_1(x \cup \xi \mid x)]\lambda(d\xi)\\
    &\quad + (1-q(x))\sum_{\eta \in x} \mathsf{D}(x / \eta; \eta)[1-A_1(x / \eta \mid x)]\Big\}\\
    Q_1(F_{m-1} \mid x) &= (1-q(x))\sum_{\eta \in x} \mathbbm{1}[x / \eta \in F_{m-1}]\mathsf{D}(x / \eta; \eta)A_1(x / \eta \mid x)
\end{align*}
where $A_1$ is the \textit{acceptance rate} of the transition kernel. Time reversibility is ensured with the following acceptance rates.
\begin{align*}
    A_1(x \mid x \cup \xi) &= \begin{cases}
        \min\{1, 1/r(x, \xi)\} & \text{if } x \cup \xi \in H,\\
        0 & \text{otherwise}
    \end{cases}\\
    A_1(x \cup \xi \mid x) &= \begin{cases}
        \min\{1, r(x, \xi)\} & \text{if } x \cup \xi \in H,\\
        0 & \text{otherwise}
    \end{cases}
\end{align*}
where
\begin{align}
    r(x, \xi) &= \frac{g(x \cup \xi)}{g(x)}\frac{1-q(x \cup \xi)}{q(x)}\frac{\mathsf{D}(x; \xi)}{\mathsf{b}(x, \xi)} \quad \text{if } x \cup \xi \in H \label{eq:acceptance_ratio}
\end{align}

While the continuous-time birth-death process and the discrete-time mixture kernel framework of \cite{geyer_simulation_1994} are often treated as distinct approaches to transdimensional sampling, they can be unified. We formalize this connection in Lemma \ref{theorem:mixture_kernel_bdmcmc}, which demonstrates that the Metropolis-Hastings acceptance ratio for a Metropolized jump process is exactly determined by the ratio of the BDMCMC waiting times. This result ensures that a posterior that is compatible with a birth-death process can be sampled using a mixture kernel process (and thus the birth-death rates defined in Theorem \ref{theorem:distribution_class} can be seamlessly integrated into a mixture kernel MCMC architecture).

\begin{lemma}\label{theorem:mixture_kernel_bdmcmc}
    Let $\mu$ be a transdimensional probability measure that is compatible with a birth-death process defined by $\beta, \delta$. Then $r(x, \xi)$, the acceptance ratio in the mixture kernel sampler of \cite{geyer_simulation_1994}, is given by the ratio of waiting times $w(x \cup \xi), w(x)$ of the jump process:
    \begin{align*}
        r(x, \xi) &= \frac{\frac{1}{\beta(x \cup \xi) + \delta(x \cup \xi)}}{\frac{1}{\beta(x) + \delta(x)}} = \frac{w(x \cup \xi)}{w(x)}
    \end{align*}
\end{lemma}
\begin{proof}
    See Section \ref{app:mixture_kernel_bdmcmc}.
\end{proof}

\subsection{BROOD Algorithm}

The lemma \ref{theorem:mixture_kernel_bdmcmc} signifies that the limiting distribution of a standard birth-death sampler can equivalently be obtained from a mixture of MH transition kernels, where the acceptance ratio of $Q_1$ is a waiting time ratio. This equivalence is quite useful for sampling on $(\bm{H}, \bm{O})$. Specifically, the computationally efficient hybrid order MCMC sampler (targeting $\prec \mid \mathcal{H}$) serves as the kernel $Q_0$, while $Q_1$ facilitates the birth-death transitions between search spaces (targeting $\mathcal{H}$). Since Theorem \ref{theorem:distribution_class} defines a class of distributions that can be sampled by birth-death processes, consequentially, Theorem \ref{theorem:distribution_class} and Lemma \ref{theorem:mixture_kernel_bdmcmc} together demonstrate that hybrid order MCMC can be extended into a rigorous framework for sampling from a flexible, non-degenerate class of transdimensional distributions on $(\bm{H}, \bm{O})$.

We implement our novel dynamic restricted order MCMC sampler, \textbf{B}irth-death processes \textbf{R}estricted \textbf{O}ver \textbf{O}rder \textbf{D}istributions (BROOD), using the mixture MH method. Compared to the previous hybrid order MCMC samplers, BROOD's two methodological innovations are to (1) allow for not only expansions to the search space, but contractions too, and (2) dynamically change the search space at a rate $\ell \in [0, 1]$ throughout the simulation, instead of doing it for a finite number of iterations. We outline BROOD in pseudocode form in algorithm \ref{algo:brood}, and illustrate it visually in Figure \ref{fig:brood}.

\begin{algorithm}[htbp]
    \SetAlgoLined
    \KwIn{Current order $\prec_t$; current restricted search space $\mathcal{H}_t$; restricted target score and probability distribution $(R_{\mathcal{H}_t}(\cdot | D), \pi_{Ord}^{(\mathcal{H}_t)}(\cdot | D))$;}
    1. Sample whether to use $Q_0$ or $Q_1$ with probability $(1-\ell, \ell)$\;
    2. \eIf{$Q_0$}{
    a. $\mathcal{H}_{t+1} \gets \mathcal{H}_t$\;
    b. Update $\prec_{t+1}$ via $\pi_{Ord}^{(\mathcal{H}_{t+1})}(\prec | D)$, with M-H proposals from $\prec_t$\;
    }{
    a. For $e \in E_\mathcal{H}$, $e' \notin E_\mathcal{H}$, calculate the birth and death rates, $B_{e'}(\mathcal{H}_t, \prec_t)$, $D_e(\mathcal{H}_t, \prec_t)$\;
    b. Calculate the waiting time $w_t = \frac{1}{\sum B_{e'}(\mathcal{H}_t, \prec_t)+ \sum D_{e}(\mathcal{H}_t, \prec_t)}$\;
    c. Sample whether to propose expansion or contraction of the search space: \\
    $\text{adaptation type} \sim \begin{cases}
        \text{expansion} & \text{ with probability } w_t\sum B_{e'}(\mathcal{H}_t, \prec_t)\\
        \text{contraction} & \text{ with probability } w_t \sum D_{e}(\mathcal{H}_t, \prec_t)
    \end{cases}$\;
    d. \eIf{\text{adaptation type }==\text{expansion}}{
        $\mathcal{H}'$ gets updated proportionally to $B_{e'}(\mathcal{H}_t, \prec_t)$ for $e' \notin E_\mathcal{H}$\;
    }{
        $\mathcal{H}'$ gets updated proportionally to $D_{e}(\mathcal{H}_t, \prec_t)$ for $e \in E_\mathcal{H}$\;
    }
    e. Calculate birth and death rates of $(\mathcal{H}', \prec_t)$, and corresponding waiting time $w'$\;
    f. $A_1((\mathcal{H}', \prec_t)) \gets \frac{w'}{w_t}$\;
    g. Sample $u \sim \text{Uniform(0, 1)}$\;
    h. \eIf{$u < A_1((\mathcal{H}', \prec_t))$}{
    $\mathcal{H}_{t+1} \gets \mathcal{H}'$\;
    }{
    $\mathcal{H}_{t+1} \gets \mathcal{H}_t$\;
    }
    i. $\prec_{t+1} \gets \prec_t$
    }
    3. If DAG sampling as well, sample $G_{t+1} \sim \pi_{DAG}^{(\mathcal{H}_{t+1})}(G | \prec_{t+1}, D)$\;
    \KwOut{$\prec_{t+1}$; $\mathcal{H}_{t+1}$; $(R_{\mathcal{H}_{t+1}}(\cdot | D), \pi_{Ord}^{\mathcal{H}_{t+1}}(\cdot | D))$;}
\caption{BROOD Algorithm at step $t$}\label{algo:brood}
\end{algorithm}

\begin{figure}
    \centering
    \includegraphics[width=0.99\linewidth]{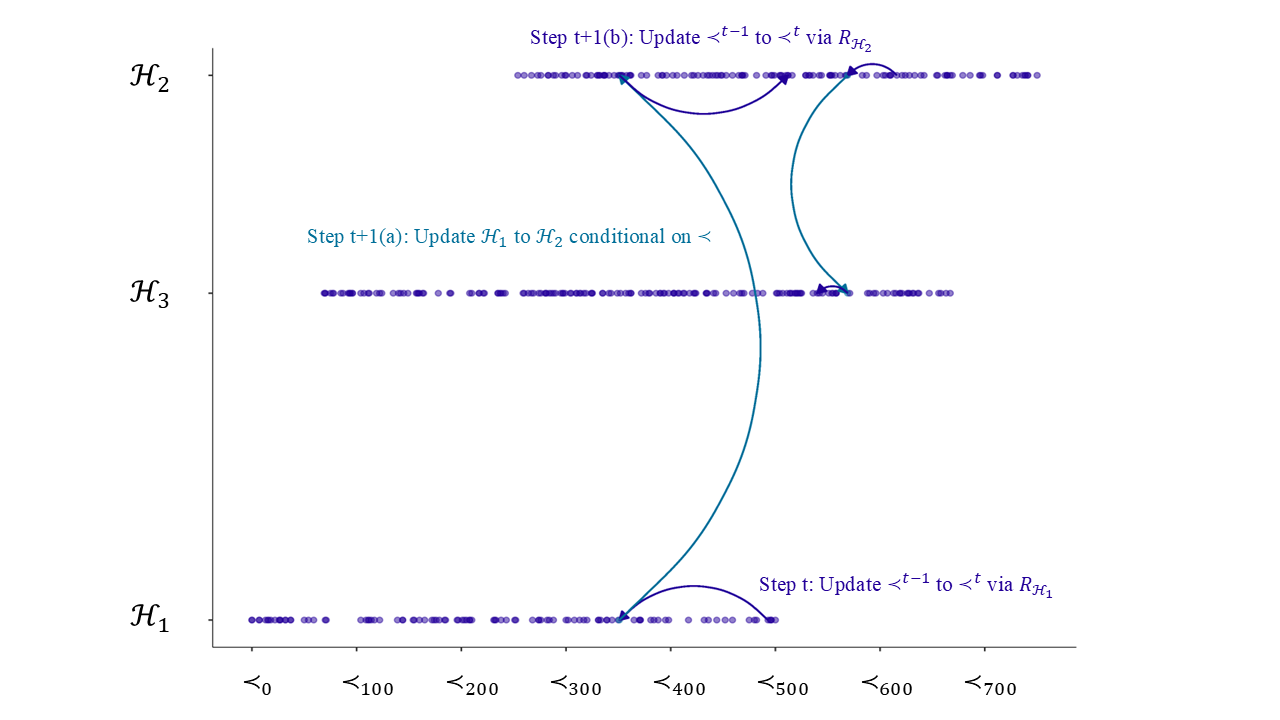}
    \caption{Illustration at step $t$ of combined BROOD Algorithm. Purple dots and arrows correspond to the $Q_0$ kernel in equation \eqref{eq:mixture_kern} (moves across the orders space $\prec$), and blue dots and arrows correspond to the $Q_1$ kernel (moves across the restricted search spaces $\mathcal{H}$)}
    \label{fig:brood}
\end{figure}

\subsection{Gaussian DAG Models: Intuition Behind Operating on \texorpdfstring{$(\mathcal{H}, \prec)$}{(H, ≺)} Pairs}

Although it may seem contrived to operate on $(\mathcal{H}, \prec)$, the Gaussian DAG model (GDM) provides the intuition for doing so. 

For analysis of continuous random variables, it is standard to assume the data are drawn from a \textit{Gaussian graphical model} (GGM). A GGM is any model in the set $\mathcal{M}_{G} = \{\mathcal{N}_p(0, \Sigma), K=\Sigma^{-1} \in \mathbb{P}_{G}\}$, where $\mathcal{N}_p$ is the $p$-dimensional multivariate Gaussian distribution, $\mathbb{P}_{G} = \{Y \in \mathbb{P}_p: Y_{ij} =0 \text{ if } i-j \notin E_{G}, \forall i,j \in V\}$, and $\mathbb{P}_p$ is the cone of $p \times p$ positive-definite matrices \citep{dempster_covariance_1972}. Similarly, a Gaussian DAG model (GDM) stipulates that if $i \rightarrow j \notin E_{G}$ then $L_{ij}=0$, where $K=LD^{-1}L^T$ is the modified Cholesky decomposition \citep{ben-david_high_2015}. Here, the order $\prec$ is assumed to be known and fixed, such that the rows and columns of $L$ are arranged by $\prec$.

In our setup, a DAG $G$ can only be used to define $L$ if $G \in \mathcal{G}_\mathcal{H}$, where $\mathcal{H} \in \mathbb{D}_p$ is the restricted search space. In addition, the ordering $\prec$ dictates the arrangement the rows and columns of $K, L$, and $D$ (i.e., $L$, a lower triangular matrix, is arranged such that $L_{ij}$ can be nonzero if $\prec_{[i]} < \prec_{[j]}$).  As a result, $\theta_{(\mathcal{H}, \prec)}$ represents a realization of the nonzero entries in the L matrix based on both $\mathcal{H}$ and $\prec$. Figure \ref{fig:matrix_h_prec_example} illustrates this concept in a 4-node graph problem.

The mixture kernel sampler is defined by two move types: (1) $Q_0$, the transition kernel for the standard restricted order MCMC sampler (displacements within the current search space) and (2) $Q_1$, the transition kernel for the birth-death MH algorithm (expansions or contractions of the current search space). Figure \ref{fig:gaussian_dag_moves_mixture} illustrates how these two move types allow an MCMC sampler to change which graphs can be considered; each gives access to unique sets of $L$ matrices in Gaussian DAG setups.

\begin{figure}[ht]
    \centering
    \begin{subfigure}[b]{\textwidth}
        \centering
        \[
        \begin{array}{c}
        H = \begin{bNiceMatrix}[first-row,first-col]
        & 1 & 2 & 3 & 4\\
    1   & 0 & 0 & 1 & 1\\
    2   & 1 & 0 & 0 & 0\\
    3   & 1 & 1 & 0 & 1\\
    4   & 0 & 1 & 1 & 0
    \end{bNiceMatrix}
    \\[2.5em]
    \prec = (1,3,4,2)
    \end{array}
    \quad
    \tikz[baseline]{\node[anchor=base, scale=1.5] at (0,0) {$\Rightarrow$};}
    \quad
        L \in \begin{bNiceMatrix}[first-row,first-col]
            & 1 & 3 & 4 & 2\\
            1   & \Block[fill=red!15,rounded-corners]{1-1}{} 1 & 
            \Block[fill=violet!15, rounded-corners]{1-2}{} 0 & 0 & 
            \Block[fill=red!15,rounded-corners]{1-1}{} 0\\
            3   &\Block[fill=cyan!15,rounded-corners]{1-1}{}  \ell_{31} & 
            \Block[fill=red!15,rounded-corners]{1-1}{} 1 &
            \Block[fill=violet!15,rounded-corners]{1-2}{} 0 & 0\\
            4   &\Block[fill=green!15,rounded-corners]{1-1}{} 0 & 
            \Block[fill=cyan!15,rounded-corners]{1-1}{} \ell_{43} & 
            \Block[fill=red!15,rounded-corners]{1-1}{} 1 & 
            \Block[fill=violet!15,rounded-corners]{1-1}{} 0\\
            2   &\Block[fill=cyan!15,rounded-corners]{1-1}{} \ell_{21} &
            \Block[fill=green!15,rounded-corners]{1-2}{} 0 & 0 & 
            \Block[fill=red!15,rounded-corners]{1-1}{} 1
        \end{bNiceMatrix}
        \]
        \caption{}
        \label{fig:matrix_h_prec_example}
    \end{subfigure}

    \vspace{0.5em} 

    \begin{subfigure}[b]{\textwidth}
        \centering
        \begin{tikzpicture}[every node/.style={inner sep=0}, node distance=1.5cm and 2.5cm]
            \node (base) at (0, 0) {
            \begin{minipage}{0.3\textwidth}
                \centering
                \[
                \begin{bNiceMatrix}[first-row,first-col]
                    & 1 & 3 & 4 & 2\\
                1   & \Block[fill=red!15,rounded-corners]{1-1}{} 1 & 
                \Block[fill=violet!15, rounded-corners]{1-2}{} 0 & 0 & 
                \Block[fill=red!15,rounded-corners]{1-1}{} 0\\
                3   &\Block[fill=cyan!15,rounded-corners]{1-1}{}  \ell_{31} & \Block[fill=red!15,rounded-corners]{1-1}{} 1 &
                \Block[fill=violet!15,rounded-corners]{1-2}{} 0 & 0\\
                4   &\Block[fill=green!15,rounded-corners]{1-1}{} 0 & 
                \Block[fill=cyan!15,rounded-corners]{1-1}{} \ell_{43} & \Block[fill=red!15,rounded-corners]{1-1}{} 1 & 
                \Block[fill=violet!15,rounded-corners]{1-1}{} 0\\
                2   &\Block[fill=cyan!15,rounded-corners]{1-2}{}
                \ell_{21} &\Block[fill=green!15,rounded-corners]{1-2}{} 0 & 0 & \Block[fill=red!15,rounded-corners]{1-1}{} 1
                \end{bNiceMatrix}\]
            \end{minipage}
            };

            \node (q0) at (7, 1.35) {
            \begin{minipage}{0.3\textwidth}
                \centering
                \[
                \begin{bNiceMatrix}[first-row,first-col]
                    & 1 & 2 & 3 & 4\\
                1   & \Block[fill=red!15,rounded-corners]{1-1}{} 1 & 
                \Block[fill=violet!15, rounded-corners]{1-2}{} 0 & 0 & 
                \Block[fill=red!15,rounded-corners]{1-1}{} 0\\
                2   &\Block[fill=cyan!15,rounded-corners]{1-1}{}
                \ell_{21}  & \Block[fill=red!15,rounded-corners]{1-1}{} 1 &\Block[fill=red!15,rounded-corners]{1-2}{} 0 & 0\\
                3   &\Block[fill=cyan!15,rounded-corners]{1-2}{}  \ell_{31} & \ell_{32} & \Block[fill=red!15,rounded-corners]{1-1}{} 1 &
                \Block[fill=violet!15,rounded-corners]{1-1}{} 0\\
                4   &\Block[fill=green!15,rounded-corners]{1-1}{} 0 & 
                \Block[fill=cyan!15,rounded-corners]{1-2}{} \ell_{42} & \ell_{43} &
                \Block[fill=red!15,rounded-corners]{1-1}{} 1
                \end{bNiceMatrix}\]
            \end{minipage}
            };

            \node (q1) at (7, -1.35) {
            \begin{minipage}{0.3\textwidth}
                \centering
                \[
                \begin{bNiceMatrix}[first-row,first-col]
                    & 1 & 3 & 4 & 2\\
                1   & \Block[fill=red!15,rounded-corners]{1-1}{} 1 & 
                \Block[fill=violet!15, rounded-corners]{1-2}{} 0 & 0 & 
                \Block[fill=red!15,rounded-corners]{1-1}{} 0\\
                3   &\Block[fill=cyan!15,rounded-corners]{1-1}{}  \ell_{31} & \Block[fill=red!15,rounded-corners]{1-1}{} 1 &
                \Block[fill=violet!15,rounded-corners]{1-2}{} 0 & 0\\
                4   &\Block[fill=green!15,rounded-corners]{1-1}{} 0 & 
                \Block[fill=cyan!15,rounded-corners]{1-1}{} \ell_{43} & \Block[fill=red!15,rounded-corners]{1-1}{} 1 & 
                \Block[fill=violet!15,rounded-corners]{1-1}{} 0\\
                2   &\Block[fill=cyan!15,rounded-corners]{1-2}{}
                \ell_{21} & \ell_{23}& \Block[fill=green!15,rounded-corners]{1-1}{} 0 & \Block[fill=red!15,rounded-corners]{1-1}{} 1
                \end{bNiceMatrix}\]
            \end{minipage}
            };

            \draw[->, thick] (base.east) -- node[above left] {\(Q_0\)} (q0.west);
            \draw[->, thick] (base.east) -- node[below left] {\(Q_1\)} (q1.west);
        \end{tikzpicture}
        \caption{}
        \label{fig:gaussian_dag_moves_mixture}
    \end{subfigure}
    
     \caption{(a) Construction of a restricted Cholesky space given restricted search space $\mathcal{H}$ and variable order $\prec$. Light-blue entries are parameters in $\Theta_{(\mathcal{H}, \prec)}$. Purple entries are in $\Theta_{\mathcal{H}}$ but not $\Theta_{\prec}$. Green entries are in $\Theta_{\prec}$ but not $\Theta_{\mathcal{H}}$. Red entries are in neither.  
    (b) Two local moves on the space of Gaussian DAG models: $Q_0$ changes the ordering $\prec$ by relocating node 2 from position 4, and $Q_1$ adds edge $3 \to 2$ to $\mathcal{H}$. Both modify the Cholesky factor $L$.}
    \label{fig:combined_matrix_and_moves}
\end{figure}

\subsection{Choice of parameters}\label{subsec:param_choices}
In Sections \ref{subsec:bdmcmc} and \ref{subsec:mixture_kernel}, we introduce 2 user-specified parameters, $c^*$ and $\ell$, as shown in \eqref{eq:death_rate} and \eqref{eq:mixture_kern}, respectively. We briefly discuss their implications.

\subsubsection{Controlling Mixing of the Chain via \texorpdfstring{$\ell$}{TEXT}}

The parameter $\ell \in [0,1]$ is the mixture weight placed on $Q_1$. For any value $\ell \in (0, 1]$, the target distribution of the Markov chain is invariant and converges to the desired transdimensional $\mu$ (the choice $\ell=0$ also leads to a valid sampler but only for the conditional $\mu_m$). However, $\ell$ affects the mixing time of the chain. Since the $Q_1$ steps are more computationally expensive than the $Q_0$ steps in BROOD, it is practically valuable to have $\ell$ be small, and we suggest the empirical choice $\ell=0.1$, consistent with recommendations from \cite{geyer_simulation_1994}.

\subsubsection{Controlling Fidelity to the True Posterior via \texorpdfstring{$c^*$}{TEXT}}

Unlike $\ell$, the parameter $c^* \in (0,1]$ determines the limiting distribution of the Markov chain (and thus the approximate target), as specified in Corollary \ref{theorem:posterior_form}. We note that while setting $c^*=0$ does not satisfy the necessary conditions for recovering a posterior using the techniques in \cite{preston_spatial_1977}, $\lim_{t \to \infty}\mathcal{H}_t = K_p$ where $\mathcal{H}_t$ is the $t^{\text{th}}$ sample of this process; this amounts to sampling from a chain that converges to the posterior supported by the entire space.  Making $c^*$ nonzero loses the fidelity of the full posterior, but allows for computationally tractable sampling for even high-dimensional problems. Corollary \ref{theorem:posterior_form} directly shows the impact of $c^*$ on the fidelity of the posterior by highlighting a smooth mapping between the full space posterior which occurs when $\lim_{c^* \to 0} \mu_{Space}^{(c^*)}$ to increasingly coarser approximate posteriors as $c^*$ increases.

This relationship allows us to re-characterize the $\varepsilon_{\mathcal{H}}$ term in \eqref{eq:lower_upper_bound}. As $c^* \to 0$, the algorithm places exponentially more weight on larger search spaces (those with larger $|E_{\mathcal{H}}|$). This drives the expected omitted mass $\mathbb{E}[\varepsilon_{\mathcal{H}}]$ toward zero, which in turn collapses the TV bound:

\begin{equation}\lim_{c^* \to 0} D_{TV}(\pi_{Ord}^{(c^*)}, \pi_{Ord}) = 0
\end{equation}
Conversely, choosing $c^*=1$ (our suggested default) calibrates the death rates with the birth rates, resulting in a posterior that has a penalty term of 1/2 scaling for each extra edge, coinciding with the notion that adding an edge to the search space doubles the number of graphs introduced. This results in a leaner sampler that stays within high-probability, lower-dimension search spaces. This choice effectively accepts a higher $\varepsilon_{\mathcal{H}}$ in exchange for the computational efficiency of the hybrid $Q_0$ kernel.
 
If a user has access to large computational resources, it could be advantageous to lower $\varepsilon_\mathcal{H}$ by making $c^* < 1$. Values $c^* \in (0,1)$ act as a compromise between the full space posterior and the edge-penalty posterior. This directly allows a user to control the fidelity-efficiency tradeoff. In contrast, while we may expect for existing procedures to initialize the search space to have small errors in specific scenarios, such as those described in Remark \ref{rem:strong_consistency}, it is unclear how a researcher can directly control the tradeoff between fidelity and efficiency in these frameworks. Future work could involve optimizing the choice of $c^*$ by balancing the size of the bounds in Theorem \ref{theorem:tv_bound} with the available resources. One special case is when $c^*=0.5$, which amounts to assigning individual search space a relative weight proportional to the scores of graphs compatible with the search space (with no direct edge-penalty for the size of the search space).

\subsection{Computational Considerations}\label{subsec:comp_considerations}
Recent work on the restricted search space order MCMC (\cite{kuipers_efficient_2021}, \cite{viinikka_towards_2020}) shows that by reducing the number of parents considered to those contained in a restricted search space, decomposable graph scores can be (1)  precomputed prior to sampling and (2) stored in cleverly designed data structures that streamline scoring. Specifically, \cite{kuipers_efficient_2021} constructs banned score tables that allow for order MCMC scoring to have $O(Kp)$ operations, where $p$ is the number of nodes and $K$ is the sparsity between nodes (see \ref{app:score_table_bkgrd} for more details). \cite{viinikka_towards_2020} details an efficient method to build the banned score table in $O(p2^K)$. Finally, creating an initial score table for the BGe score \citep{heckerman_learning_1995} that is used to build the banned score table requires $O(K^3p2^K)$ operations.

Updating the restricted search space adds further computational overhead, with each expansion requiring $O(K^3p2^K(p-K+1))$ in previous work, or rounded up to $O(K^3p^22^K)$. However, this can be viewed as a fixed cost \citep{kuipers_efficient_2021} for finitely many expansions (and running a chain thereafter does not add computational effort). Similarly, building finitely many score tables and banned score tables can be viewed as fixed costs.

BROOD loses the ability to treat the expensive operations as finite fixed costs (because the search space changes dynamically at a rate $\ell$ instead of finitely), so we further streamline the sampler in order to lessen the computational cost.

By construction, one computational advantage of BROOD is that it only updates 1 node for any space transition, whereas the expansion procedure in \cite{kuipers_efficient_2021} updates up to all nodes for any space transition. As a result, BROOD removes a factor of $p$ operations from all table constructions at expansion steps. We acknowledge that this design choice to have smaller update steps could lead to less efficient exploration in a vacuum, though the choice is a result of having both expansion and contraction steps that could offset each other in the BROOD sampler as opposed to only expansion steps previously. $p$-sized updates can introduce a different form of inefficiency: if the current search space under-represents the parent sets of certain nodes while over-representing others, a $p$-node update forces a transition on ``settled'' nodes to accommodate ``unsettled'' ones. By allowing for node-specific births and deaths, BROOD enables the sampler to surgically refine only the regions of the graph that require adaptation, avoiding the computational and mixing overhead of unnecessary global reconfigurations.

In practice, BROOD expands the search space more efficiently than in existent implemented algorithms by using recursive block matrix formulas for determinants and Cholesky decompositions \citep{osborne_bayesian_2010}, the two key contributors to the computational burden of the BGe score \citep{kuipers_addendum_2014}. By vectorizing these recursive operations across each potential parent outside of the search space (see \ref{app:plone_bge_score} for details), BROOD removes an additional $(p-K+1)$ operations compared to previous work that processes each potential ``plus-one parent'' set one at a time. This reduces expansion scoring to typical BGe scoring within a search space applied to 1 node, which requires $O(K^32^K)$ operations.

Finally, contracting the search space is a low-cost operation, as for any contraction, all calculations have previously been performed in the larger space (and all that is required is to eliminate entries which contain parents outside the updated search space). BROOD also recycles previous calculations when creating the banned score table for the plus-one sets. We outline a procedure to do so in \ref{app:plone_memoize} that requires $O(K)$ operations for 1 node.

\section{Simulation Study}\label{sec:results_sim}
\subsection{Simulation Setup}\label{subsec:sim_setup}
We compared BROOD to two existing algorithms: the restricted order MCMC in \cite{kuipers_efficient_2021} implemented in the \texttt{BiDAG} R package (which we call BiDAG), and the graph birth-death MCMC sampler in \cite{mohammadi_bayesian_2015} implemented in the \texttt{BDgraph} R package (which we call BDgraph). As $p$ grows, BDgraph can take orders of magnitude longer to run than BROOD and BiDAG due to not restricting the search space, so we only run BDgraph for graphs with 80 or fewer nodes. We generate synthetic data from linear acyclic structural equation models (SEM) where the graph and error models vary across experiments.

To evaluate relative performance of different algorithms, we compare their estimated edge probabilities with where edges actually exist, and aggregate these results by calculating the Receiver Operator Characteristics Area Under the Curve (ROC AUC), Precision-Recall Area Under the Curve (PR AUC), the mean true-edge probability ($Pr^+$), and the mean non-edge probability ($Pr^-$). These metrics have been used to evaluate model performance in other recent work on Bayesian inference of graph models (e.g., \cite{kuipers_efficient_2021}, \cite{vogels_bayesian_2024}). Full details for calculating these can be found in \ref{app:sim_setup}. Results in Section \ref{subsec:sim_results} are displayed for ROC AUC and PR AUC, while the results for the others are shown in \ref{app:further_sims}.

We employ the BGe score \citep{heckerman_learning_1995} for BROOD and BiDAG to score and search these models, while for BDgraph, we use the G-Wishart score \citep{letac_wishart_2007} applied to centered data, due to the respective R packages for BiDAG and BDgraph providing direct native support for them. These scores are similar (though not identical, see section \ref{app:sim_setup} for details) in that they both come from a Gaussian model and use a Wishart prior on the precision for DAGs and undirected graphs, respectively. To analyze performance, we compare the estimated edge probability for each ordered pair $(i, j)$ (i.e., the prevalence of $i \to j$ in the graph samples drawn) with the true graph $G$.

By default, BiDAG does not use $R_{\mathcal{H}}$ throughout, but a relaxed restricted-score of $R_{\mathcal{H}}^+$, which allows up to 1 additional parent outside of the search space per node:
\begin{align}
    R_{\mathcal{H}}^+(\prec | D) \propto \prod_{i=1}^n \sum_{\substack{\mathbf{Pa}_i \in \prec\\ \mathbf{Pa}_i \subseteq h^i}} \left[S(X_i, \mathbf{Pa}_i | D) + \sum_{\substack{X_j \notin h^i\\ \pi_{\prec}[i] < \pi_{\prec}[j]}} S(X_i, \{\mathbf{Pa}_i, X_j \}| D)\right]\label{eq:plus_score_funct}
\end{align}
Within a BROOD chain, we sample graphs on the graph score analogous to $R_\mathcal{H}$ and $R^{+}_\mathcal{H}$ in \eqref{eq:plus_score_funct} to allow more direct comparison with BiDAG. We label these in \ref{subsec:sim_results} as BROOD and BROOD$+$, respectively. This is not identical to BiDAG's design, which uses $R^{+}_\mathcal{H}$ on the order sampling as well. BDgraph is implemented for GGMs, but not GDMs. As a result, we compare the BDgraph graph samples with undirected versions of the BROOD and BiDAG graph samples.

For each experiment, we run BROOD and BiDAG with two different starting search spaces: the skeleton from a popular constraint-based method, the PC algorithm \citep{spirtes_causation_nodate}, and the search space returned from the previous iterative expansion procedure in \cite{kuipers_efficient_2021}.

We consider two error distributions: a standardized $\mathcal{N}(0, 1)$, and an even mixture between $\mathcal{N}(0, 1)$ and $\mathcal{N}(0, 2)$. The latter allows for identifiability of edge direction \citep{hoyer_nonlinear_2008}. The former is not identifiable within the equivalence class under score-equivalent scoring, but is identifiable if the topological order is known (\cite{chen_causal_2019}). Thus, BROOD and BiDAG could both correctly identify the true edge direction as long as the prior on $\mathcal{H}$ does not skew the order sampling too much. The results in Section \ref{subsec:sim_results} are for the Gaussian error model, the results for the mixture error model can be found in \ref{app:further_sims}.

We consider three unique graph generation procedures: Erd\H{o}s-R\'{e}nyi models (ER), Stochastic Block models (SBM) and Hierarchical Stochastic Block models (hSBM). Setup and parameter choices for these models are detailed in \ref{app:sim_setup}. At a high level, block models (and hierarchical block models) produce community structures that more frequently induce multi-modal posterior graph scores than non-block models (where individual peaks correspond to the effects from unique communities). Based on Remark \ref{rem:strong_consistency}, we expect multi-modal problems to be more challenging for fixed restricted search space methods to infer. In addition, the \cite{kuipers_efficient_2021} expansion procedure continually adds the highest probability graph from outside the search space until graphs immediately outside of the search space have low probability, which suggests that it could also struggle with multi-modal problems.

Finally, we vary $K$, the maximal allowed parent set size, a required parameter for both BROOD and BiDAG and the primary determinant of computational cost, as discussed in \ref{subsec:comp_considerations}. We test two unique maximal parent set sizes for BROOD: (1) ``plus-one sparsity'', which allows up to one additional parent beyond the maximal parent set size of the initial search space, and (2) ``fixed sparsity'', which enforces a fixed maximal parent set size throughout. The former illustrates the extent to which the transdimensional move proposals can refine the search space under a limited budget, while the latter tests BROOD's ability to escape local optima when the score is multi-modal and the initialization is poor. Since fixed sparsity is more computationally expensive, we limit these experiments to graphs with at most 140 nodes. In addition, for the fixed sparsity experiments, we include a third initialization of the search space (to more comprehensively investigate the impact of initialization) which is the skeleton returned from Greedy Equivalence Search \citep{chickering_optimal_structure_2002}.

Further details about the simulation setup can be found in Appendix \ref{app:sim_setup}.

\subsection{Simulation Results}\label{subsec:sim_results}
\subsubsection{Performance}
Figures \ref{fig:auc_roc_gauss_dag} and \ref{fig:auc_pr_gauss_dag} show the ROC AUC and PR AUC, respectively, of the estimated edge probabilities obtained from DAGs sampled by BiDAG and BROOD, under plus-one sparsity. These experiments align with the intuition provided by Remark \ref{rem:strong_consistency} on strong consistency: BiDAG (using the search space returned from the iterative expansion procedure) is strongest when the number of samples is $10$ times larger than the number of nodes. In these settings, BiDAG outperforms BROOD, which is unsurprising; settings where a smartly-chosen restricted search space introduces minimal expected error should yield better empirical performance than a flexible model that puts a prior on the set of search spaces. 

In contrast, when $p$ is large or when there are very few samples relative to $p$, BROOD outperforms BiDAG within search spaces, especially in PR AUC. This suggests that the added contraction steps of BROOD increase the precision because it can prune poorly-chosen edges in the initial search space (although this sometimes also reduces the true positive rate by pruning true edges). Since it is often unrealistic in practical applications to obtain many more samples than the number of nodes in graphical data, it highlights the flexibility of an adaptive model like BROOD and its ability to robustly handle real-world data.

\begin{figure}
    \centering
    \includegraphics[width=0.99\linewidth]{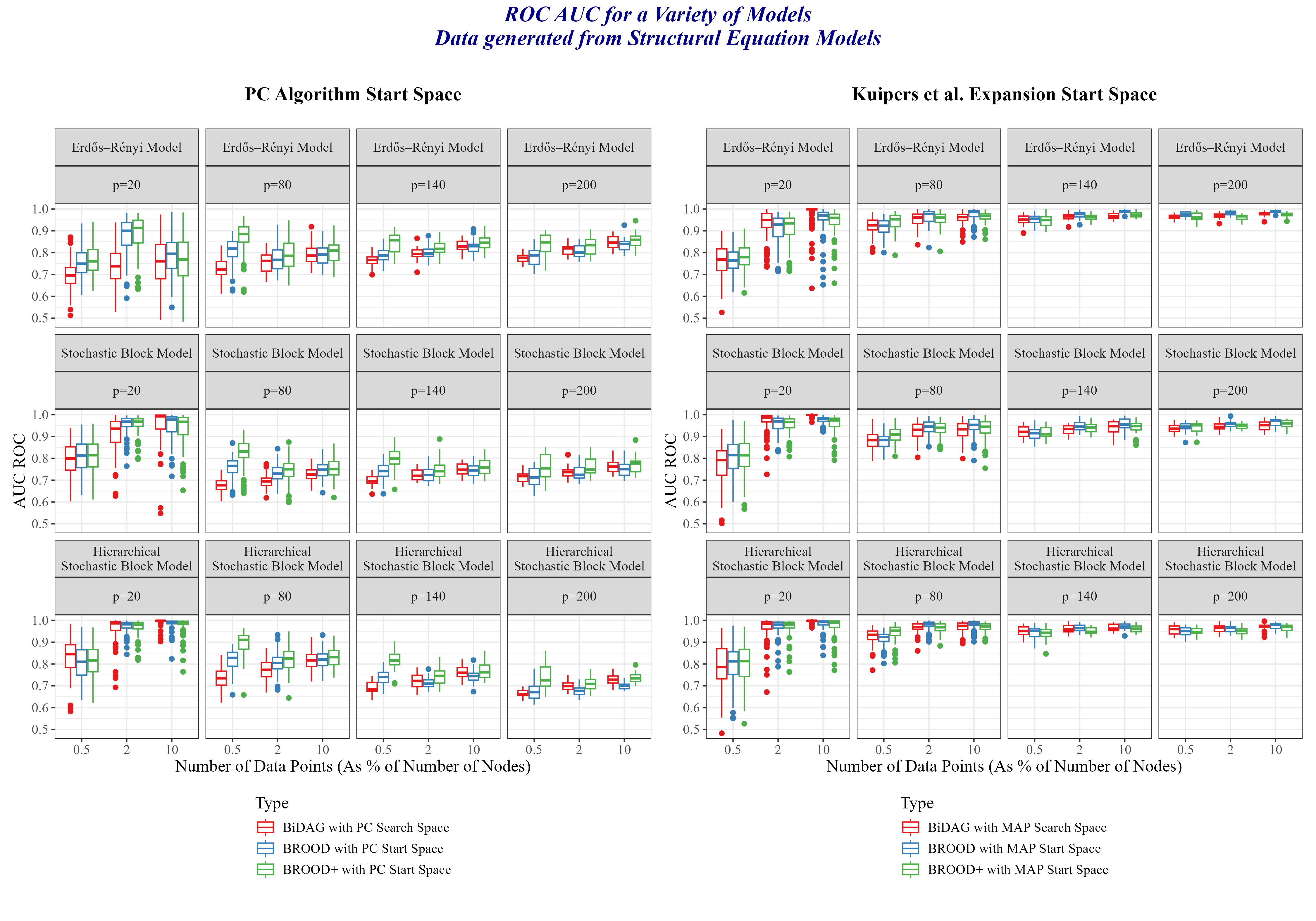}
    \caption{ROC AUC results with Gaussian SEM data, using plus-one sparsity.}
    \label{fig:auc_roc_gauss_dag}
\end{figure}

\begin{figure}
    \centering
    \includegraphics[width=0.99\linewidth]{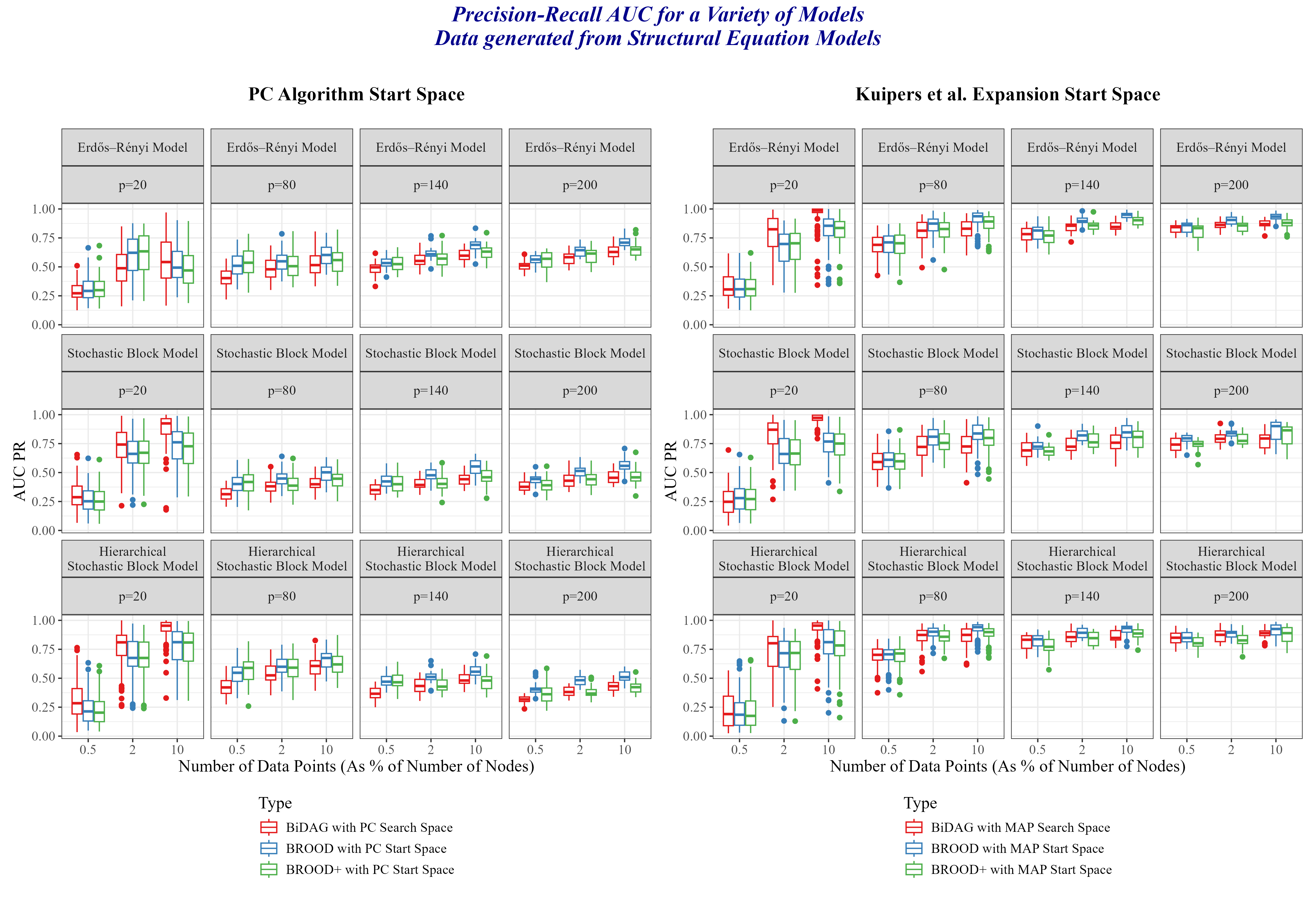}
    \caption{ROC PR results with Gaussian SEM data, using plus-one sparsity.}
    \label{fig:auc_pr_gauss_dag}
\end{figure}

Figure \ref{fig:auc_roc_gauss_dag2} compares the ROC AUC of estimated edge probabilities of DAGs sampled from BiDAG and BROOD, on the three starting search spaces mentioned in \ref{subsec:sim_setup}, where we impose fixed sparsity across all initializations. The within-initialization trends are similar to what we observed in the plus-one sparsity experiments: BROOD significantly outperforms BiDAG in harder regimes for inference, and BiDAG thrives in settings that align with Remark \ref{rem:strong_consistency}. 

The fixed sparsity experiments shed light on between-initialization trends. For small datasets, initializing BROOD using only the PC or GES search space performs similarly to the fixed-space sampler after applying the greedy iterative expansion algorithm. We suspect that this result is present for our small dataset experiments but not for our large ones because the posteriors on small datasets are more diffuse than those on larger datasets on average. Under diffuser posteriors, the birth-death process accepts more space change proposals than when the posterior is more concentrated around local maxima. This implies that the space-changing procedure in BROOD can approach the performance fixed search space sampler when the posterior is diffuse enough, even when the initialization is poor for BROOD. This trend holds true across all models, including for data generated from mixture SEMs, which can be seen in Appendix \ref{app:further_sims}.

\begin{figure}
    \centering
    \includegraphics[width=0.99\linewidth]{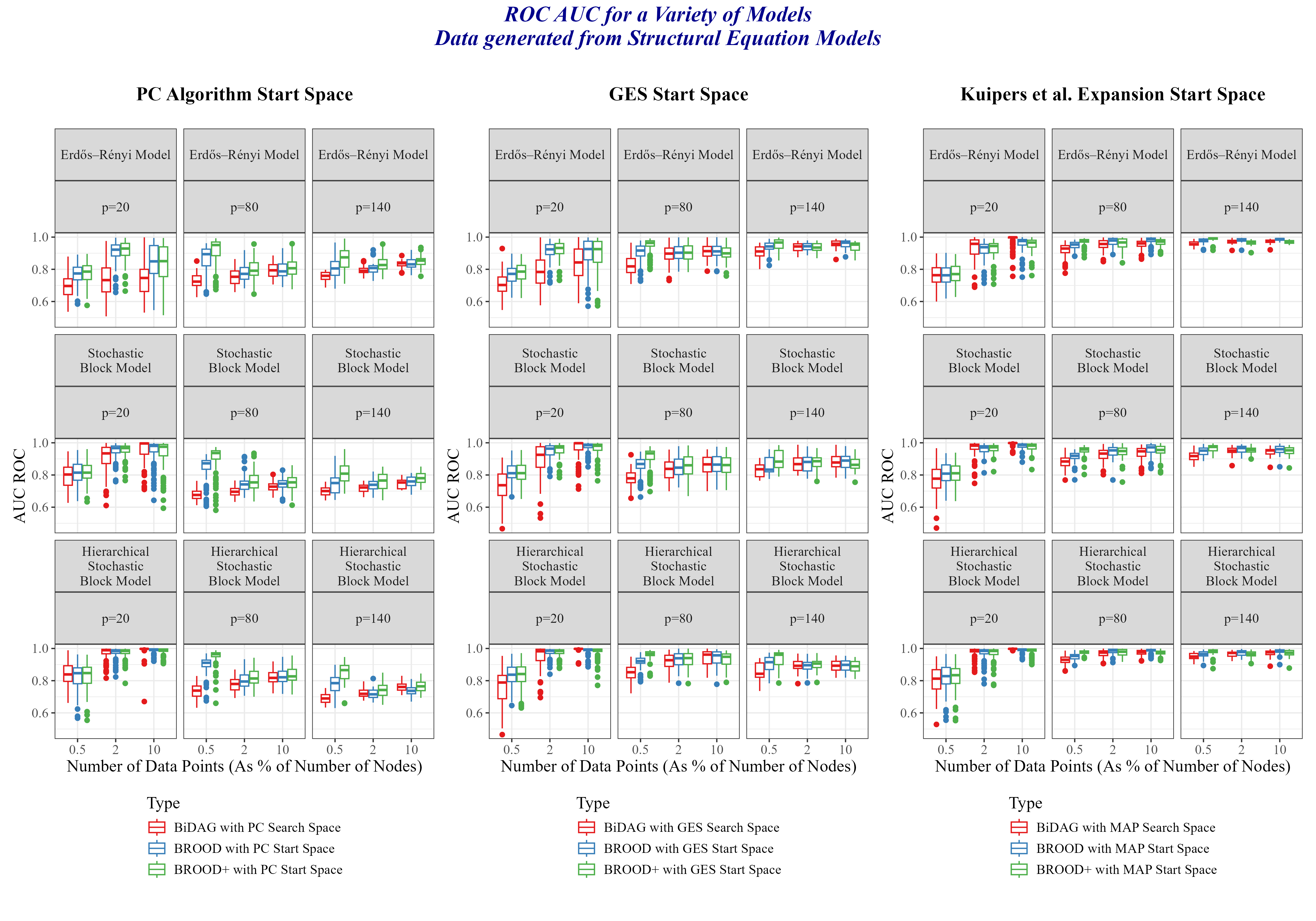}
    \caption{ROC AUC results with Gaussian SEM data, using fixed sparsity.}
    \label{fig:auc_roc_gauss_dag2}
\end{figure}

Figure \ref{fig:auc_roc_gauss_skel} compares the skeletons of all 3 methods, which tells a similar story -- BROOD excels relative to the other methods in complex learning environments. In addition, the order-based samplers outperform BDgraph on average, suggesting that the mixing time of procedures on the space of topological orders can have large practical advantages over those on the space of graphs.

\begin{figure}
    \centering
    \includegraphics[width=0.99\linewidth]{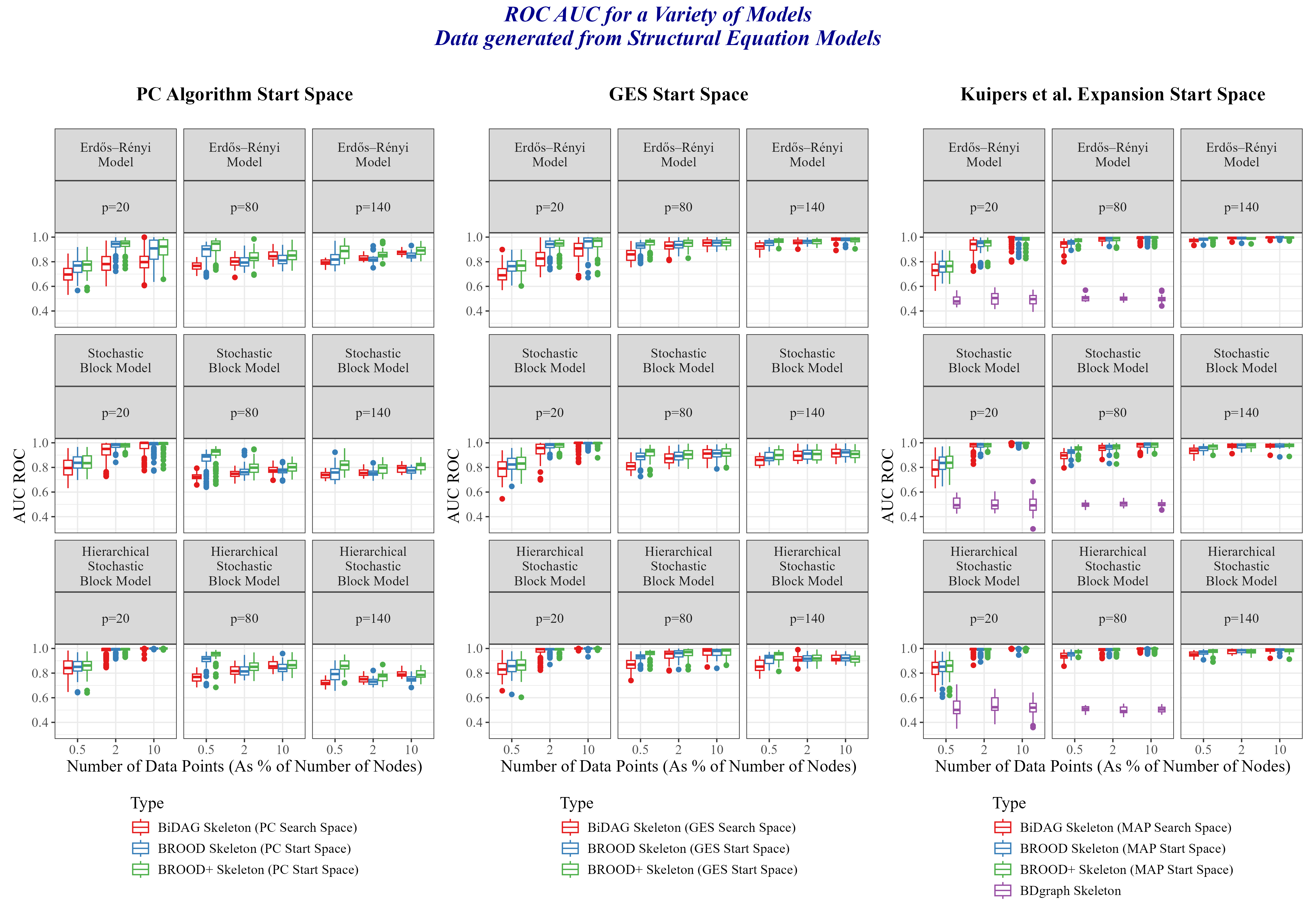}
    \caption{Skeleton version of ROC AUC results, using fixed sparsity.}
    \label{fig:auc_roc_gauss_skel}
\end{figure}

\subsubsection{Computation Time}

We plot the run time of the models in figure \ref{fig:log_time_gauss}. Using the mixture kernel sampler (which much of the time uses the fast restricted order sampling) allows for BROOD to be significantly faster than BDgraph, the previous birth-death sampler. However, BROOD's added flexibility comes at the cost of compute time compared to BiDAG. While BROOD is significantly slower than BiDAG on sparse data problems, the more complex the problem, the closer the run times are between BROOD and BiDAG. This suggests that in scenarios where the stopping criterion in the BiDAG expansion procedure is difficult to meet, BROOD may be an especially good choice for performing structure inference. Moreover, when BROOD is significantly slower than BiDAG, for $p=200$, it still runs $200^2\log(200) = 211,913$ MCMC steps within hours, which is quite reasonable in the structure inference literature at large \citep{vogels_bayesian_2024}.

\begin{figure}
    \centering
    \includegraphics[width=0.99\linewidth]{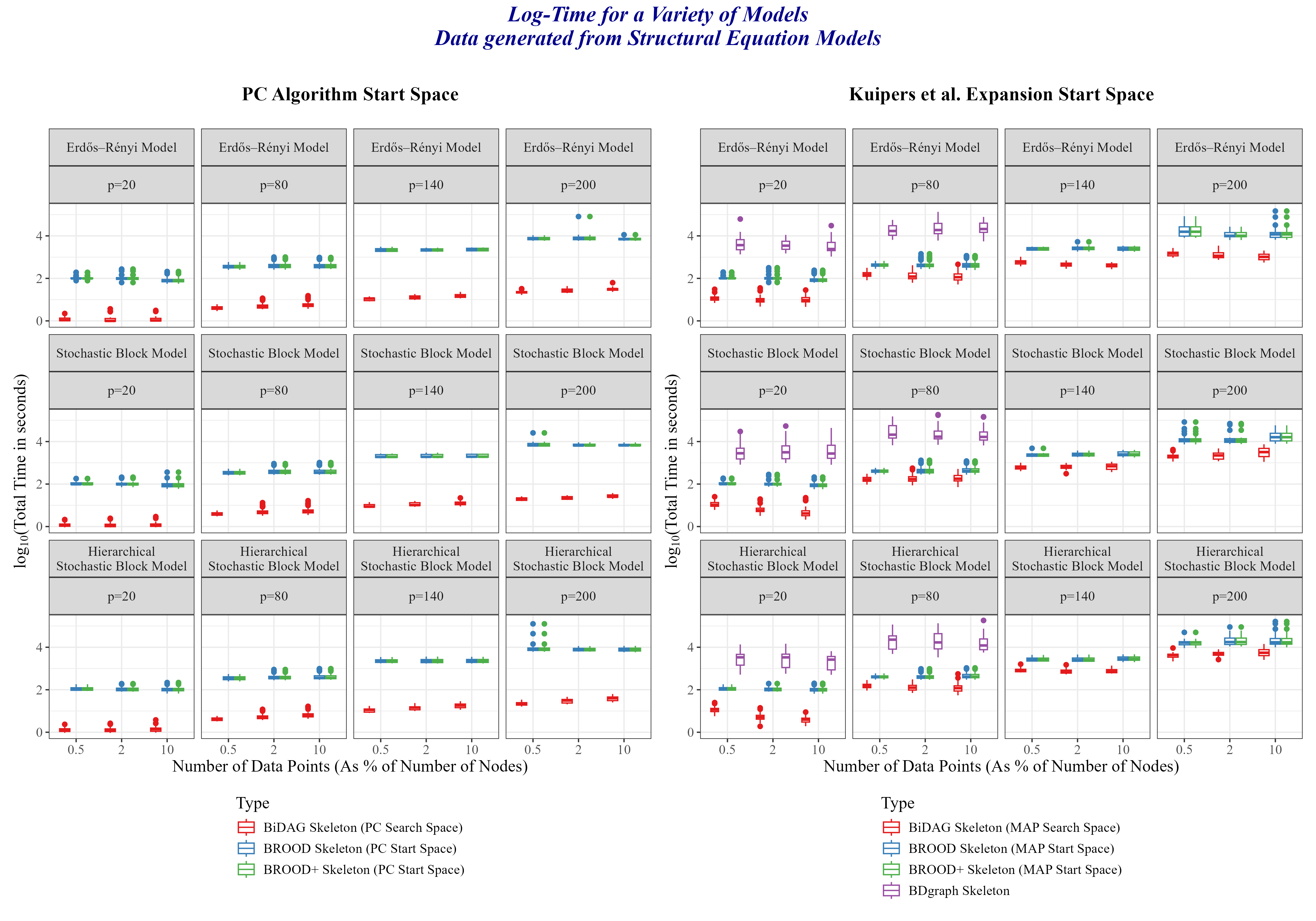}
    \caption{$\log_{10}(\text{Run Time})$ results for Gaussian SEM data, using plus-one sparsity}
    \label{fig:log_time_gauss}
\end{figure}

\section{Discussion}\label{sec:discussion}

In this work, we quantify the error introduced by fixed-space hybrid order MCMC samplers, and develop a transdimensional alternative, BROOD, which samples from a richer class of posteriors than previous hybrid methods. This class of posteriors gives the user direct control of the compromise between fidelity to the full posterior with no restrictions and computational efficiency, a quality that is lacking in previous hybrid MCMC samplers for DAGs. In addition, by allowing transitions between restricted search spaces, BROOD avoids the approximation error induced by conditioning on a fixed restricted search space, at the cost of computational expedience. This flexibility proves especially valuable empirically when the posterior is likely to be multi-modal, aligning with our intuition based on theorem \ref{theorem:tv_bound}.

While BROOD is a promising starting point towards understanding the extent of error in hybrid MCMC graph inference techniques, we do not directly show any theoretical guarantees where BROOD outperforms other methods given a finite budget of computational resources. In addition, empirically, BROOD can only sometimes recover from a poorly-initialized starting search space compared to using existing faster hybrid procedures to search for a high-quality fixed restricted search space.

Two natural directions for future work both involve linking our theoretical results to user-controlled choices.

First, one could optimize the user-specified parameter $c^*$, which regulates the size of the death rates relative to the birth rates. Given a fixed computational budget (perhaps in terms of the maximal allowed sparsity $K$ and/or the number of nodes $p$), there may be an optimal $c^*$ that can be recovered from constrained optimization that minimizes the expectation of the total variation distance between the hybrid sampler and the full sampler. Characterizing such an optimum may allow for a direct proof that transdimensional algorithms not only \textit{can} outperform fixed-space algorithms, but \textit{will} do so in sufficiently challenging problems. It may also be possible to tune $c^*$ for specific problems, though empirically evaluating relative performance across different $c^*$ may be computationally demanding. 

Second, it may be possible to refine the choice of starting space. While we showed simulation results for three useful starting spaces (what is returned from the PC algorithm, what is returned from the GES algorithm, and what is returned from the existing expansion procedure), there is opportunity for further exploration. For the expansion procedure in \cite{kuipers_efficient_2021}, it is not recommended to use algorithms like Greedy Equivalent Search \citep{chickering_optimal_structure_2002} to form the initial search space, due to their tendency to include many false positive edges (which would increase the computational cost of expanding the search space). BROOD, on the other hand, can afford to consider starting search spaces with many non-edges since it can both expand and contract the search space at any transdimensional step. As a result, relating properties of the starting search space (such as the false positive rate, false negative rate, true positive rate, and true negative rate) to the expected search space error could reveal principled guidelines for initializing BROOD.

Another potentially valuable future extension is to add parallel tempering to BROOD. Given that empirically, BROOD is most robust to poorly initialized search spaces when the search space posterior is flatter (which allows for more acceptances of transdimensional proposals), adding a temperature component could allow BROOD to explore the space better when the posterior is jagged or concentrated. This approach has been employed in other structure learning MCMC procedures to efficiently search concentrated posteriors \citep{barker_mc4tempering_2010}.

Even without these potential future optimizations, our transdimensional procedure adds flexibility compared to previous methods, and provides a rich opportunity for more accurate inference on complicated structures.

More broadly, there is scope to build flexible, computationally-efficient samplers based on transdimensional samplers beyond the particular choice used by BROOD. BROOD is defined by rates corresponding to a specific prior on the set of the search spaces that (1) meets the conditions required for valid birth-death sampling, and (2) favors search spaces whose associated graphs have high posterior mass on average. Other principled transdimensional priors could also satisfy these criteria, or could be advantageous in other ways.

Finally, while theorem \ref{theorem:tv_bound} clarifies certain sources of error, we have not addressed the broader convergence of order MCMC methods. This remains a gap in the current literature that is difficult to solve because order MCMC can get stuck within an equivalence class for arbitrarily long periods \citep{zhou_complexity_2023}. Nonetheless, numerical experiments suggest that order MCMC methods mix faster than graph-space samplers \citep{friedman_being_2003}, but we would like to see new theory that can formally check convergence for order MCMC approaches.

\clearpage
\glsresetall{}

\clearpage

\bibliographystyle{apalike}
\bibliography{bib}

\clearpage
\appendix
\section{Proof Work}
\subsection{Proof of Theorem \ref{theorem:tv_bound}}\label{app:tv_bound}
Let $\mathbb{G}_p, \mathbb{D}_p$ be the space of DAGs with $p$ nodes and the space of directed graphs with $p$ nodes, respectively, and let $\mathcal{H}=(V, E_{\mathcal{H}}) \subseteq \mathbb{D}_p$ be the restricted search space with $\mathcal{G}_\mathcal{H}$ representing the corresponding set of admissible DAGs based on $\mathcal{H}$. Define $\mathcal{G}_{\neg \mathcal{H}} := \mathbb{G}_p \setminus \mathcal{G}_{\mathcal{H}}$ and let $\neg \mathcal{H}$ serve as the symbolic index for this complementary DAG set. Define $\varepsilon_{\mathcal{H}} := 1 - \pi_{DAG}(\mathcal{H} | D)$.

We will use the following equality that relates $\pi_{Ord}(\cdot|D)$ and $\pi_{Ord}^{(\mathcal{H})}(\cdot|D)$ in bounding the total variation:
\begin{align}
    \pi_{Ord}(\cdot | D) &= (1-\varepsilon_{\mathcal{H}})\pi_{Ord}^{(\mathcal{H})}(\cdot | D) + \varepsilon_{\mathcal{H}} \pi_{Ord}^{(\neg \mathcal{H})}(\cdot | D)\label{eq:mixture}
\end{align}
where $\pi_{Ord}^{(\neg \mathcal{H})}$ is the order posterior restricted to the graphs in $\mathcal{G}_{\neg \mathcal{H}}$. Equation \ref{eq:mixture} can be thought of as rewriting the order posterior as a mixture between the restricted target and an omitted target.

We first consider the Hellinger distance ($d_H$) on the set of orders $\prec_1,...,\prec_{p!} \in \mathbb{S}^p$:

\begin{align*}
    d_H(\pi_{Ord}^{(\mathcal{H})}, \pi_{Ord}) &= \sqrt{\sum_{i=1}^{p!} \left(\sqrt{\pi_{Ord}^{(\mathcal{H})}(\prec_i | D)}- \sqrt{\pi_{Ord}(\prec_i)}\right)^2}\\
    &= \sqrt{\sum_{i=1}^{p!} \left(\sqrt{\pi_{Ord}^{(\mathcal{H})}(\prec_i | D)}- \sqrt{(1-\varepsilon_\mathcal{H})\pi_{Ord}^{(\mathcal{H})}(\prec_i | D) + \varepsilon_\mathcal{H} \pi_{Ord}^{(\neg\mathcal{H})}(\prec_i | D)}\right)^2}\\
    &= \sqrt{
    \sum_{i=1}^{p!} (2-\varepsilon_\mathcal{H})a + \varepsilon_\mathcal{H} b -2 \sqrt{(1-\varepsilon_\mathcal{H})a^2 + \varepsilon_\mathcal{H} ab}
    } \\
    &\qquad \text{  where } a=\pi_{Ord}^{(\mathcal{H})}(\prec_i | D), b=\pi_{Ord}^{(\neg\mathcal{H})}(\prec_i | D)\\
    &= \sqrt{(2-\varepsilon_\mathcal{H}) + \varepsilon_\mathcal{H} - 2 \sum_{i=1}^{p!} \sqrt{(1-\varepsilon_\mathcal{H})\pi_{Ord}^{(\mathcal{H})}(\prec_i | D)^2 + \varepsilon_\mathcal{H} \pi_{Ord}^{(\mathcal{H})}(\prec_i | D)\pi_{Ord}^{(\neg\mathcal{H})}(\prec_i | D)}} \\
    &= \sqrt{2 - 2 \sum_{i=1}^{p!} \sqrt{(1-\varepsilon_\mathcal{H})\pi_{Ord}^{(\mathcal{H})}(\prec_i | D)^2 + \varepsilon_\mathcal{H} \pi_{Ord}^{(\mathcal{H})}(\prec_i | D)\pi_{Ord}^{(\neg\mathcal{H})}(\prec_i | D)}}
\end{align*}

Since the Hellinger distance can take on values between $[0, \sqrt{2}]$ \citep{gibbs_choosing_2002}, the summation above must be between $[0, 1]$.

We now consider the summation above, $\sum_{i=1}^{p!} \sqrt{(1-\varepsilon_\mathcal{H})\pi_{Ord}^{(\mathcal{H})}(\prec_i | D)^2 + \varepsilon_\mathcal{H}\pi_{Ord}^{(\mathcal{H})}(\prec_i | D)\pi_{Ord}^{(\neg\mathcal{H})}(\prec_i | D)}$, and tighten the lower bound of 0:

\begin{align*}
    \sum_{i=1}^{p!} \sqrt{(1-\varepsilon_\mathcal{H})\pi_{Ord}^{(\mathcal{H})}(\prec_i | D)^2 + \varepsilon_\mathcal{H} \pi_{Ord}^{(\mathcal{H})}(\prec_i | D)\pi_{Ord}^{(\neg\mathcal{H})}(\prec_i | D)} &\geq \sum_{i=1}^{p!} \sqrt{(1-\varepsilon_\mathcal{H})\pi_{Ord}^{(\mathcal{H})}(\prec_i | D)^2}\\
    &= \sqrt{1-\varepsilon_\mathcal{H}}\sum_{i=1}^{p!}\pi_{Ord}^{(\mathcal{H})}(\prec_i | D) \\
    &= \sqrt{1-\varepsilon_\mathcal{H}} = \sqrt{1-\varepsilon_\mathcal{H} \times 1}\\
    \sum_{i=1}^{p!} \sqrt{(1-\varepsilon_\mathcal{H})\pi_{Ord}^{(\mathcal{H})}(\prec_i | D)^2 + \varepsilon_\mathcal{H} \pi_{Ord}^{(\mathcal{H})}(\prec_i | D)\pi_{Ord}^{(\neg\mathcal{H})}(\prec_i | D)} &\leq 1 = \sqrt{1-\varepsilon_\mathcal{H} \times 0}
\end{align*}

Since the sum above lies between [$\sqrt{1-\varepsilon_\mathcal{H} \times 1}, \sqrt{1-\varepsilon_\mathcal{H} \times 0}$], there exists a constant $c$, $0 \leq c \leq 1$ such that $\sqrt{1-\varepsilon_\mathcal{H} \times c}$ precisely equals the sum. As a result, $d_H(\pi_{Ord}^{(\mathcal{H})}, \pi_{Ord}) = \sqrt{2 - 2\sqrt{1-c\varepsilon_\mathcal{H}}}$.

Using the relationship that $\frac{d_H^2}{2} \leq D_{{TV}} \leq d_H$ \citep{gibbs_choosing_2002},

\begin{align*}
    D_{{TV}}(\pi_{Ord}^{(\mathcal{H})}, \pi_{Ord}) &\geq \frac{d_H^2}{2}(\pi_{Ord}^{(\mathcal{H})}, \pi_{Ord})\\
    &\geq \frac{\sqrt{2 - 2 \sqrt{1-c\varepsilon_\mathcal{H}}}^2}{2}\\
    &= 1 - \sqrt{1-c\varepsilon_\mathcal{H}}\\
    D_{{TV}}(\pi_{Ord}^{(\mathcal{H})}, \pi_{Ord}) &\leq d_H(\pi_{Ord}^{(\mathcal{H})}, \pi_{Ord}) = \sqrt{2 - 2\sqrt{1-c\varepsilon_\mathcal{H}}}
\end{align*}

Finally, since $D_{{TV}}$ is bounded by 1, it follows that $D_{{TV}}(\pi_{Ord}^{(\mathcal{H})}, \pi_{Ord}) \leq \min\left(\sqrt{2-2\sqrt{1-c\varepsilon_\mathcal{H}}}, 1\right)$.

\subsection{Proof of Lemmas \ref{theorem:min_c}, \ref{theorem:max_c}}\label{app:c_constant}

\subsubsection{Proof of Lemma \ref{theorem:min_c}}\label{app:min_c}
Assume $\pi_{Ord}^{(\mathcal{H})}(\prec | D) = \pi_{Ord}^{(\neg\mathcal{H})}(\prec | D)$ for all $\prec$. Then,
\begin{align*}
    \sum_{i=1}^{p!} \sqrt{(1-\varepsilon_\mathcal{H})\pi_{Ord}^{(\mathcal{H})}(\prec_i | D)^2 + \varepsilon_\mathcal{H} \pi_{Ord}^{(\mathcal{H})}(\prec_i | D)\pi_{Ord}^{(\neg\mathcal{H})}(\prec_i | D)} &= \sum_{i=1}^{p!} \pi_{Ord}^{(\mathcal{H})}(\prec_i | D) = 1 = \sqrt{1 - \varepsilon_\mathcal{H} \times 0}
\end{align*}

Thus, $c = 0$.

Now, assume $c=0$. Then,
\begin{align*}
    \sum_{i=1}^{p!} \sqrt{(1-\varepsilon_\mathcal{H})\pi_{Ord}^{(\mathcal{H})}(\prec_i | D)^2 + \varepsilon_\mathcal{H} \pi_{Ord}^{(\mathcal{H})}(\prec_i | D)\pi_{Ord}^{(\neg\mathcal{H})}(\prec_i | D)} &=1
\end{align*} 

Additionally, using equation \ref{eq:mixture}:

\begin{align*}
    &\sum_{i=1}^{p!} \sqrt{(1-\varepsilon_\mathcal{H})\pi_{Ord}^{(\mathcal{H})}(\prec_i | D)^2 + \varepsilon_\mathcal{H} \pi_{Ord}^{(\mathcal{H})}(\prec_i | D)\pi_{Ord}^{(\neg\mathcal{H})}(\prec_i | D)} \nonumber \\
    &\quad= \sum_{i=1}^{p!} \sqrt{(1-\varepsilon_\mathcal{H})\pi_{Ord}^{(\mathcal{H})}(\prec_i | D)^2 + \pi_{Ord}^{(\mathcal{H})}(\prec_i | D)\left(\pi_{Ord}(\prec_i | D) - (1-\varepsilon_\mathcal{H})\pi_{Ord}^{(\mathcal{H})}(\prec_i | D)\right)} \\
    &\quad= \sum_{i=1}^{p!} \sqrt{\pi_{Ord}^{(\mathcal{H})}(\prec_i | D)\pi_{Ord}(\prec_i | D)} \\
    &\quad\leq \sqrt{\sum_{i=1}^{p!} \pi_{Ord}^{(\mathcal{H})}(\prec_i | D)}\sqrt{\sum_{i=1}^{p!} \pi_{Ord}(\prec_i | D)} = 1
\end{align*}

The inequality on the last line above is a result of the Cauchy-Schwarz inequality. Equality occurs precisely when the vector square-root probabilities, i.e.,
\begin{align*}
    \begin{bmatrix}\sqrt{\pi_{Ord}^{(\mathcal{H})}(\prec_1 | D)} \\ \vdots \\ \sqrt{\pi_{Ord}^{(\mathcal{H})}(\prec_{p!} | D)}\end{bmatrix}, \begin{bmatrix}\sqrt{\pi_{Ord}(\prec_1 | D)} \\ \vdots \\ \sqrt{\pi_{Ord}(\prec_{p!} | D)}\end{bmatrix}
\end{align*}
are linearly dependent.

Since both $\pi_{Ord}^{(\mathcal{H})}$ and $\pi_{Ord}$ are probability distributions (whose components sum to 1), linear dependence implies equality of the distributions:
\begin{align*}
    \pi_{Ord}^{(\mathcal{H})}(\prec_i \mid D) &= \pi_{Ord}(\prec_i \mid D) \quad \forall i.
\end{align*}

By equation \ref{eq:mixture}, this implies $\pi_{Ord}^{(\mathcal{H})} = \pi_{Ord}^{(\neg\mathcal{H})}$.

\subsubsection{Proof of Lemma \ref{theorem:max_c}}\label{app:max_c}

Assume the PMFs $\pi_{Ord}^{(\mathcal{H})}(\cdot | D)$ and $\pi_{Ord}^{(\neg \mathcal{H})}(\cdot | D)$ have disjoint supports. This implies that for every order $\prec_i$, the pointwise product of the two distributions is zero:
\begin{align*}
    \pi_{Ord}^{(\mathcal{H})}(\prec_i|D) \cdot \pi_{Ord}^{(\neg\mathcal{H})}(\prec_i|D) &= 0 \quad \text{for all } i = 1, \dots, p!.
\end{align*}

We now evaluate the expression for $c$. Under the assumption of disjoint supports, each term in the summation simplifies as follows:
\begin{align*}
    \sum_{i=1}^{p!}\sqrt{(1-\varepsilon_\mathcal{H}) \pi_{Ord}^{(\mathcal{H})}(\prec_i|D)^2 + \varepsilon_\mathcal{H} \pi_{Ord}^{(\mathcal{H})}(\prec_i|D)\pi_{Ord}^{(\neg\mathcal{H})}(\prec_i|D)} &= \sum_{i=1}^{p!}\sqrt{(1-\varepsilon_\mathcal{H}) \pi_{Ord}^{(\mathcal{H})}(\prec_i|D)^2 + 0} \\
    &= \sqrt{1-\varepsilon_\mathcal{H}} \sum_{i=1}^{p!} \pi_{Ord}^{(\mathcal{H})}(\prec_i|D) \\
    &= \sqrt{1-\varepsilon_\mathcal{H}}.
\end{align*}
Since the sum equals $\sqrt{1-\varepsilon_\mathcal{H}}$, it follows from the definition of $c$ that $c=1$.

Conversely, assume $c=1$. By the definition of $c$, this requires:
\begin{align*}
    \sum_{i=1}^{p!}\sqrt{(1-\varepsilon_\mathcal{H}) \pi_{Ord}^{(\mathcal{H})}(\prec_i|D)^2 + \varepsilon_\mathcal{H} \pi_{Ord}^{(\mathcal{H})}(\prec_i|D)\pi_{Ord}^{(\neg\mathcal{H})}(\prec_i|D)} &= \sqrt{1-\varepsilon_\mathcal{H}}.
\end{align*}

Note that for all $i$, since probabilities are non-negative:
\begin{align*}
    \sqrt{(1-\varepsilon_\mathcal{H}) \pi_{Ord}^{(\mathcal{H})}(\prec_i|D)^2 + \varepsilon_\mathcal{H} \pi_{Ord}^{(\mathcal{H})}(\prec_i|D)\pi_{Ord}^{(\neg\mathcal{H})}(\prec_i|D)} \geq \sqrt{(1-\varepsilon_\mathcal{H}) \pi_{Ord}^{(\mathcal{H})}(\prec_i|D)^2}.
\end{align*}
Summing over all $i$ yields:
\begin{align*}
    \sum_{i=1}^{p!}\sqrt{(1-\varepsilon_\mathcal{H}) \pi_{Ord}^{(\mathcal{H})}(\prec_i|D)^2 + \varepsilon_\mathcal{H} \pi_{Ord}^{(\mathcal{H})}(\prec_i|D)\pi_{Ord}^{(\neg\mathcal{H})}(\prec_i|D)} \geq \sqrt{1 - \varepsilon_\mathcal{H}}.
\end{align*}

For this inequality to hold as an equality (which is required for $c=1$), each individual term in the sum must achieve equality. This occurs if and only if:
\begin{align*}
    \varepsilon_\mathcal{H} \pi_{Ord}^{(\mathcal{H})}(\prec_i|D)\pi_{Ord}^{(\neg\mathcal{H})}(\prec_i|D) &= 0 \quad \text{for all } i,
\end{align*}
which implies that the supports of $\pi_{Ord}^{(\mathcal{H})}(\cdot | D)$ and $\pi_{Ord}^{(\neg \mathcal{H})}(\cdot | D)$ are disjoint.

\subsection{Proof of Theorem \ref{theorem:distribution_class}}\label{app:distribution_class}

We define the spatial birth-death process following the framework of \cite{preston_spatial_1977}.

Let $\Lambda$ be a space where an individual point lives, equipped with a $\sigma$-field $\mathcal{B}$. Let $\Lambda_n$ be $n$ copies of $\Lambda$, $\Omega_n$ be the subset of $\Lambda_n$ that considers a permutation of the $n$ observed points the same, and $\Omega$ be the disjoint union of all $\Omega_n$ with accompanying $\sigma$-field $\mathcal{F}$.

The birth-death process is a continuous time jump process that evolves by adding (births) or removing (deaths) elements of $\Lambda$ over time. The process is characterized by birth and death rates. Let $B(x; \cdot)$ be a finite measure on $(\Lambda, \mathcal{B})$ describing the birth rate conditional on the current state $x$ and $D(x/\xi; \xi)$ be a $\mathcal{B}$-measurable function describing the death rate of individual point $\xi \in x$. Define $\beta, \delta: \Omega \to \mathbb{R}^{+}$ by
\begin{align*}
    \beta(x) &= B(x; \Lambda), x \in \Omega\\
    \delta(x) &= \sum_{\xi \in x} D(x / \xi; \xi) \text{ if } x \neq \emptyset, \delta(\emptyset)=0
\end{align*}

For any set $F \in \mathcal{F}$, let $F_n = F \cap \Omega_n$ denote the restriction of $F$ to the space of exactly $n$ observed points. Then, for a current state $x \in \Omega_n$, a birth results in a transition to $F_{n+1} \subset \Omega_{n+1}$, while a death results in a transition to $F_{n-1} \subset \Omega_{n-1}$. Rates $\beta$ and $\delta$ define the jump process with intensity of $\beta + \delta$, and transition kernels of:
\begin{align}
    K_{\beta}^{(n)}(x, F) &= \frac{1}{\beta(x)}B(x; \{\xi \in \Lambda: x \cup \xi \in F\}),\nonumber\\
    K_{\delta}^{(n)}(x, G) &= \frac{1}{\delta(x)}\sum_{\xi \in x, x/\xi \in G} D(x/\xi; \xi)\nonumber\\
    K(x, F) &= \frac{\beta(x)}{\beta(x)+\delta(x)} K_{\beta}^{(n)}(x, F_{n+1}) + \frac{\delta(x)}{\beta(x)+\delta(x)} K_{\delta}^{(n)}(x, F_{n-1})\label{eq:trans_kernel_bd}
\end{align}

Under mild regularity conditions (uniqueness of the solution of the Kolmogorov backward equations as specified in proposition 5.1 in \cite{preston_spatial_1977}, finite birth and death rates), the birth-death process admits a stationary distribution $\mu$ on $(\Omega, \mathcal{F})$ if and only if
\begin{align*}
    \int_F &\beta(z)+\delta(z) d\mu_n(z) = \int \beta(x)K_{\beta}^{(n-1)}(x, F)d\mu_{n-1}(x)+\int\delta(y)K_{\delta}^{(n+1)}(y, F)d\mu_{n+1}(y), \\
    &\text{ for } n \geq 1, F \in \mathcal{F}_n,
\end{align*}
which can alternatively be written as equations \ref{eq:det_bal_1} and \ref{eq:det_bal_2}.

To construct the stationary distribution, let $\omega$ be a $\sigma$-finite measure on $(\Lambda, \mathcal{B})$ and suppose there exists an $\mathcal{F} \times \mathcal{B}$-measurable function $b: \Omega \times \Lambda \to \mathbb{R}^+$ such that for all $x \in \Omega$, $F \in \mathcal{B}$ we have
\begin{align*}
    B(x; F) &= \int_F b(x, \xi) d\omega(\xi)
\end{align*}
For $n \geq 1$, let $\omega_n$ be the product of $n$ copies of $\omega$, and let $\overset{\sim}{\omega}_n$ be the measure induced by $\omega_n$ on $(\Omega_n, \mathcal{F}_n)$, and let $\overset{\sim}{\omega}_0$ be the point mass at 0.

Let $\mu$ be a probability measure on $(\Omega, \mathcal{F})$ such that for $n \geq 0$, $\mu_n$ has density $f_n$ with respect to $\overset{\sim}{\omega}=\frac{1}{n!}\overset{\sim}{\omega}_n$.

Suppose from now on that $D(x; \xi) > 0$ for all $x \in \Omega$, $\xi \in \Lambda$ and define 
\begin{align*}
    \gamma(x, \xi) &= \frac{b(x, \xi)}{D(x;\xi)}
\end{align*}

Then, 
\begin{align}
    f_{n+1}(x \cup \xi) = \gamma(x, \xi)f_n(x) \text{ for all } n \geq 0, x \in \Omega_n, \xi \in \Lambda \label{eq:recursive_definition}
\end{align}

If the $f$ is well-defined above, in the sense that
\begin{align}
    \gamma(x \cup \xi, \eta) \gamma(x, \xi) = \gamma(x \cup \eta, \xi)\gamma(x, \eta) \text{ for all } x \in \Omega, \xi, \eta \in \Lambda,\label{eq:balanced_rates}
\end{align}

then $\mu$ is defined as follows:
\begin{align}
    \mu(F) &= Z^{-1}\int_F f d\overset{\sim}{\omega} \text{ , with } Z = \int f d \overset{\sim}{\omega}\label{eq:recursive_distribution}
\end{align}

Thus, the induced stationary probability measure from a birth-death process can be recursively defined by the birth and death rates, so long as their ratio is well-defined across different $\xi, \eta \in \Lambda$.

We now apply this theory to the parameter space associated with restricted search space-order pair $(\mathcal{H}, \prec)$, as an extension of the local parameter vector $\theta_G$ (and later we can sum across all parameter sets for a fixed $(\mathcal{H}, \prec)$ to find space-order-level birth and death rates). We formally define the parameter set as the union of parameters across all consistent DAGs:
\begin{align}
\theta_{(\mathcal{H}, \prec)} := \bigcup_{G \in \mathcal{G}_{\mathcal{H}} \cap \mathcal{G}_{\prec}} \theta_G
\end{align}
In this setup, $\theta_{(\mathcal{H}, \prec)}$ characterizes the local parameter vector (e.g., the Cholesky entries $\ell_{ij}$ in a GDM) that are ``active'' under the current restrictions. We use $\pi_{\bm{G}, \theta_{\bm{G}}}$ to denote the posterior density associated with the random variable $\bm{G}$ and its parameters. We then define $\pi_{\bm{G}, \theta_{(\bm{H}, \bm{O})}}$ as the aggregated density obtained by summing over the set of DAGs consistent with a realization $(\mathcal{H}, \prec)$.

We define the birth rate associated with adding edge $e$, over a parameter set $F \subseteq \Theta_e$ in a similar way to the previous graph-space versions, but since we operate over search spaces-order pairs (which induce sets of graphs) instead of over graphs, we take an average across the elements in the induced set:
{\begin{align}
    B_e(\theta_{(\mathcal{H}, \prec)}; F) 
    &\propto \frac{1}{|\mathcal{G}_{\mathcal{H}^{+e}}|}\int_F \pi_{\bm{G}, \theta_{(\bm{H}, \bm{O})}}(\mathcal{H}^{+e}, \theta_{(\mathcal{H}, \prec)} \cup \theta_e \mid \prec, D) \, d\omega(\theta_e) \\
    &= \frac{1}{|\mathcal{G}_{\mathcal{H}^{+e}}|}\sum_{G \in \mathcal{G}_{\mathcal{H}^{+e}} } \mathbbm{1}[G \in \mathcal{G}_{\prec}]\int_F \Bigl(
       \mathbbm{1}[e \in E_G] \pi_{\bm{G}, 
       \theta_{\bm{G}}}(G, \theta_G^{-e} \cup \theta_e \mid \prec, D) \\
    &\quad\quad + \mathbbm{1}[e \notin E_G] \pi_{\bm{G}, \theta_{\bm{G}}}(G, \theta_G \mid \prec, D)
       \Bigr) d\omega(\theta_e) \nonumber \\
    &= \frac{1}{|\mathcal{G}_{\mathcal{H}^{+e}}|}\sum_{G \in \mathcal{G}_\mathcal{H} } \int_F \Bigl(
       \mathbbm{1}[G^{+e} \in \mathcal{G}_{\prec}] \pi_{\bm{G}, \theta_{\bm{G}}}(G^{+e}, \theta_G \cup \theta_e \mid \prec, D) \\
    &\quad\quad + \mathbbm{1}[G \in \mathcal{G}_{\prec}] \pi_{\bm{G}, \theta_{\bm{G}}}(G, \theta_G \mid \prec, D)
       \Bigr) d\omega(\theta_e) \nonumber
\end{align}}

We can then naturally define $b_e$ by differentiating the birth rate:
\begin{align}
    b_e(\theta_{(\mathcal{H}, \prec)}, \theta_e)
    &\propto \frac{1}{|\mathcal{G}_{\mathcal{H}^{+e}}|}\sum_{G \in \mathcal{G}_{\mathcal{H}}} \, \Big(\mathbbm{1}[G^{+e} \in \mathcal{G}_\prec]\pi_{\bm{G}, \theta_{\bm{G}}}(G^{+e}, \theta_G \cup \theta_e \mid \prec, D)\label{eq:birth_rate_param}\\
    &\quad+\mathbbm{1}[G \in \mathcal{G}_{\prec}]\pi_{\bm{G}, \theta_{\bm{G}}}(G, \theta_G \mid \prec, D)\Big)\nonumber
\end{align}
Consistent with \cite{preston_spatial_1977}, we define $\beta(\theta_{(\mathcal{H}, \prec)}) = \sum_{e \notin E_\mathcal{H}} B_e(\theta_{(\mathcal{H}, \prec)}; \Theta_e)$.

We define the death rate in a similar way, but introduce a constant $c^* \in (0, 1]$.
\begin{align}
    D_e(\theta_{(\mathcal{H}^{-e}, \prec)}; \theta_e) &\propto\frac{c^*}{|\mathcal{G}_{\mathcal{H}^{-e}}|} \pi_{\bm{G}, \theta_{(\bm{H}, \bm{O})}}(\mathcal{H}^{-e}, \theta_{(\mathcal{H}^{-e}, \prec)} \mid \prec, D) \\
    &= \frac{c^*}{|\mathcal{G}_{\mathcal{H}^{-e}}|}\sum_{G \in \mathcal{G}_{\mathcal{H}^{-e}}} \mathbbm{1}[G \in \mathcal{G}_{\prec}] \, \mathbbm{1}[\theta_G \cup \theta_e \in \theta_{(\mathcal{H}, \prec)}] \, \pi_{\bm{G}, \theta_{\bm{G}}}(G, \theta_G \mid \prec, D)\label{eq:death_rate_param}
\end{align}
We can also define $\delta(\theta_{(\mathcal{H}, \prec)}) = \sum_{e \in E_\mathcal{H}} D_e(\theta_{(\mathcal{H}^{-e}, \prec)}; \theta_e)$.

We still have not shown that these rates allow for $f$ to be well defined as shown in condition \ref{eq:balanced_rates}. We will do so now, showing that $\gamma$ is well-defined for any 2 random edges outside of $E_\mathcal{H}$, $e_1, e_2$. We need to show that 

\begin{align}
    \frac{b(\theta_{(\mathcal{H}^{+e_1}, \prec)}, \theta_{e_2})}{D(\theta_{(\mathcal{H}^{+e_1}, \prec)}, \theta_{e_2})}\frac{b(\theta_{(\mathcal{H}, \prec)}, \theta_{e_1})}{D(\theta_{(\mathcal{H}, \prec)}, \theta_{e_1})} &= \frac{b(\theta_{(\mathcal{H}^{+e_2}, \prec)}, \theta_{e_1})}{D(\theta_{(\mathcal{H}^{+e_2}, \prec)}, \theta_{e_1})}\frac{b(\theta_{(\mathcal{H}, \prec)}, \theta_{e_2})}{D(\theta_{(\mathcal{H}, \prec)}, \theta_{e_2})}\label{eq:balance_rates_graphs}
\end{align}
To prove \ref{eq:balance_rates_graphs}, we define the following 4 disjoint sets of graphs:
\begin{align*}
    G_0 &:= \{G: G \in \mathcal{G}_{\mathcal{H}}\}\\
    G_1 &:= \{G: e_1 \in G, G \in \mathcal{G}_{\mathcal{H}^{+e_1}}\}\\
    G_2 &:= \{G: e_2 \in G, G \in \mathcal{G}_{\mathcal{H}^{+e_2}}\}\\
    G_3 &:= \{G: e_1, e_2 \in G, G \in \mathcal{G}_{\mathcal{H}^{+e_1+e_2}}\}
\end{align*}

We first simplify the left-hand side of \ref{eq:balance_rates_graphs}:
\begin{align*}
    \frac{b(\theta_{(\mathcal{H}^{+e_1}, \prec)}, \theta_{e_2})}{D(\theta_{(\mathcal{H}^{+e_1}, \prec)}, \theta_{e_2})}\frac{b(\theta_{(\mathcal{H}, \prec)}, \theta_{e_1})}{D(\theta_{(\mathcal{H}, \prec)}, \theta_{e_1})}  &= \frac{1}{4c^{*^2}}\frac{\sum_{G \in G_1 \cup G_2 \cup G_3}\pi_{\bm{G}, \theta_{\bm{G}}}(G, \theta_G \mid \prec, D)}{\sum_{G \in G_0 \cup G_1}\pi_{\bm{G}, \theta_{\bm{G}}}(G, \theta_G \mid \prec, D)}\frac{\sum_{G \in G_0 \cup G_1}\pi_{\bm{G}, \theta_{\bm{G}}}(G, \theta_G \mid \prec, D)}{\sum_{G \in G_0}\pi_{\bm{G}, \theta_{\bm{G}}}(G, \theta_G \mid \prec, D)}\\
    &=\frac{1}{4c^{*^2}}\frac{\sum_{G \in G_1 \cup G_2 \cup G_3}\pi_{\bm{G}, \theta_{\bm{G}}}(G, \theta_G \mid \prec, D)}{\sum_{G \in G_0}\pi_{\bm{G}, \theta_{\bm{G}}}(G, \theta_G \mid \prec, D)}
\end{align*}

Now turning to the right-hand side:
\begin{align*}
    \frac{b(\theta_{(\mathcal{H}^{+e_2}, \prec)}, \theta_{e_1})}{D(\theta_{(\mathcal{H}^{+e_2}, \prec)}, \theta_{e_1})}\frac{b(\theta_{(\mathcal{H}, \prec)}, \theta_{e_2})}{D(\theta_{(\mathcal{H}, \prec)}, \theta_{e_2})} &= \frac{1}{4c^{*^2}}\frac{\sum_{G \in G_1 \cup G_2 \cup G_3}\pi_{\bm{G}, \theta_{\bm{G}}}(G, \theta_G \mid \prec, D)}{\sum_{G \in G_0 \cup G_2}\pi_{\bm{G}, \theta_{\bm{G}}}(G, \theta_G \mid \prec, D)}\frac{\sum_{G \in G_0 \cup G_2}\pi_{\bm{G}, \theta_{\bm{G}}}(G, \theta_G \mid \prec, D)}{\sum_{G \in G_0}\pi_{\bm{G}, \theta_{\bm{G}}}(G, \theta_G \mid \prec, D)}\\
    &=\frac{1}{4c^{*^2}}\frac{\sum_{G \in G_1 \cup G_2 \cup G_3}\pi_{\bm{G}, \theta_{\bm{G}}}(G, \theta_G \mid \prec, D)}{\sum_{G \in G_0}\pi_{\bm{G}, \theta_{\bm{G}}}(G, \theta_G \mid \prec, D)}
\end{align*}
Thus, the chosen rates yield a well-defined posterior distribution on $\left((\mathcal{H},\prec), \theta_{(\mathcal{H},\prec)}\right)$. Integrating across all parameter values, we can define $B_e(\mathcal{H}, \prec)$, $D_e(\mathcal{H}, \prec)$ as follows:
\begin{align}
    B_e(\mathcal{H}, \prec) &\propto \frac{\pi_{DAG}(\mathcal{H}^{+e} \mid \prec, D)}{|\mathcal{G}_{\mathcal{H}^{+e}}|} = \frac{\sum_{G \in \mathcal{G}_{\mathcal{H}^{+e}}}\pi_{DAG}(G \mid \prec, D)}{|\mathcal{G}_{\mathcal{H}^{+e}}|}\label{eq:birth_rate_edge}\\
    D_e(\mathcal{H}, \prec) &\propto \frac{c^*\pi_{DAG}(\mathcal{H}^{-e} \mid \prec, D)}{|\mathcal{G}_{\mathcal{H}^{-e}}|} = \frac{c^*\sum_{G \in \mathcal{G}_{\mathcal{H}^{-e}}}\pi_{DAG}(G \mid \prec, D)}{|\mathcal{G}_{\mathcal{H}^{-e}}|}\label{eq:death_rate_edge}
\end{align}

Dividing \ref{eq:birth_rate_edge} and \ref{eq:death_rate_edge} by $\frac{\pi_{DAG}(\mathcal{H} | \prec, D)}{|\mathcal{G}_\mathcal{H}|}$ (which does not change the proportionality) yields the rates in theorem \ref{theorem:distribution_class}. In addition, since we consider graphs with finite dimension (i.e., $p < \infty$), these rates satisfy proposition 5.1 of \cite{preston_spatial_1977} and are finite, ensuring the validity of the detailed balance condition.

\subsection{Proof of Corollary \ref{theorem:posterior_form}}\label{app:posterior_form}

Without loss of generality, for a search space $\mathcal{H}$ with $|E_{\mathcal{H}}|=k$, define a sequence of search spaces $\{\mathcal{H}_0, \mathcal{H}_1, \dots, \mathcal{H}_k = \mathcal{H}\}$, where $\mathcal{H}_0 = (V, \emptyset)$ is the empty graph, and each $\mathcal{H}_i = \mathcal{H}_{i-1} \cup \{e_i\}$ with $e_i \in E_{\mathcal{H}}$. From \eqref{eq:recursive_definition}, the stationary density $f$ satisfies the recursive relation $f_{n+1}(x \cup \xi) = \gamma(x, \xi) f_n(x)$. 

Let $f_{c^*}((\mathcal{H}_0, \prec), \theta_{((\mathcal{H}_0, \prec))}) =\pi_{DAG}(\mathcal{H}_0 | \prec, D)$ (since $\mathcal{H}_0$ is the empty graph, $\pi_{DAG}(\mathcal{H}_0, \theta_{\mathcal{H}_0} | \prec, D)=\pi_{DAG}(\mathcal{H}_0 | \prec, D)$). By induction, extending the recursive relationship across all $k$ dimensions, $f_{c^*}(\theta_{(\mathcal{H}, \prec)})$ can then be written as the following:

\begin{align*}
    f_{c^*}((\mathcal{H}, \prec), \theta_{(\mathcal{H}, \prec)}) &= f_{c^*}((\mathcal{H}_0, \prec), \theta_{(\mathcal{H}_0, \prec)})\prod_{i=1}^{k} \gamma(((\mathcal{H}, \prec), \theta_{(\mathcal{H}_{i-1}, \prec)}), ((\mathcal{H}^{+e_i}, \prec), \theta_{e_i}))\\
    &= \pi_{DAG}(\mathcal{H}_0 | \prec, D)\prod_{i=1}^{k}\frac{b(\theta_{(\mathcal{H}_{i-1}, \prec)}, \theta_{e_i})}{D(\theta_{(\mathcal{H}_{i-1}, \prec)}, \theta_{e_i})}\\
    &= \pi_{DAG}(\mathcal{H}_0 | \prec, D)\prod_{i=1}^{k}\frac{1}{2c^*}\frac{\sum_{G \in \mathcal{G}_{\mathcal{H}_i}}\pi_{\mathbf{G}, \Theta_{\mathbf{G}}}(G, \theta_G | \prec, D)}{\sum_{G \in \mathcal{G}_{\mathcal{H}_{i-1}}}\pi_{\mathbf{G}, \Theta_{\mathbf{G}}}(G, \theta_G | \prec, D)}\\
    &= \pi_{DAG}(\mathcal{H}_0 | \prec, D)(2c^*)^{-k}\frac{\sum_{G \in \mathcal{G}_{\mathcal{H}_k}}\pi_{\mathbf{G}, \Theta_{\mathbf{G}}}(G, \theta_G | \prec, D)}{\pi_{DAG}(\mathcal{H}_0 | \prec, D)}\\
    &= (2c^*)^{-|E_{\mathcal{H}}|}\sum_{G \in \mathcal{G}_{\mathcal{H}}}\pi_{\mathbf{G}, \Theta_{\mathbf{G}}}(G, \theta_G | \prec, D)
\end{align*}

Integrating out $\theta_{(\mathcal{H}, \prec)}$ makes $f_{c^*}(\mathcal{H}, \prec) = (2c^*)^{-|E_{\mathcal{H}}|}\sum_{G \in \mathcal{G}_{\mathcal{H}_i}}\pi_{DAG}(G| \prec, D)$

Summing over all possible orders $\prec \in \mathbb{S}^p$ to obtain the marginal density for $\mathcal{H}$:

\begin{align}
    f_{c^*}(\mathcal{H}) = \sum_{\prec \in \mathbb{S}^p} f_{c^*}(\mathcal{H}, \prec) = \sum_{\prec \in \mathbb{S}^p} (2c^{*})^{-|E_{\mathcal{H}}|}\pi_{DAG}(\mathcal{H} | \prec, D)
\end{align}

Using \eqref{eq:recursive_definition},
\begin{align}
    \mu_{Space}^{(c^*)}(\mathcal{H}) \propto (2c^{*})^{-|E_{\mathcal{H}}|}\sum_{\prec \in \mathbb{S}^p} \pi_{DAG}(\mathcal{H} | \prec, D).
\end{align}

\subsection{Proof of Lemma \ref{theorem:mixture_kernel_bdmcmc}}\label{app:mixture_kernel_bdmcmc}
Consider $\mu$ defined according to \ref{eq:recursive_distribution}. Since it is recursively defined such that $\frac{\mu(x \cup \xi)}{\mu(x)} = \frac{f(x \cup \xi)}{f(x)} = \gamma(x, \xi) = \frac{b(x, \xi)}{D(x; \xi)}$, we can directly plug this in to the density ratio term in \ref{eq:acceptance_ratio}, $\frac{g(x \cup \xi)}{g(x)}$. 

Following the setup of sampling from the birth-death process, we define $\mathsf{b}(x, \xi) = \frac{b(x, \xi)}{\beta(x)}$ (in order to have $\mathsf{b}$ be a density), similarly define $\mathsf{D}(x; \xi) = \frac{D(x;\xi)}{\delta(x)}$, and set $q(x) = \frac{\beta(x)}{\beta(x) + \delta(x)}$, and $1-q(x) = \frac{\delta(x)}{\beta(x)+\delta(x)}$ (which is compatible with the birth-death process transition kernel, as shown in equation \ref{eq:trans_kernel_bd}). We get the following simplification of \ref{eq:acceptance_ratio}:

\begin{align*}
    r(x, \xi) &= \frac{g(x \cup \xi)}{g(x)}\frac{1-q(x \cup \xi)}{q(x)}\frac{\mathsf{D}(x; \xi)}{\mathsf{b}(x, \xi)}\\
    &= \frac{b(x, \xi)}{D(x;\xi)} \frac{\frac{\delta(x \cup \xi)}{\beta(x \cup \xi)+\delta(x \cup \xi)}}{\frac{\beta(x)}{\beta(x) + \delta(x)}}\frac{\frac{D(x; \xi)}{\delta(x \cup \xi)}}{\frac{b(x, \xi)}{\beta(x)}}\\
    &= \frac{\frac{1}{\beta(x \cup \xi) + \delta(x \cup \xi)}}{\frac{1}{\beta(x) + \delta(x)}} = \frac{w(x \cup \xi)}{w(x)}
\end{align*}

\section{Implementation}

\subsection{Updating the Plus-One Banned Matrix}\label{app:plone_update}
\subsubsection{Background}\label{app:score_table_bkgrd}
In the efficient scoring techniques proposed in \cite{kuipers_efficient_2021} and \cite{viinikka_towards_2020}, it is standard to create a ``banned'' score table tabulating all of the valid scores that are compatible with particular ``banned'' parent possibilities (to allow constant-time lookup of valid parent-set scores under order constraints). Consider a node $i$ with $pa_\mathcal{H}(i)$ its allowed parents under the search space $\mathcal{H}$. Then, the past work would first create a $2^{|pa_\mathcal{H}(i)|} \times 1$ table of scores, as well as a table of ``banned'' scores of the same dimensions. Table \ref{table:score_tables} illustrates these tables for a node $i$ with 3 possible parents in $\mathcal{H}$ (marked by $h_1^i, h_2^i, h_3^i$). Scoring orders in the restricted search scoring problem then amounts to identifying the banned parent configuration implied by the order for particular node $i$, which can be found by intersecting $pa_{\mathcal{H}}(i)$ with $\left(\cup_{j=\prec_{[i]}}^{p} \prec_j\right)$. Algorithmically, to evaluate an order's score in $O(Kp)$, one can encode each row in the table to a unique integer, where the binary mapping of the integer describes the allowed and banned parents. For instance, $4$ in binary is $100$, indicating that $S_4^i$ in the left table \ref{table:score_tables} corresponds to the score of just including the third parent $h_3^i$ whereas $5$ is $011$, corresponding to the score of including $h_1^i$ and $h_2^i$; the banned equivalent can be found by subtracting the integer from $2^{|pa_{\mathcal{H}}(i)|}-1$. 

\begin{table}[ht!]
\centering
\begin{tabular}{|c|c||c|c|}
\hline
\textbf{Parents} & \textbf{Parent score} & \textbf{Banned parents} & \textbf{Banned parent score} \\
\hline
$\emptyset$ & $S_0^i = S(X_i, \{\emptyset\} \mid D)$ & $\emptyset$ & $B_7^i = S_0^i + \dots + S_7^i$ \\
\hline
$h_1^i$ & $S_1^i = S(X_i, \{h_1^i\} \mid D)$ & $h_1^i$ & $B_6^i = S_0^i + S_2^i + S_4^i + S_6^i$ \\
\hline
$h_2^i$ & $S_2^i = S(X_i, \{h_2^i\} \mid D)$ & $h_2^i$ & $B_5^i = S_0^i + S_1^i + S_4^i + S_5^i$ \\
\hline
$h_3^i$ & $S_4^i = S(X_i, \{h_3^i\} \mid D)$ & $h_3^i$ & $B_3^i = S_0^i + S_1^i + S_2^i + S_3^i$ \\
\hline
$h_1^i, h_2^i$ & $S_3^i = S(X_i, \{h_1^i, h_2^i\} \mid D)$ & $h_1^i, h_2^i$ & $B_4^i = S_0^i + S_4^i$ \\
\hline
$h_1^i, h_3^i$ & $S_5^i = S(X_i, \{h_1^i, h_3^i\} \mid D)$ & $h_1^i, h_3^i$ & $B_2^i = S_0^i + S_2^i$ \\
\hline
$h_2^i, h_3^i$ & $S_6^i = S(X_i, \{h_2^i, h_3^i\} \mid D)$ & $h_2^i, h_3^i$ & $B_1^i = S_0^i + S_1^i$ \\
\hline
$h_1^i, h_2^i, h_3^i$ & $S_7^i = S(X_i, \{h_1^i, h_2^i, h_3^i\} \mid D)$ & $h_1^i, h_2^i, h_3^i$ & $B_0^i = S_0^i$ \\
\hline
\end{tabular}
\caption{An example score table (left) and ``banned'' score table (right) as shown in \cite{kuipers_efficient_2021} for a node $i$ with 3 possible parents in $\mathcal{H}$.}\label{table:score_tables}
\end{table}

To expand the space for node $i$ by one parent, the past work makes two-dimensional versions of the previous score tables by adding exactly one arbitrary parent outside of $\mathcal{H}$ in a unique column (resulting in $2^{|pa_\mathcal{H}(i)|} \times (p-|pa_\mathcal{H}(i)|-1)$ sized tables). Table \ref{table:plus_scores_table} shows an example of the ``plus-one'' parent node table, for the same node $i$ as table \ref{table:score_tables}. Each column corresponds to adding one possible parent not in $\mathcal{H}$, where in this case, the graphs are in $\mathbb{G}_7$ (so there are three ``plus-one'' candidates to consider). 

\begin{table}[ht!]
\centering
\begin{tabular}{|c|c|c|c|}
\hline
\textbf{Parents} & \textbf{$+h_4^i$} & \textbf{$+h_5^i$} & \textbf{$+h_6^i$} \\
\hline
$\emptyset$ & $S(X_i, \{\emptyset, h_4^i\} \mid D)$ & $S(X_i, \{\emptyset, h_5^i\} \mid D)$ & $S(X_i, \{\emptyset, h_6^i\} \mid D)$ \\
\hline
$h_1^i$ & $S(X_i, \{h_1^i, h_4^i\} \mid D)$ & $S(X_i, \{h_1^i, h_5^i\} \mid D)$ & $S(X_i, \{h_1^i, h_6^i\} \mid D)$ \\
\hline
$h_2^i$ & $S(X_i, \{h_2^i, h_4^i\} \mid D)$ & $S(X_i, \{h_2^i, h_5^i\} \mid D)$  & $S(X_i, \{h_2^i, h_6^i\} \mid D)$  \\
\hline
$h_3^i$ & $S(X_i, \{h_3^i, h_4^i\} \mid D)$ & $S(X_i, \{h_3^i, h_5^i\} \mid D)$ & $S(X_i, \{h_3^i, h_6^i\} \mid D)$ \\
\hline
$h_1^i, h_2^i$ & $S(X_i, \{h_1^i, h_2^i, h_4^i\} \mid D)$ & $S(X_i, \{h_1^i, h_2^i, h_5^i\} \mid D)$ & $S(X_i, \{h_1^i, h_2^i, h_6^i\} \mid D)$ \\
\hline
$h_1^i, h_3^i$ & $S(X_i, \{h_1^i, h_3^i, h_4^i\} \mid D)$ & $S(X_i, \{h_1^i, h_3^i, h_5^i\} \mid D)$ & $S(X_i, \{h_1^i, h_3^i, h_6^i\} \mid D)$ \\
\hline
$h_2^i, h_3^i$ & $S(X_i, \{h_2^i, h_3^i, h_4^i\} \mid D)$ & $S(X_i, \{h_2^i, h_3^i, h_5^i\} \mid D)$ & $S(X_i, \{h_2^i, h_3^i, h_6^i\} \mid D)$ \\
\hline
$h_1^i, h_2^i, h_3^i$ & $S(X_i, \{h_1^i, h_2^i, h_3^i, h_4^i\} \mid D)$ & $S(X_i, \{h_1^i, h_2^i, h_3^i, h_5^i\} \mid D)$ & $S(X_i, \{h_1^i, h_2^i, h_3^i, h_6^i\} \mid D)$ \\
\hline
\end{tabular}
\caption{An example ``plus'' score table in $\mathbb{G}_7$ for a node $i$ with 3 possible parents in $\mathcal{H}$. Rows correspond to base parent sets in $\mathcal{H}$; columns correspond to the added parent.}\label{table:plus_scores_table}
\end{table}

From there, a ``banned'' plus score table can be created by matching the appropriate rows of the plus score table to specific banned parent sets, in the same way displayed in table \ref{table:score_tables}. When scores are stored in log-space for numerical stability, computing the ``banned'' score requires the use of the LogSumExp (LSE) function:
\begin{align*}
    LSE(x_1,...,x_m) &= \log(\exp(x_1)+...+\exp(x_m))
\end{align*}

In previous implementations of restricted order MCMC, BGe \citep{heckerman_learning_1995} score tables take $O(K^3p2^K)$ to compute, BGe plus score tables take $O(K^3p2^K(p-K-1))$ to compute (by repeating the score table process for each of the $(p-|pa_\mathcal{H}(i)|-1)$ excluded nodes from $\mathcal{H}$, and the ``banned'' score table takes $O(2^Kp)$ to compute. Creating the plus ``banned'' score table is also $O(2^Kp)$, as the row-matching process to create the $B^i$ values illustrated in table \ref{table:score_tables} can be applied to each of the $(p-|pa_\mathcal{H}(i)|-1)$ columns of the plus score table in a vectorized implementation.

As mentioned in Section \ref{subsec:comp_considerations}, BROOD only performs a transdimensional update step on one node as opposed to all nodes like the previous work. As a result, all BROOD transdimensional updates eliminate a factor of $p$ by construction.

Below, we discuss two further ways to speed up the implementation for our birth-death version of the restricted order MCMC sampler.

\subsubsection{Efficient Expansions for BGe score via Block Matrix Operations}\label{app:plone_bge_score}

We aim to reduce the largest computational cost: constructing the plus score table for the BGe score. A naive approach yields a cost of $O(K^32^K(p-K-1))$, where the $K^3$ terms comes from the cubic computational cost of finding a determinant, as required in the BGe score. The total cost arises from filling all entries of a single node’s plus score table (using decomposable scores). There are $2^K$ rows per table (for the valid subsets of parent sets), and $p - K - 1$ columns (for each of the excluded possible parents). Previous literature calculates the BGe score directly for each of the $2^K(p-K-1))$ entries of the plus score table, resulting in the full cubic cost per cell.

We reduce the cost by applying intermediate matrix calculations to the plus-one scores via rank-one update formulas. This allows us to avoid re-calculating the BGe score for the plus-one parent sets by reusing the Cholesky factorization of the base parent set (reducing the number of BGe scores to directly calculate to the $2^K$ original scores instead of all $2^K \times (p-K-1)$ entries). Paired with vectorization across all $p-K-1$ candidate excluded parents, we generate all plus-one scores per row (as illustrated in table \ref{table:plus_scores_table}) in a batch. 

Specifically, the BGe score for a parent set $pa(i)$ on a node $i$ has two components: a normalizing constant that depends on prior parameters, and a data-dependent term that involves a ratio of two determinants raised to different powers. The two terms in this ratio are:
\begin{enumerate}
    \item $|R_{pa(i), pa(i)}|^{(\alpha^* - p + |pa(i)|)/2}$, and
    \item $|R_{pa(i) \cup \{i\}, pa(i) \cup \{i\}}|^{(\alpha^* - p + |pa(i)|+1)/2}$
\end{enumerate}
where $R$ is the scale matrix in the posterior BGe score, and $\alpha^*$ is the degrees of freedom parameter in the posterior BGe score. Previous work \citep{kuipers_addendum_2014} implemented the ratios efficiently by expressing the reciprocal of the ratio of determinants in terms of a Schur complement, and using a modified Cholesky decomposition to solve for the appropriate determinants. These computations can be performed on the log-scale by converting multiplicative components to sums and additive components to LSEs.

We calculate determinants on the plus-one parent sets by using the rank-one update formulas \citep{osborne_bayesian_2010} for Cholesky decompositions. This reduces the cost per calculation from $O(K^3)$ to $O(K^2)$. The rank-one updates can also be vectorized across all possible plus-one sets for a specific parent set, causing the total cost of a row of the plus score table to be $O(K^2)$ rather than $O(K^3(p-K-1))$ (and the overall cost to be $O(K^22^K)$ instead of $O(K^32^K(p-K-1))$). Lowering the computational burden of the plus score table makes the original score table (as illustrated on the left-hand side of table \ref{table:score_tables}) the main computational bottleneck. In our experiments for an 80-node problem with a maximal allowed sparsity of 18, this optimization yields a $\sim 30\times$ speedup over the naive implementation on creating the score table and plus score table.

\subsubsection{Efficient Contractions for Decomposable Scores via Memoization}\label{app:plone_memoize}
In contraction steps, an edge is removed from the search space, so the score tables must be updated to reflect this change. However, since contraction yields a subset of the previous parent configurations, no new score calculations are needed. In fact, all plus scores and banned plus scores on the contracted space can be taken directly from the existing score tables, or can be found from them with simple arithmetic operations.

Suppose we contract the space by 1 parent for $i$ in $\mathcal{H}$ where we remove a parent $j$ from $pa(i)$. We will have new 2-dimensional plus tables that are $2^{|pa_\mathcal{H}(i)|-1} \times (p-|pa_\mathcal{H}(i)|)$ instead of $2^{|pa_\mathcal{H}(i)|} \times (p-|pa_\mathcal{H}(i)|-1)$. To construct the plus score table, we reuse scores from the original table to populate the added column (corresponding to $+h^i_j$) by mapping the appropriate $2^{|pa_\mathcal{H}(i)|-1}$ removed entries to the appropriate row in the new table. This is straightforward to do using the same mapping system as previous literature \citep{kuipers_efficient_2021}.

It takes more care to populate the new banned score table. Any entry that previously banned the removed parent $j$ remains valid and retains its relative position in the new table. For example, if $h^i_2$ is removed as a potential parent in table \ref{table:score_tables}, then the $B^i_5, B^i_4, B^i_1$, and $B^i_0$ banned scores -- which correspond to the subsets not including $h^i_2$ -- are preserved (and this will also be true for the plus-banned columns in the plus-banned score table). In addition, the retained entries are already ordered correctly from original versions of the tables (so in table \ref{table:score_tables}, the ordering of $B^i_5, B^i_4, B^i_1, B^i_0$ is appropriate for an updated contracted version of the table).

To populate the new column for re-adding $j$ in the plus-banned score table, we must subtract each retained score in the contracted banned score table from the removed supersets from the original table. For instance, continuing the example of re-adding $h^i_2$, its first entry of the $+h^i_2$ column will be $B^i_7 - B^i_5$, as this contains all graphs that would ban no parents if we force $h^i_2$ to be a parent (and the other 3 entries in the column will be $B^i_6 - B^i_4$, $B^i_3-B^i_1$, and $B^i_2-B^i_0$). When working with scores on the log-scale, this requires us to use the inverse of the LSE, which we define as LogMinusExp (LME):
\begin{align}
    \text{Suppose } LSE(a,b) &= c,\nonumber\\
    \text{then } LME(c,b) &= a, LME(c, a) = b\label{eq:lme}
\end{align}

To ensure numerical stability of the LME, especially for small numbers, it is advisable to use the \texttt{log1p} and \texttt{expm1} functions in R (or equivalents in other languages) as guided by \cite{machler_log1mexp}.

Since all values in a contracted table are derived from its pre-contracted counterpart, we can use memoization to update the tables in $O(K)$ time. As a result, updating the space to reflect a contraction does not add computational overhead to the birth-death restricted search sampler.

\section{Simulation Details}

\subsection{Setup}\label{app:sim_setup}
As mentioned in Section \ref{subsec:sim_setup}, we consider linear acyclic causal models. Specifically, they all use the following structure:
\begin{align*}
    X_i &= e_i + \sum_{X_j \in \bm{Pa}(G)_i} B_{ji}X_j\\
    B_{ji} &\sim \text{Uniform}(0.4, 2)\\
    G &\sim \text{Graph Probability Model of Choice}\\
    e_i &\sim \text{Data Error Model of Choice}
\end{align*}

\begin{table}[ht!]
\centering
\small
\setlength{\tabcolsep}{6pt}
\renewcommand{\arraystretch}{1.15}
\begin{tabular}{p{0.36\linewidth} p{0.58\linewidth}}
\hline
\multicolumn{2}{c}{\textbf{Configuration Components}} \\
\hline

Nodes ($p$) 
& $\{20,\;80,\;140,\;200\}$ \\

Sample size ($n$) 
& $n \in \{0.5p,\;2p,\;10p\}$ \\

Graph model for $G$ 
& Erd\H{o}s--R\'enyi (ER); \\
& Stochastic Block Model (SBM) (2 blocks); \\
& Hierarchical SBM with 3 clusters (hSBM) (1-, 2-, and 3-block structure) \\

Error model for $e_i$ 
& Gaussian noise: $N(0,1)$; \\
& Functional causal model with mixture error: \\
& $0.5\,N(0,1) + 0.5\,N(0,2)$ \\

Sparsity regime 
& (i) $1 + K_{\text{init}}$; (ii) Fixed sparsity across all nodes\textsuperscript{\textdagger} \\

Initial search space 
& PC-algorithm based; \\
& GES-based\textsuperscript{\textsection}; \\
& Kuipers et al.\ iterative expansion \\

Random seeds 
& 25 per configuration \\

MCMC replicates
& 4 chains for $p \in \{20,80\}$; 1 chain for $p \in \{140,200\}$ \\

\hline
\multicolumn{2}{c}{\textbf{Outputs}} \\
\hline

DAG inference outputs 
& BROOD; BiDAG; BROOD${+}$ \\

Skeleton inference outputs 
& BROOD; BiDAG; BDgraph; BROOD${+}$ \\

Evaluation Metrics & ROC AUC; PR AUC; $Pr^{+}$; $Pr^{-}$; Run-Time\textsuperscript{\maltese} \\

\hline
\end{tabular}
\caption{
Summary of synthetic-data experiments.
All data are generated according to the linear acyclic causal model
$X_i = e_i + \sum_{X_j \in \bm{Pa}(G)_i} B_{ji} X_j$.
}
\label{table:experiment_summary}
\vspace{1mm}
\raggedright
\footnotesize
\textsuperscript{\textdagger}Fixed-sparsity experiments were not conducted for $p=200$.\\
\textsuperscript{\textsection}GES-based initialization was only used for fixed-sparsity experiments.\\
\textsuperscript{\maltese}$Pr^{+}$ and $Pr^{-}$ are the estimated mean true edge probability, and mean non-edge probability, respectively.
\end{table}

Table \ref{table:experiment_summary} describes the holistic set of synthetic-data experiments. 

\subsubsection{Graph Scores}
The Bayesian Gaussian equivalence (BGe) score \citep{heckerman_learning_1995} and G-Wishart score \citep{aliye_montecarlo_2005} were used for BiDAG and BDgraph, respectively. The two primary difference between these scores are (1) the BGe score is defined for DAGs while the G-Wishart is a score for undirected graphs, and (2) the BGe score utilizes an unconstrained Wishart prior on the precision matrix, whereas the G-Wishart score employs a constrained prior. Specifically, the G-Wishart distribution is restricted to the subspace of precision matrices that have the same sparsity as the graph (so if an edge $j-i$ is not in the graph, then its score will only consider the set of precision matrices with a 0 at index $(j,i)$).

\subsubsection{Graph Model Setup}

For the Erd\H{o}s--R\'enyi experiments, we use the \texttt{randDAG} function from the \texttt{pcalg} package in R to sample the  model, specifying the expected node degree to be 4. This amounts to a connection probability of $\sim0.21$ for 20 nodes, and $\sim0.02$ for 200 nodes. 

For Stochastic Block Model (SBM) and Hierarchical Stochastic Block Model (hSBM) experiments, we respectively employ the \texttt{sample\_sbm} and \texttt{sample\_hierarchical\_sbm} functions from the \texttt{igraph} package in R to sample directed graphs. To ensure acyclicity, we then sample a random topological order of the nodes and retain only those edges consistent with the order. 

The SBM consists of two blocks, with node assignment drawn independently, with probability 0.8 of block 1 and 0.2 of block 2. Conditional on block membership, directed edges are generated independently according to the following connection probability matrix:
\begin{align}
    C_{SBM} &= \begin{bmatrix}
        \min(0.15, \frac{6}{p}) & \min(0.01, \frac{2}{p}) \\
        \min(0.01, \frac{2}{p}) & \min(0.08, \frac{4}{p})
    \end{bmatrix}\label{eq:sbm_connect_matr}
\end{align}
This model creates two blocks where within-block connections are likelier than between-block connections. 

The hSBM introduces a richer community structure by partitioning the nodes into three top-level clusters, with cluster sizes proportional to $(0.1, 0.3, 0.6)$. Each cluster is generated according to unique block models:
\begin{enumerate}
    \item A single-block cluster with $\min(0.1, \frac{4}{p})$ connection probability
    \item A two-block cluster with $(\frac{1}{3}, \frac{2}{3})$ probability of assignment and the same connection probability matrix as the SBM in~\eqref{eq:sbm_connect_matr}
    \item A three-block cluster with $(\frac{1}{6}, \frac{1}{3}, \frac{1}{2})$ probability of assignment and a connection probability matrix of 
    $\begin{bmatrix}
        \min(0.2, \frac{8}{p}) & \min(0.01, \frac{1}{p}) & \min(0.2, \frac{8}{p}) \\ 
        \min(0.01, \frac{1}{p}) & 0 & \min(0.08, \frac{4}{p}) \\ 
        \min(0.2, \frac{8}{p}) & \min(0.08, \frac{4}{p}) & \min(0.08, \frac{4}{p})
    \end{bmatrix}$
\end{enumerate}
The third cluster introduces more complex relationships both within and between blocks; while blocks 1 and 3 within this setup exhibit relatively strong within- and between-block connectivity, block 2 enforces near-deterministic separation within a set of nodes by assigning zero probability to within-block connections, but allows for relatively strong connectivity to nodes within block 3. This hierarchical design yields graphs with both community and anti-community structure.

\subsubsection{Initial Search Space Setup}

To create all search spaces, we impose a fixed sparsity capacity used across all experiments (to ensure our initial search space does not exceed our computation budget). For the plus-one experiments, this amounts to a sparsity capacity of $\max(10, 3+0.05p)$, and for the fixed sparsity experiments, this amounts to a sparsity capacity of $\max(12, \text{round}(3+0.06p ))$.

To create the search space from the PC algorithm, we use the \texttt{skeleton} function from the \texttt{pcalg} package in R, using Gaussian independence hypothesis tests, and a hypothesis testing threshold of $\min
\left(0.4, \frac{20}{p}\right)$.

We use the \texttt{iterativeMCMC} function from \texttt{BiDAG} to run the iterative expansion procedure proposed in \cite{kuipers_efficient_2021}. In addition to imposing the maximal fixed sparsity used across all experiments as the hard limit for parent set sizes, we also use a soft limit of 10, where with fewer than 10 candidate parents, we include all edges in the undirected skeleton of the MAP DAG (so if $i \to j$ is an edge in the search space, so will $j \to i$), but beyond this, only edges in the MAP DAG are included. We use the sample hypothesis testing threshold as we used for the PC algorithm to initialize this iterative procedure.

To create the GES search space, we use the \texttt{ges} function from the \texttt{pcalg} package, using the BIC score (i.e., using an $L_0$-penalized Gaussian maximum likelihood estimator, and a penalty parameter of $\log(n)$).

We used the default values for all other parameters for the functions we described above.

\subsubsection{MCMC Setup}

We employ the \texttt{orderMCMC} function in \texttt{BiDAG} and the \texttt{bdgraph} function in \texttt{BDgraph} to run BiDAG and BDgraph, respectively.

We run all chains for $B= \min(25000, \lceil p^2\log(p) \rceil)$ steps, with an additional warm-up of $\lfloor 0.1B \rfloor$ steps. For BROOD, we use the defaults of $c^* = 1$ and $\ell = 0.1$ for all experiments, as discussed in section \ref{subsec:param_choices}. To analyze estimated edge probabilities, we thin BROOD and BiDAG to keep 2500 samples; BDgraph returns edge probabilities directly from running it without a thinning parameter available, so we keep all of its samples in edge probability estimates.

We estimate all edge probabilities directly from the graph samples of an algorithm. Suppose $G^{(1)},...,G^{(B)}$ are the set of graph samples across an MCMC chain under model $\mathcal{M}$. Then, the estimated edge probability of $i \to j$ under $\mathcal{M}$ is:
\begin{align*}
    \widehat{\Pr}_{\mathcal{M}}(i \to j) &= \frac{1}{B}\sum_{b=1}^B \mathbbm{1}_{i \to j \in E_{G^{(b)}}}
\end{align*}
where $E_G$ is the edge set of a graph $G$.

For estimating the skeleton edge probability from DAG samples, we modify the formula:
\begin{align*}
    \widehat{\Pr}_{\mathcal{M}}(i-j) &= \frac{1}{B}\sum_{b=1}^B \Big(\mathbbm{1}_{i \to j \in E_{G^{(b)}}} + \mathbbm{1}_{j \to i \in E_{G^{(b)}}}\Big)
\end{align*}

Let $G_{true}$ be the true DAG used to generate the data, $\mathcal{E}^+ := \{(i,j): i\to j \in E_{G_{true}}\}$, and $\mathcal{E}^- := \{(i,j): i\to j \notin E_{G_{true}}\}$ (and for skeleton analysis, we redefine $\mathcal{E}^+ := \{(i,j): (i\to j \in E_{G_{true}}) \cup (j\to i \in E_{G_{true}})\} $, and analogously redefine $\mathcal{E}^{-}$). Then, from $\widehat{\Pr}_{\mathcal{M}}$, we output the Receiver Operator Characteristic Area Under the Curve (ROC AUC), the Precision-Recall Area Under the Curve (PR AUC), as well as $Pr^{+}$ and $Pr^{-}$, which are defined as:
\begin{align*}
    Pr^{+}(\mathcal{M}) &= \frac{1}{|\mathcal{E}^+|}
\sum_{(i,j) \in \mathcal{E}^+}
\widehat{\Pr}_{\mathcal{M}}(i \to j)\\
    Pr^{-}(\mathcal{M}) &= \frac{1}{|\mathcal{E}^-|}
\sum_{(i,j) \in \mathcal{E}^-}
\widehat{\Pr}_{\mathcal{M}}(i \to j)
\end{align*}

To compute ROC AUC and PR AUC, we treat each node pair $(i,j)$ as a binary instance, with assignment of 1 when $(i,j) \in \mathcal{E}^+$, and 0 when $(i,j) \in \mathcal{E}^-$. We then compute ROC and precision–recall curves by thresholding $\widehat{\Pr}_{\mathcal{M}}$ across all $(i,j)$ and report the corresponding areas under the curve.

\subsection{Further Results}\label{app:further_sims}

\subsubsection{Performance: Gaussian SEM, Plus-One Runs, DAG Output}\label{app:gauss_plus1_dag}

In the main paper, Figures \ref{fig:auc_roc_gauss_dag} and \ref{fig:auc_pr_gauss_dag} evaluate model quality on the DAG output for Gaussian SEM data, with plus-one sparsity capping for BROOD and BiDAG. Figure \ref{fig:prplus_gauss_dag_plus1} shows the models' mean edge probability ($Pr^+$) and Figure \ref{fig:prminus_gauss_dag_plus1} shows the mean non-edge probability ($Pr^-$) for these data. Similar to the findings in the main text, BiDAG performs well relative to BROOD in situations where using a single restricted fixed search space that is found by an iterative MAP procedure can be effective, like when there are few nodes, when there is a lot of data relative to the number of nodes, and when the graph model is expected to yield uni-modal posterior scores. BROOD performs better as the number of nodes increases, as the ratio between the number of observations and the number of nodes decreases, and when the posterior scores have a higher chance of being multi-modal.

\begin{figure}
    \centering
    \includegraphics[width=0.9\linewidth]{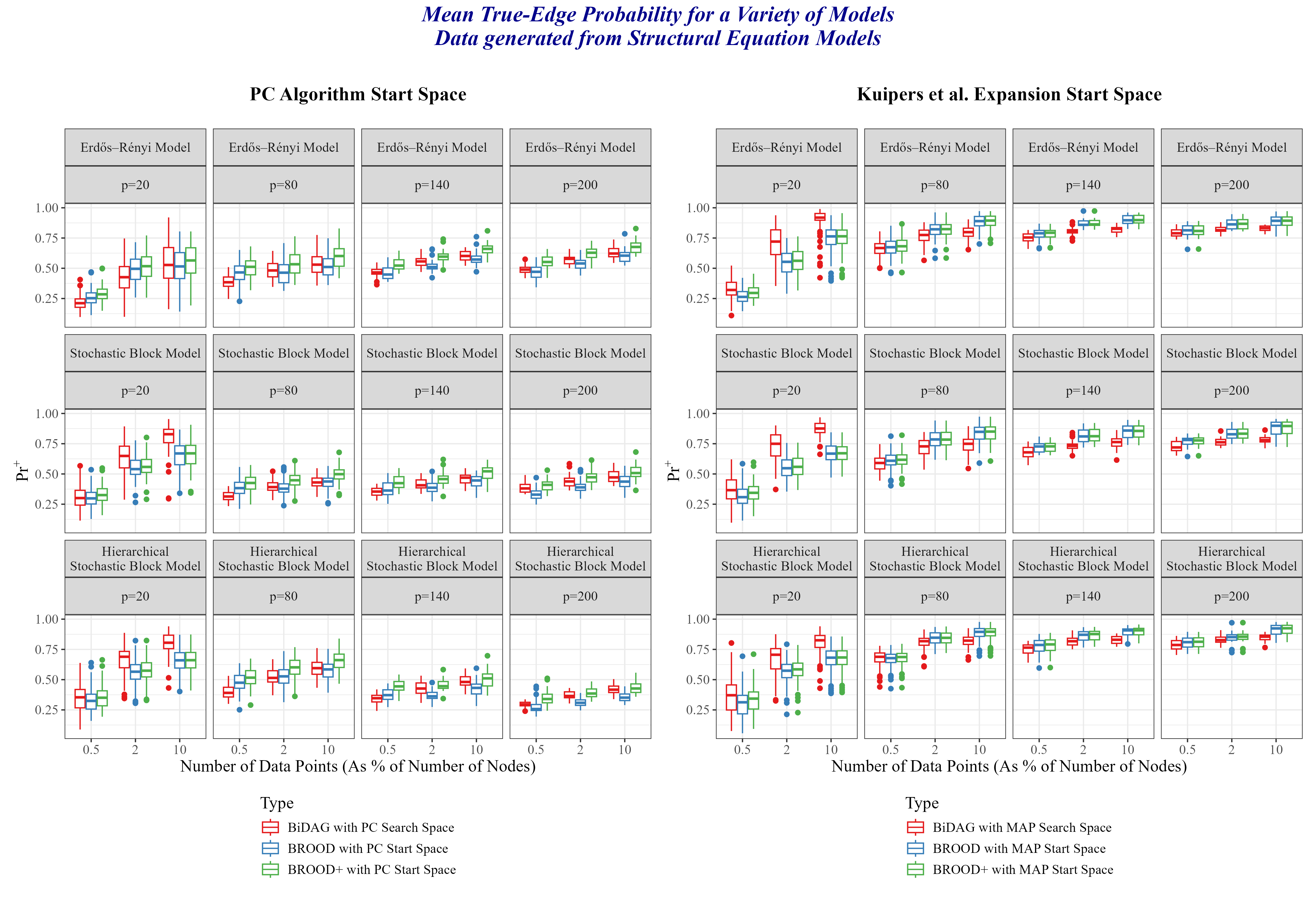}
    \caption{$Pr^+$ results with Gaussian SEM data, using plus-one sparsity.}
    \label{fig:prplus_gauss_dag_plus1}
\end{figure}

\begin{figure}
    \centering
    \includegraphics[width=0.9\linewidth]{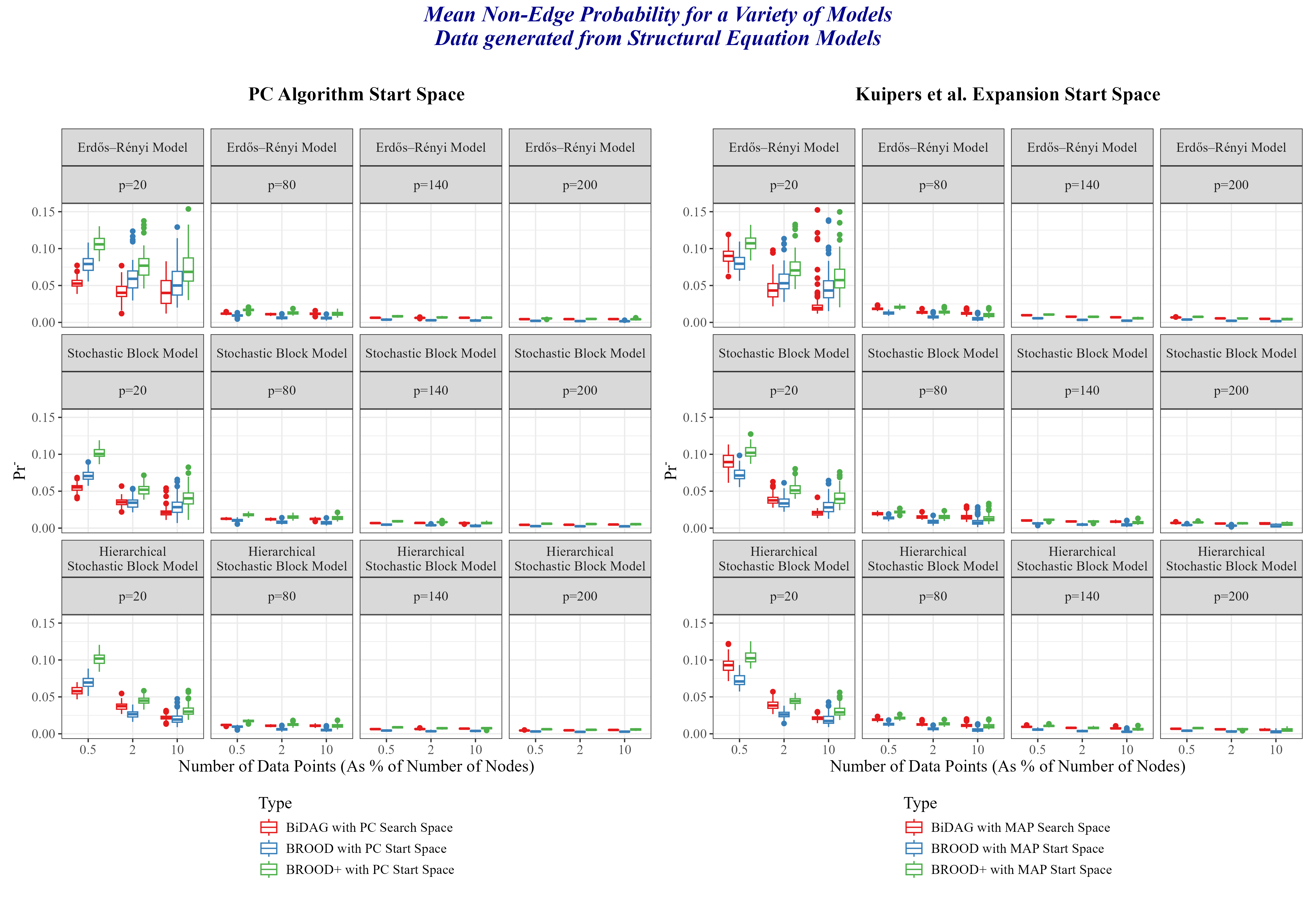}
    \caption{$Pr^-$ results with Gaussian SEM data, using plus-one sparsity.}
    \label{fig:prminus_gauss_dag_plus1}
\end{figure}

\subsubsection{Performance: Gaussian SEM, Plus-One Runs, Skeleton Output}

In the main paper, Figure \ref{fig:auc_roc_gauss_skel} showed the ROC AUC for the skeleton version of the Gaussian SEM data generation model with plus-one sparsity capping for BROOD and BiDAG. Figures \ref{fig:auc_pr_gauss_skel_plus1}, \ref{fig:prplus_gauss_skel_plus1}, \ref{fig:prminus_gauss_skel_plus1} show the PR AUC, $Pr^+$, and $Pr^-$ for these data. The trends echo the findings from Sections \ref{subsec:sim_results}, and \ref{app:gauss_plus1_dag}: BiDAG performs comparatively better in easier sampling regimes, while BROOD performs comparatively better in harder sampling regimes.

\begin{figure}
    \centering
    \includegraphics[width=0.9\linewidth]{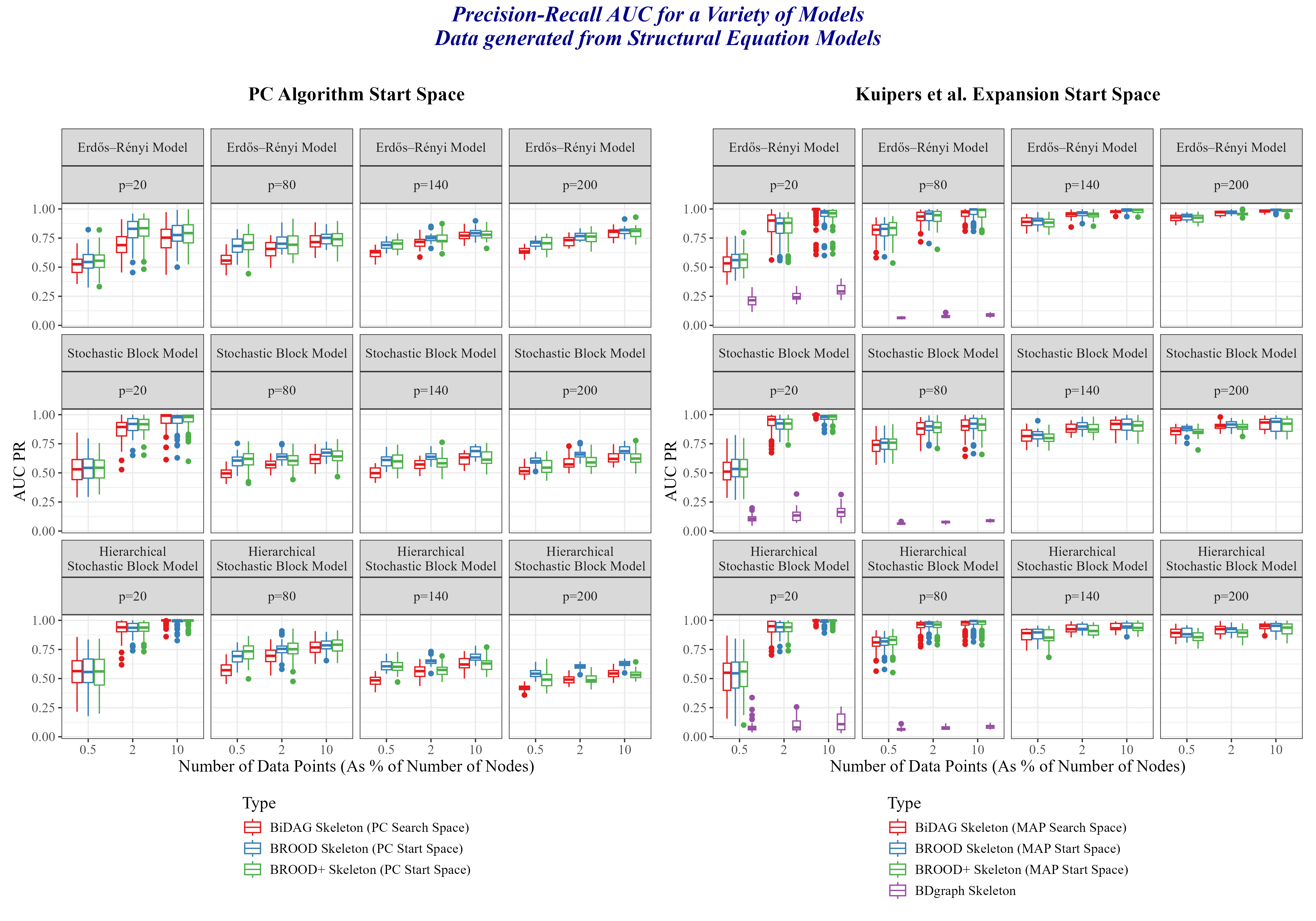}
    \caption{Skeleton version of PR AUC results with Gaussian SEM data, using plus-one sparsity.}
    \label{fig:auc_pr_gauss_skel_plus1}
\end{figure}

\begin{figure}
    \centering
    \includegraphics[width=0.9\linewidth]{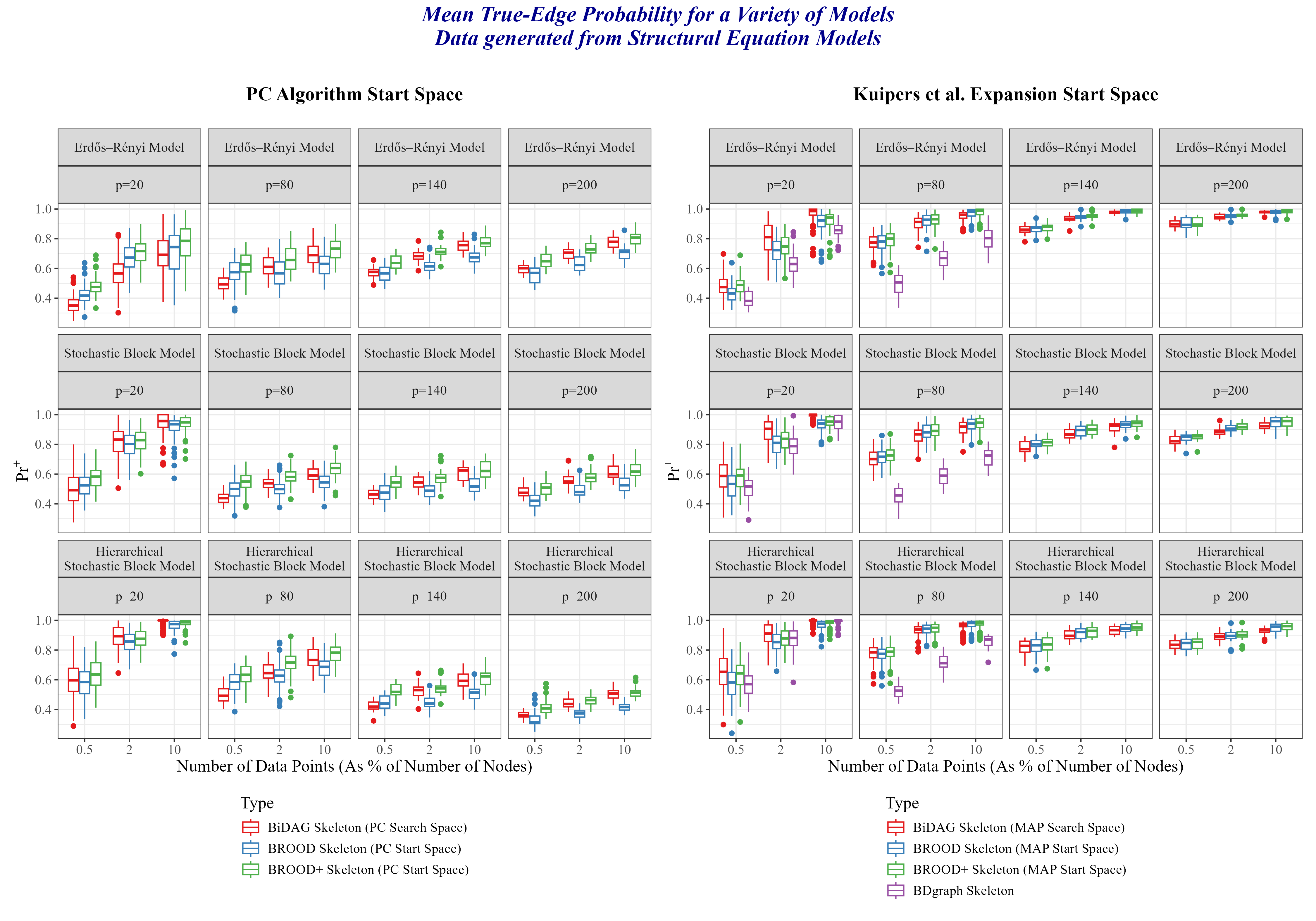}
    \caption{Skeleton version of $Pr^+$ results with Gaussian SEM data, using plus-one sparsity.}
    \label{fig:prplus_gauss_skel_plus1}
\end{figure}

\begin{figure}
    \centering
    \includegraphics[width=0.9\linewidth]{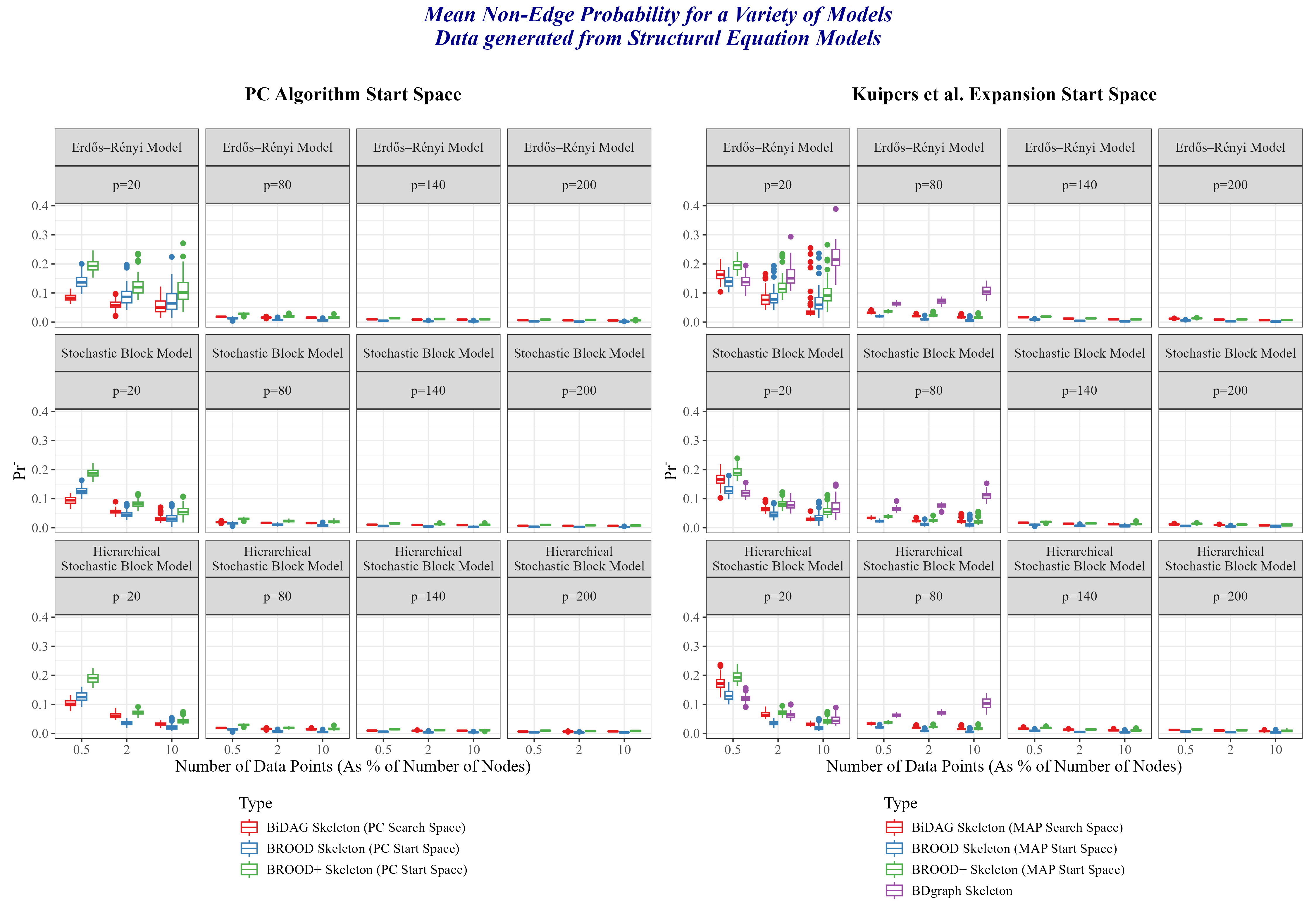}
    \caption{Skeleton version of $Pr^-$ results with Gaussian SEM data, using plus-one sparsity.}
    \label{fig:prminus_gauss_skel_plus1}
\end{figure}

\subsubsection{Performance: FCM, Plus-One Runs, DAG and Skeleton Outputs}
Changing the error model for data generation from a Gaussian SEM (which has $\mathcal{N}(0,1)$ errors) to a FCM (which is a mixture between $\mathcal{N}(0,1)$ errors and $\mathcal{N}(0,2)$ errors) has minimal impact on the results, both for DAGs and skeletons, as presented in Figures \ref{fig:auc_roc_fcm_dag_plus1}-\ref{fig:prminus_fcm_skel_plus1}. One interesting insight is that the relative performance from using DAGs to skeletons has a similar level of change for FCM data generation as for Gaussian SEM data generation. This suggests that the algorithm is not especially sensitive to identifiability issues that we may expect from a Gaussian SEM but not from a FCM. 

\begin{figure}
    \centering
    \includegraphics[width=0.9\linewidth]{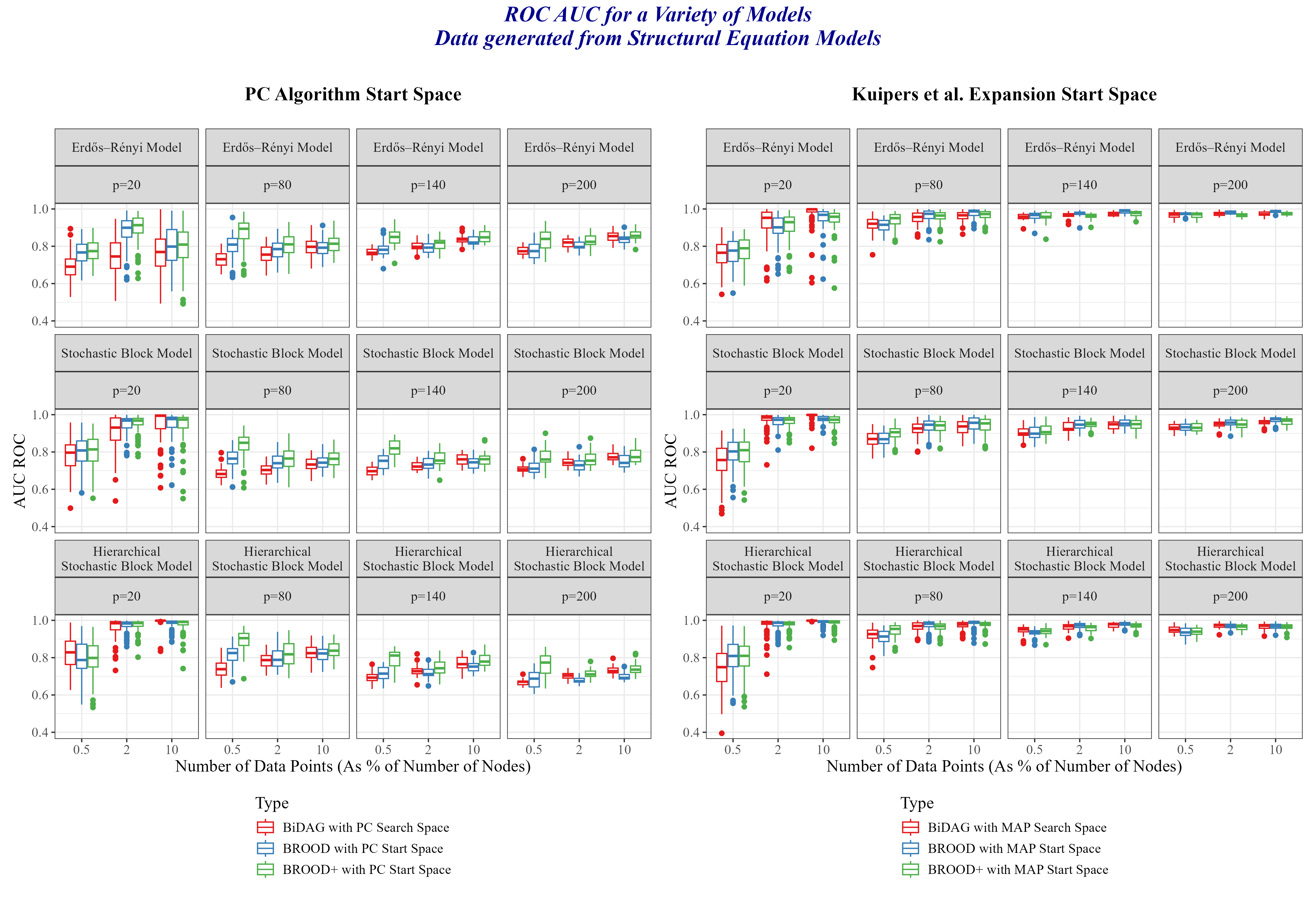}
    \caption{ROC AUC results with FCM data, using plus-one sparsity.}
    \label{fig:auc_roc_fcm_dag_plus1}
\end{figure}

\begin{figure}
    \centering
    \includegraphics[width=0.9\linewidth]{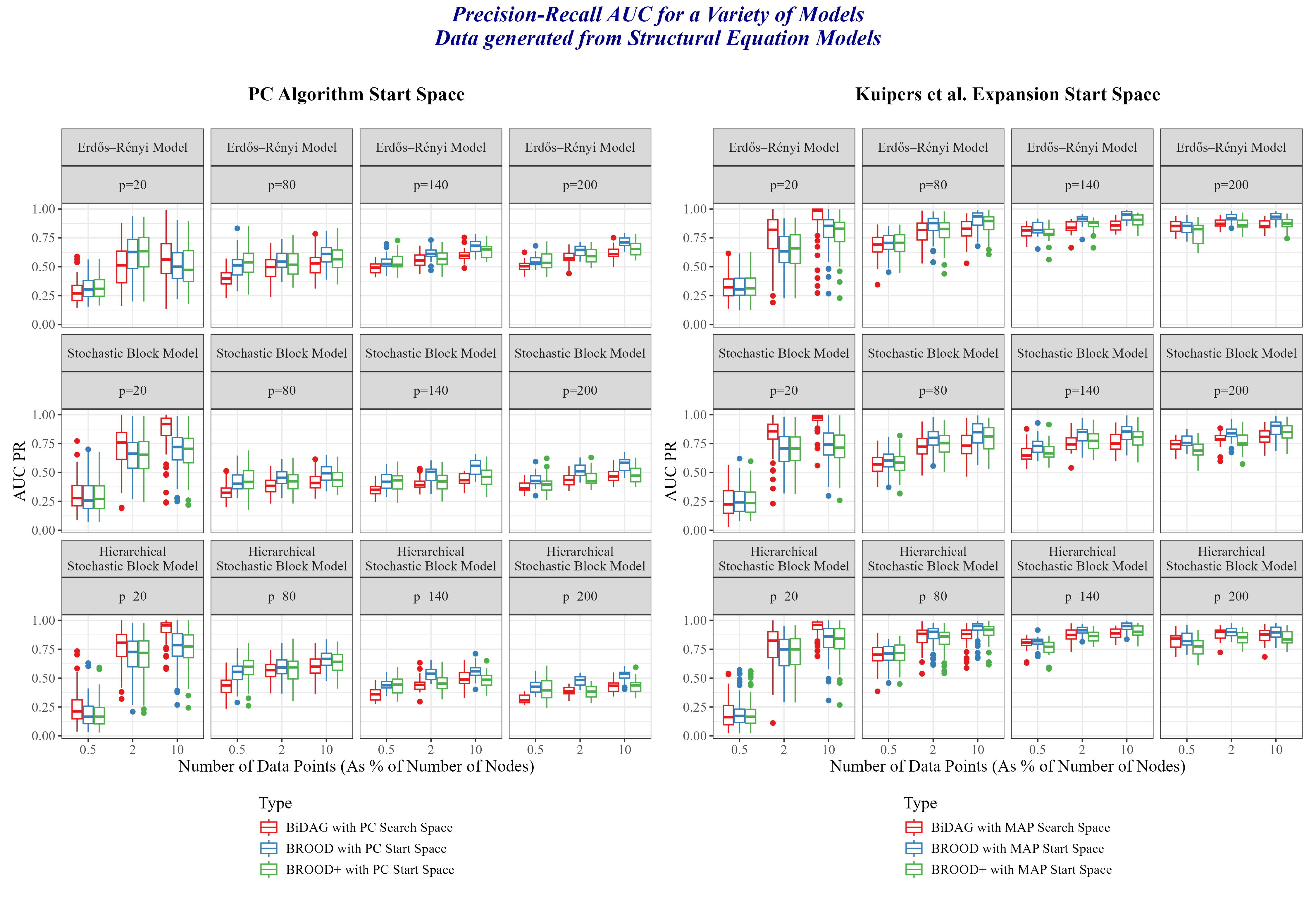}
    \caption{PR AUC results with FCM data, using plus-one sparsity.}
    \label{fig:auc_pr_fcm_dag_plus1}
\end{figure}

\begin{figure}
    \centering
    \includegraphics[width=0.9\linewidth]{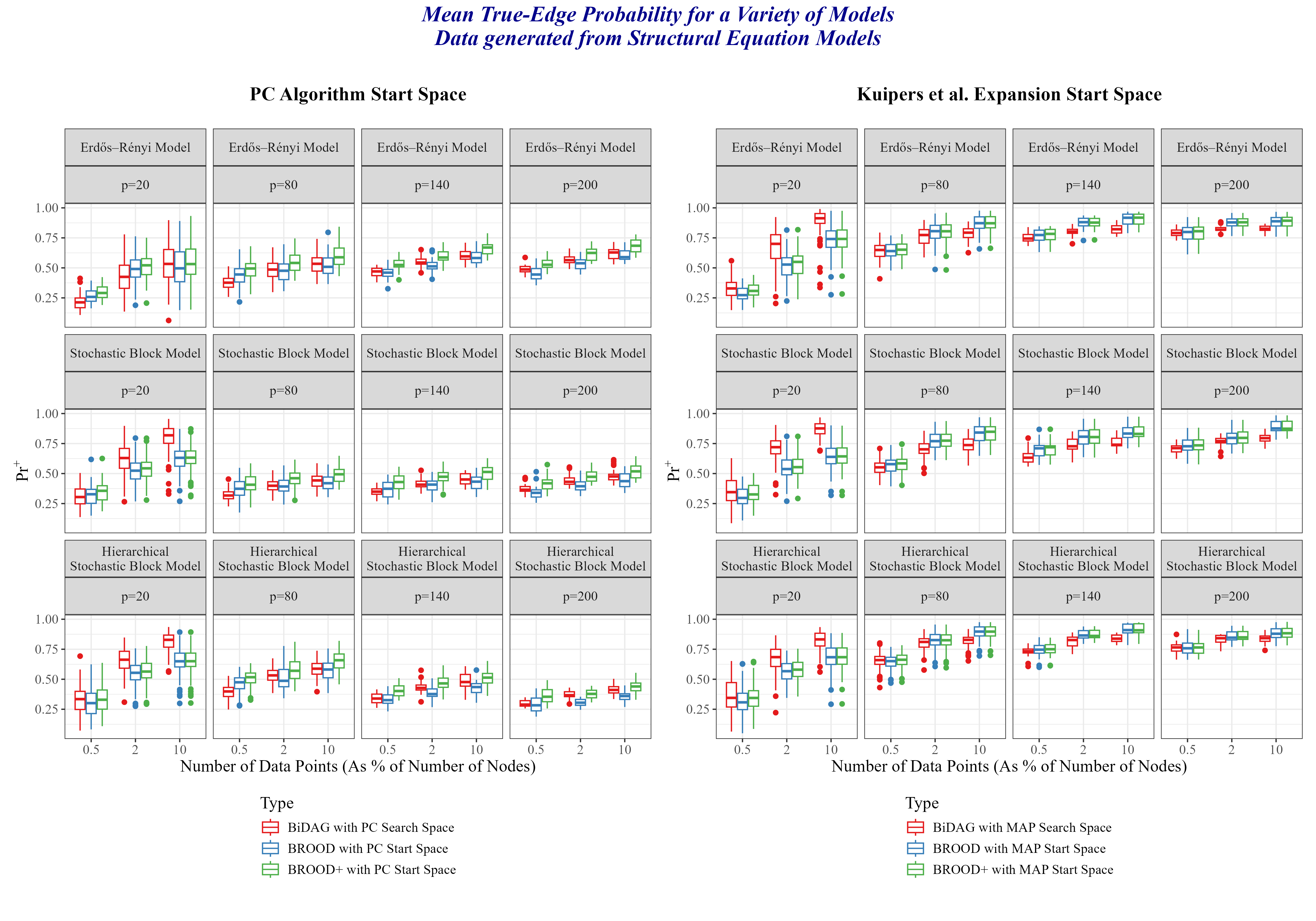}
    \caption{$Pr^+$ results with FCM data, using plus-one sparsity.}
    \label{fig:prplus_fcm_dag_plus1}
\end{figure}

\begin{figure}
    \centering
    \includegraphics[width=0.9\linewidth]{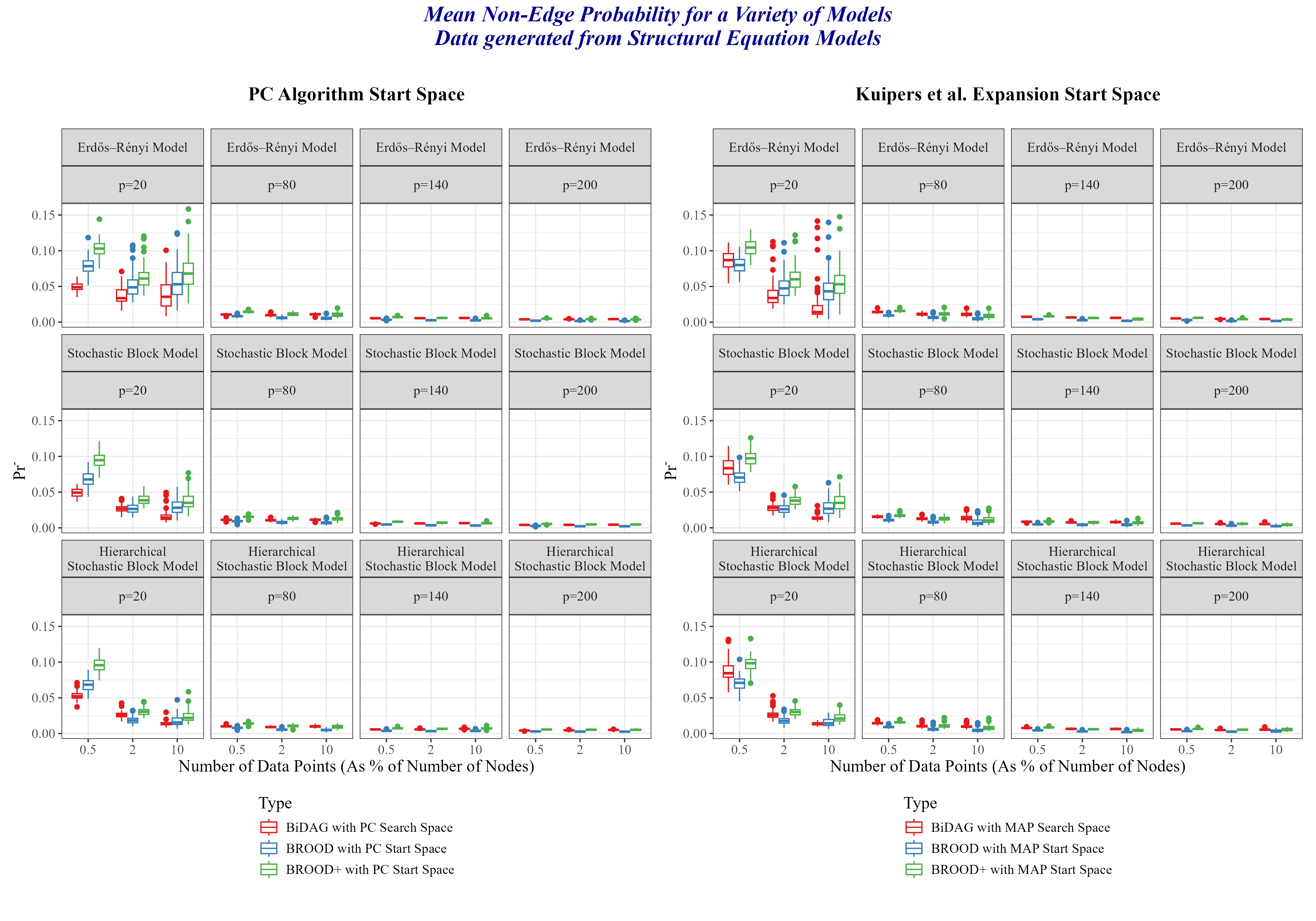}
    \caption{$Pr^-$ results with FCM data, using plus-one sparsity.}
    \label{fig:prminus_fcm_dag_plus1}
\end{figure}

\begin{figure}
    \centering
    \includegraphics[width=0.9\linewidth]{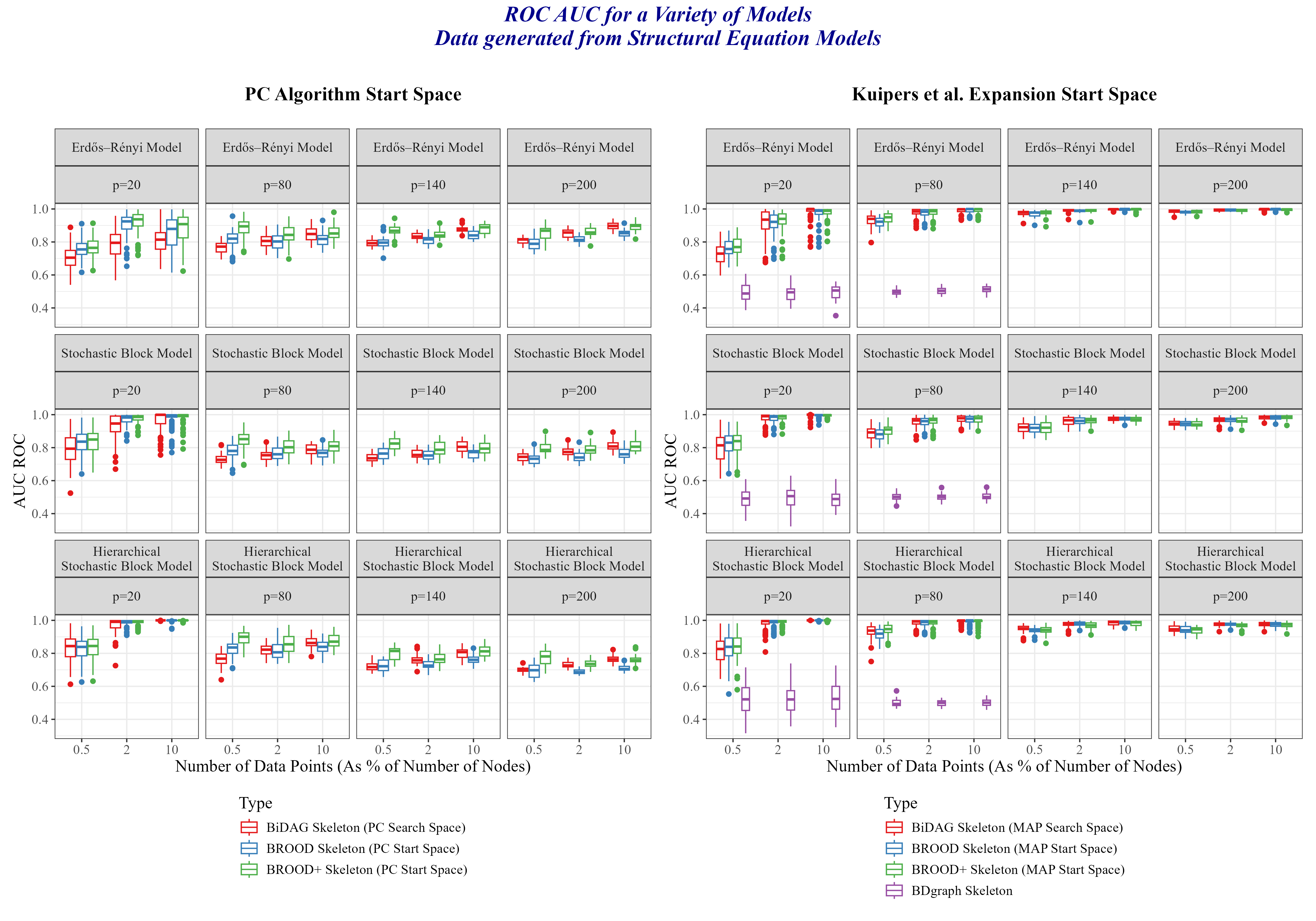}
    \caption{Skeleton version ROC AUC results with FCM data, using plus-one sparsity.}
    \label{fig:auc_roc_fcm_skel_plus1}
\end{figure}

\begin{figure}
    \centering
    \includegraphics[width=0.9\linewidth]{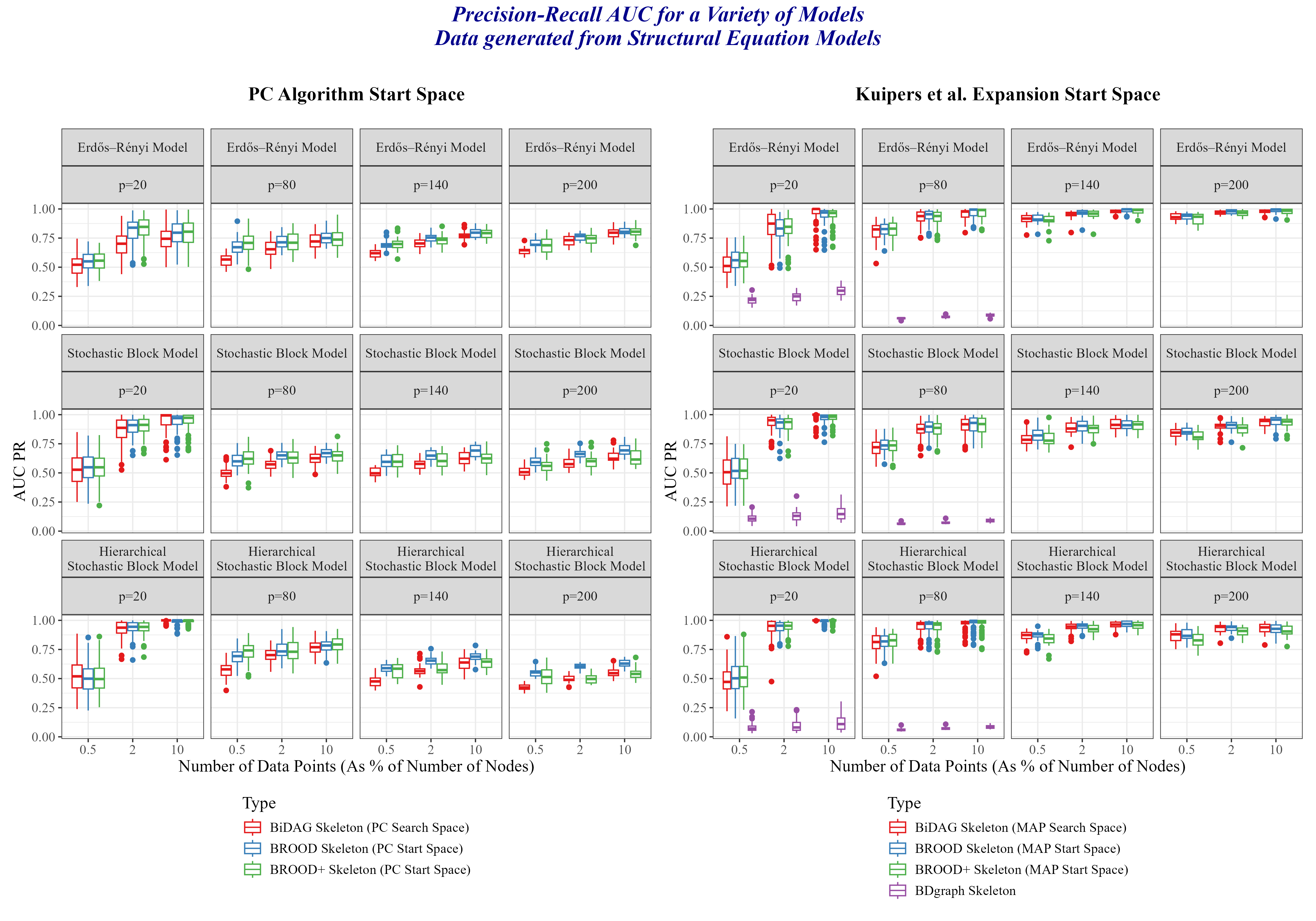}
    \caption{Skeleton version of PR AUC results with FCM data, using plus-one sparsity.}
    \label{fig:auc_pr_fcm_skel_plus1}
\end{figure}

\begin{figure}
    \centering
    \includegraphics[width=0.9\linewidth]{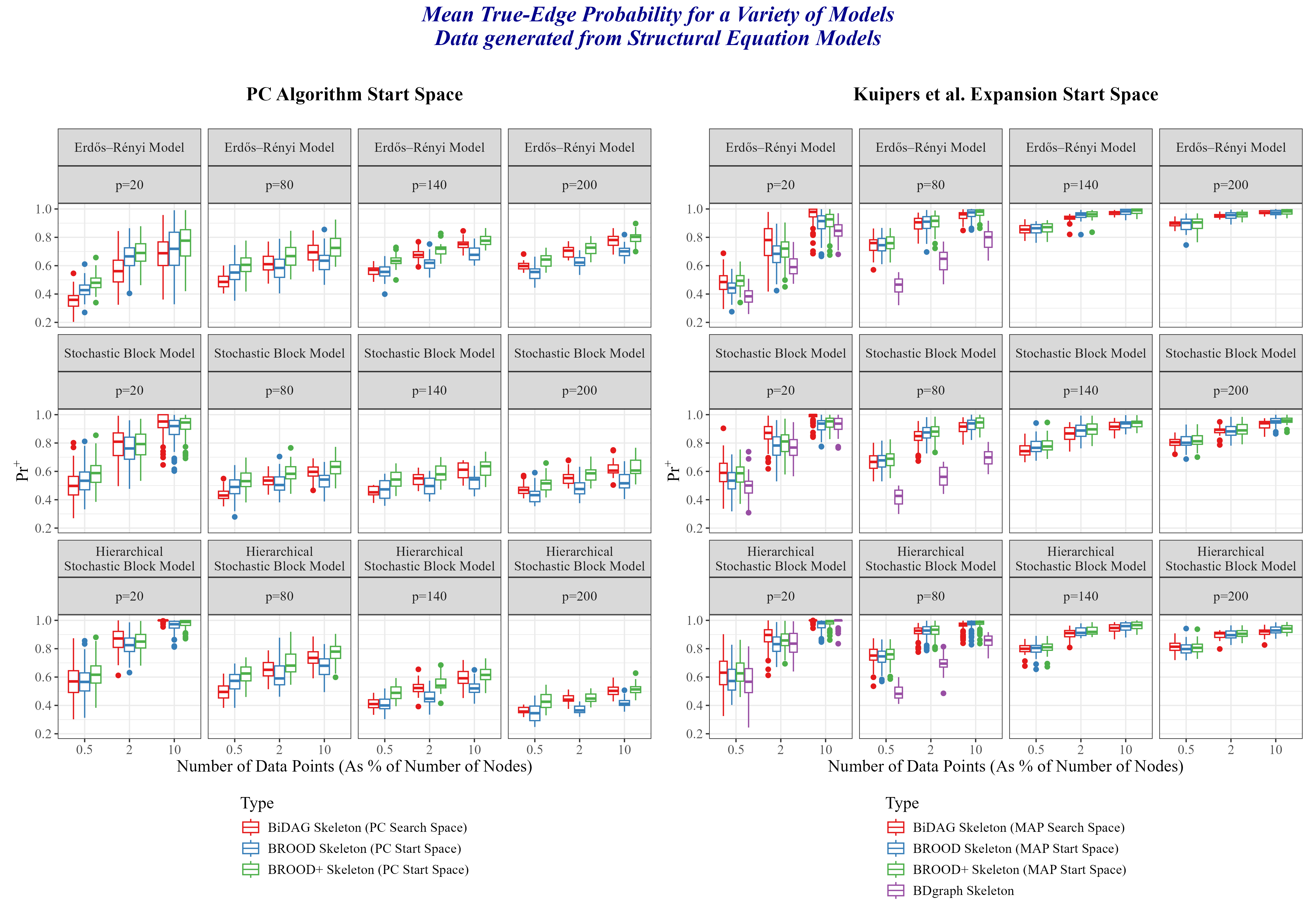}
    \caption{Skeleton version of $Pr^+$ results with FCM data, using plus-one sparsity.}
    \label{fig:prplus_fcm_skel_plus1}
\end{figure}

\begin{figure}
    \centering
    \includegraphics[width=0.9\linewidth]{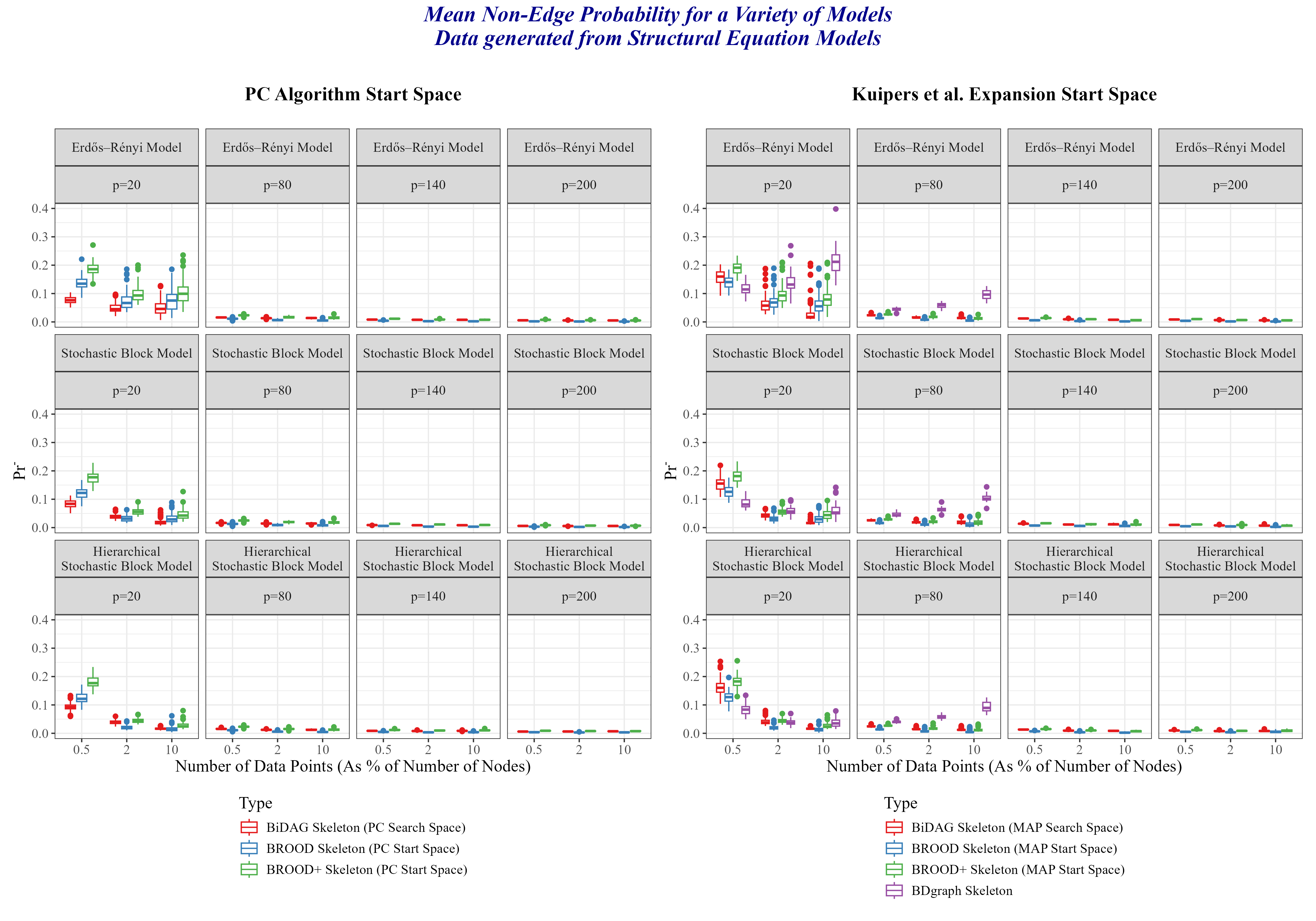}
    \caption{Skeleton version of $Pr^-$ results with FCM data, using plus-one sparsity.}
    \label{fig:prminus_fcm_skel_plus1}
\end{figure}

\subsubsection{Performance: Gaussian SEMs, Fixed Sparsity Runs}

Figure \ref{fig:auc_roc_gauss_dag2} capped the sparsity across to a fixed amount across all experiments for BROOD, regardless of the starting search spaces (whereas the plus-one sparsity allowed for the cap to be one more than the largest parent set size in the starting search space). Fixing the sparsity across initializations allows for us to see the extent to which BROOD can recover from poor initializations. We found in the main text that BROOD can best recover from a poorly initialized starting search space when the ratio between the number of data points to number of nodes is small; this implies that the more diffuse the score, the better BROOD can navigate the transdimensional space.

Figures \ref{fig:auc_pr_gauss_dag_fixed}-\ref{fig:prminus_gauss_skel_fixed} show further breakdowns of the Gaussian SEM models using fixed-sparsity capping. Figures \ref{fig:prplus_gauss_dag_fixed}, \ref{fig:prminus_gauss_dag_fixed}, \ref{fig:prplus_gauss_skel_fixed}, and \ref{fig:prminus_gauss_skel_fixed} are especially insightful into how BROOD recovers from poorly-initialized search spaces when the ratio of number of observations to number of nodes is 0.5: initializing with a PC search space, BROOD sometimes adds true edges to the space that the PC algorithm omits but also sometimes adds non-edges to the space, while initializing with the Kuipers et al. expansion search space, BROOD sometimes takes away true edges in order to exclude non-edges. Initializing with the GES search space lands in between the two. Overall, this amounts to similar $Pr^+$ and $Pr^-$ across all three initializations of the search space.

\begin{figure}
    \centering
    \includegraphics[width=0.9\linewidth]{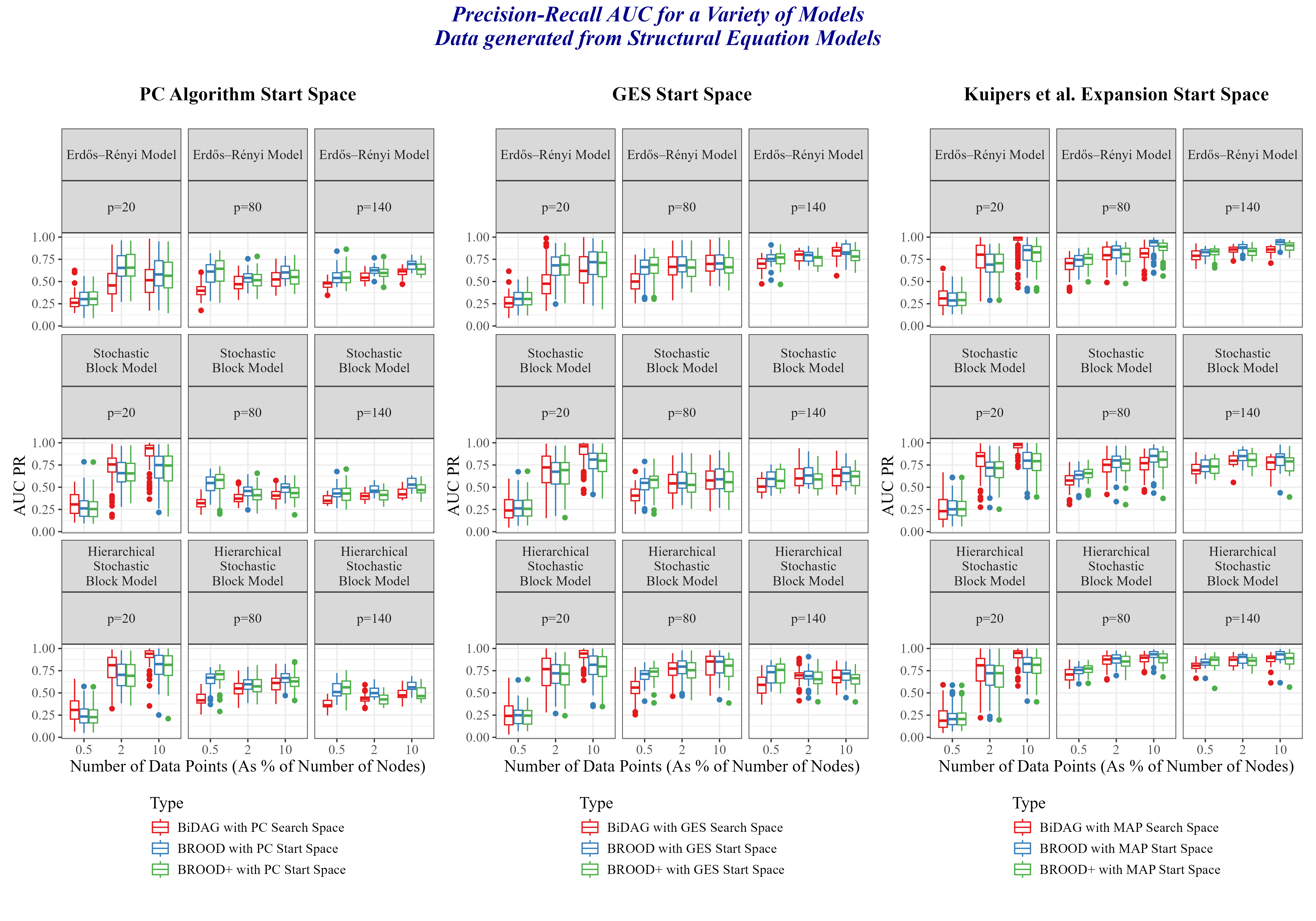}
    \caption{PR AUC results with Gaussian SEM data, using fixed sparsity.}
    \label{fig:auc_pr_gauss_dag_fixed}
\end{figure}

\begin{figure}
    \centering
    \includegraphics[width=0.9\linewidth]{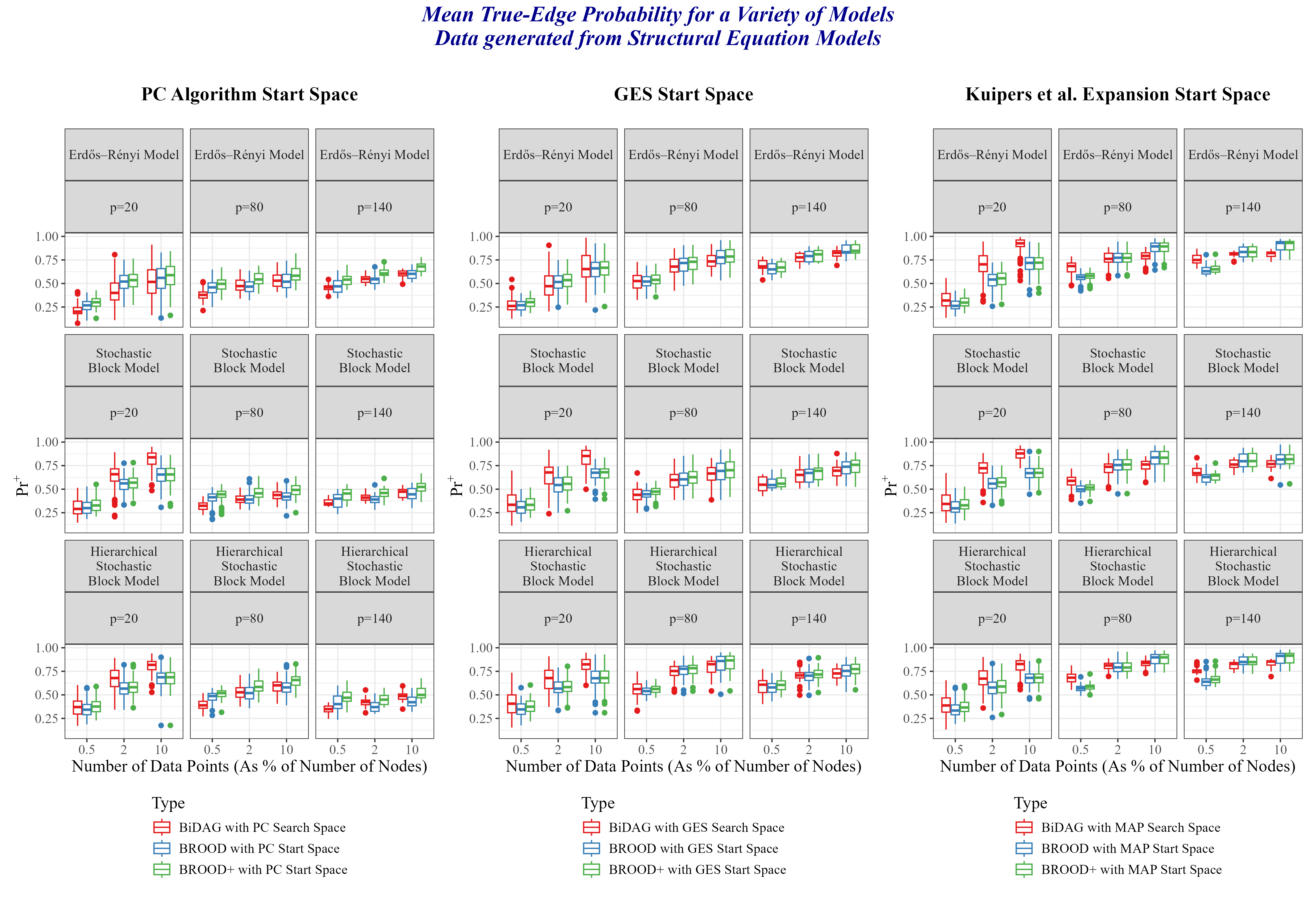}
    \caption{$Pr^+$ results with Gaussian SEM data, using fixed sparsity.}
    \label{fig:prplus_gauss_dag_fixed}
\end{figure}

\begin{figure}
    \centering
    \includegraphics[width=0.9\linewidth]{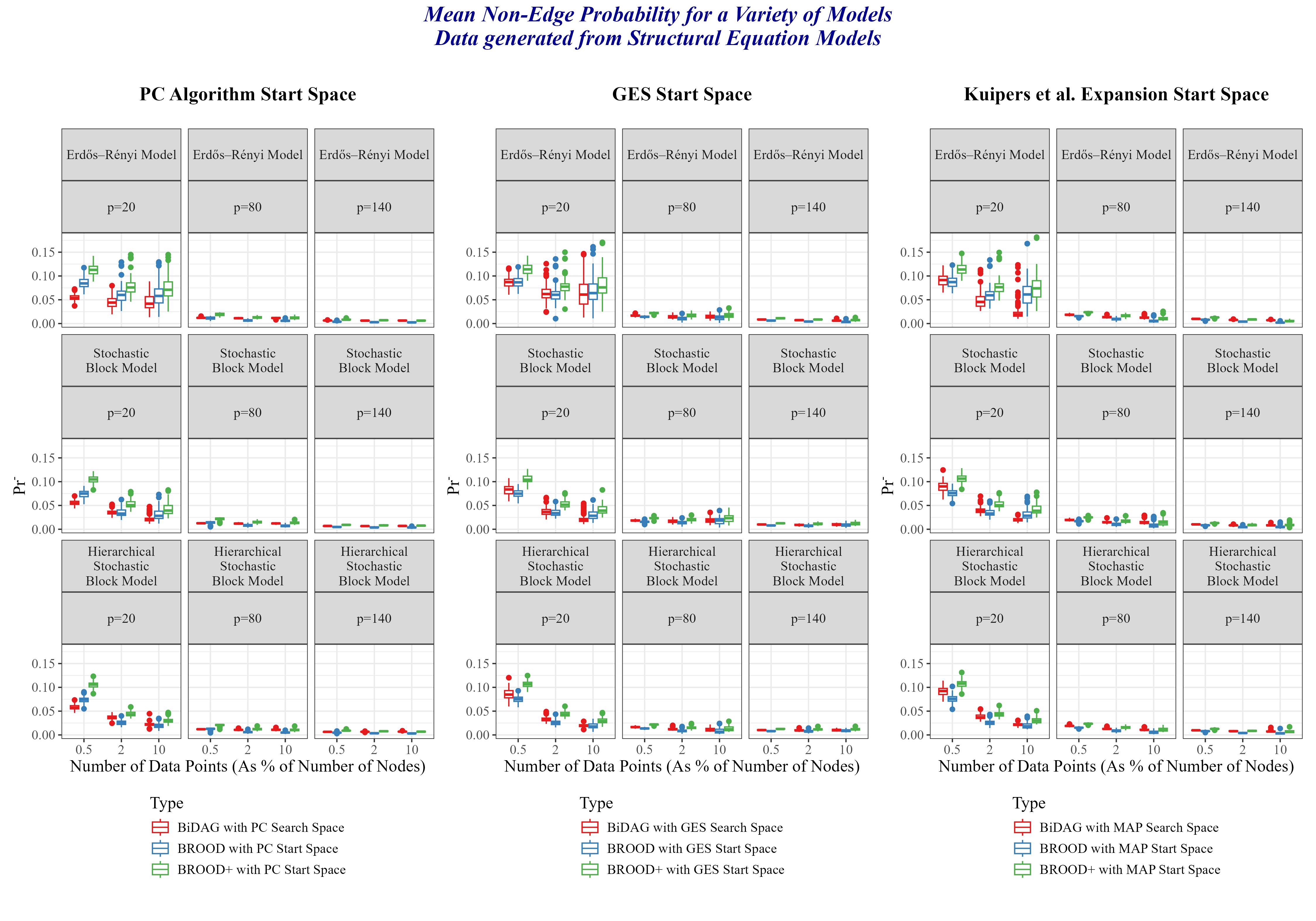}
    \caption{$Pr^-$ results with Gaussian SEM data, using fixed sparsity.}
    \label{fig:prminus_gauss_dag_fixed}
\end{figure}

\begin{figure}
    \centering
    \includegraphics[width=0.9\linewidth]{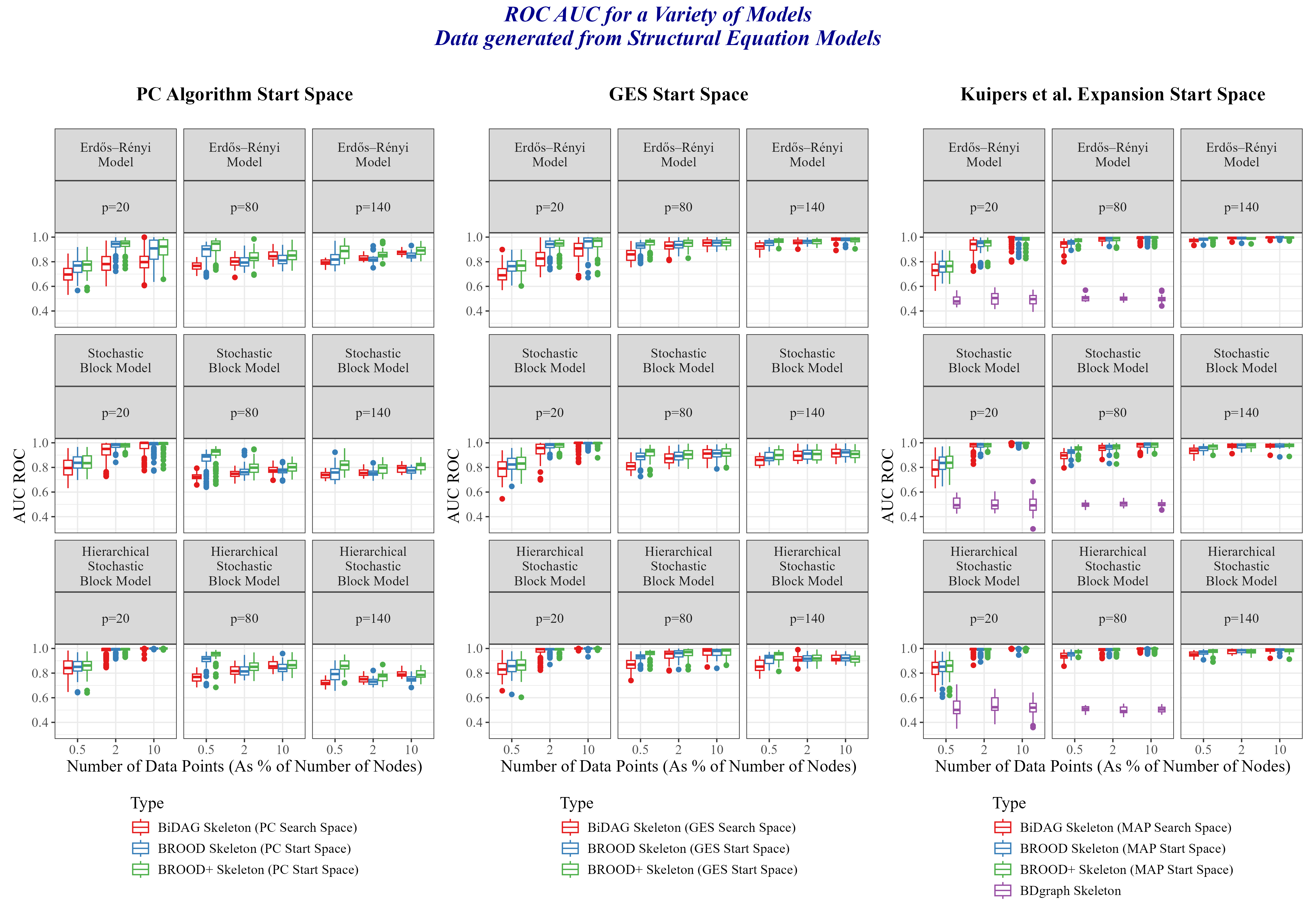}
    \caption{Skeleton version of ROC AUC results with Gaussian SEM data, using fixed sparsity.}
    \label{fig:auc_roc_gauss_skel_fixed}
\end{figure}

\begin{figure}
    \centering
    \includegraphics[width=0.9\linewidth]{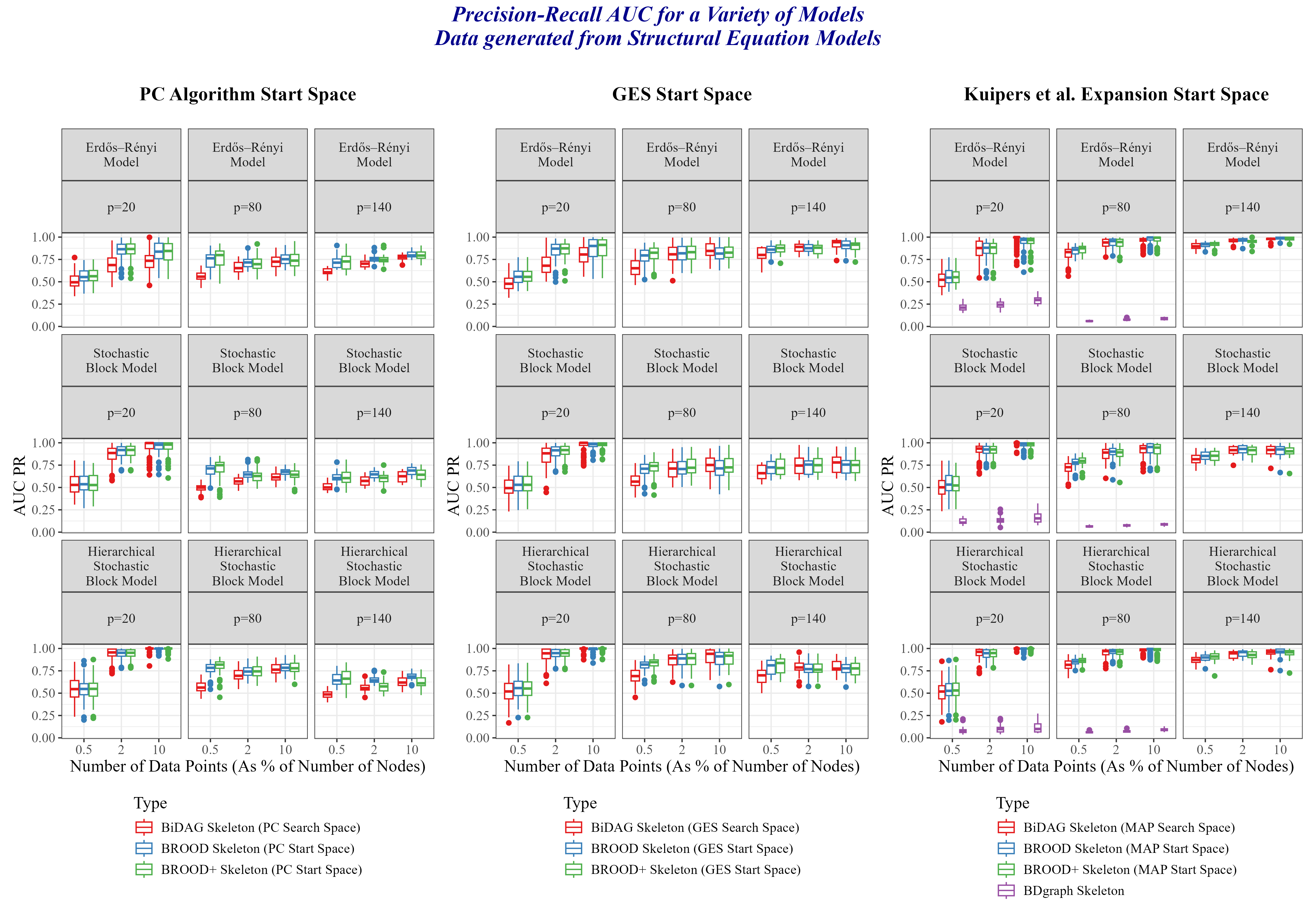}
    \caption{Skeleton version of PR AUC results with Gaussian SEM data, using fixed sparsity.}
    \label{fig:auc_pr_gauss_skel_fixed}
\end{figure}

\begin{figure}
    \centering
    \includegraphics[width=0.9\linewidth]{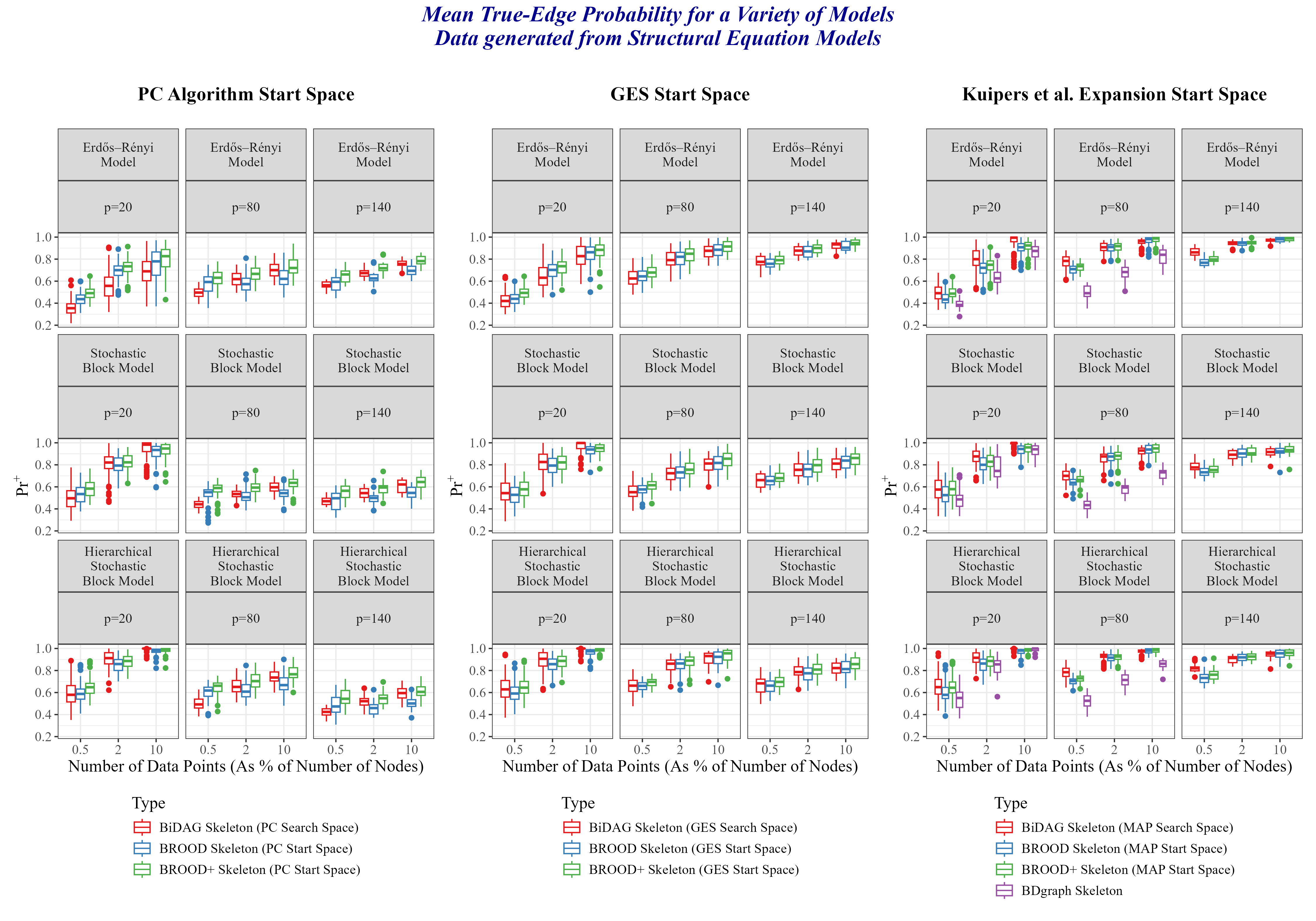}
    \caption{Skeleton version of $Pr^+$ results with Gaussian SEM data, using fixed sparsity.}
    \label{fig:prplus_gauss_skel_fixed}
\end{figure}

\begin{figure}
    \centering
    \includegraphics[width=0.9\linewidth]{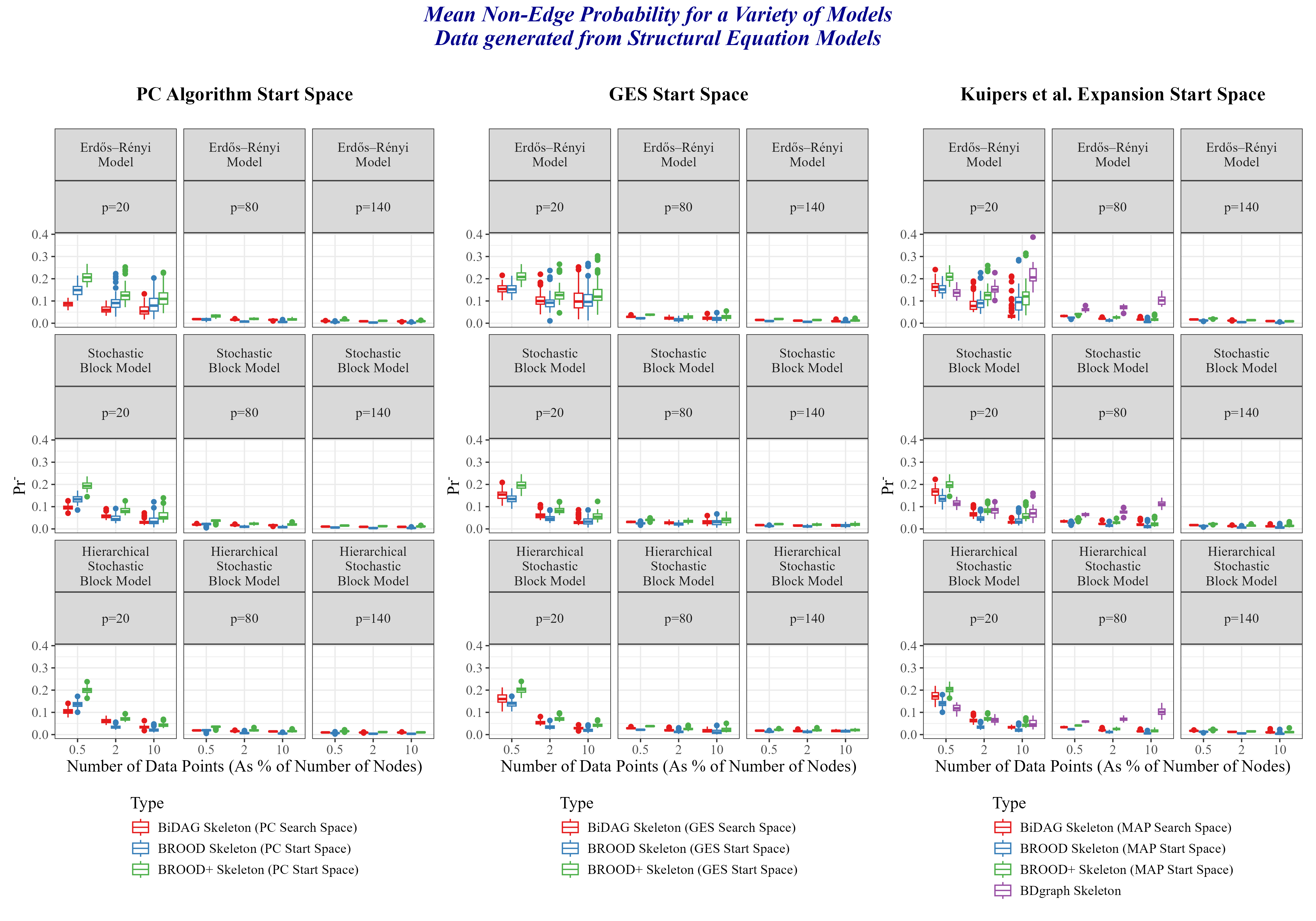}
    \caption{Skeleton version of $Pr^-$ results with Gaussian SEM data, using fixed sparsity.}
    \label{fig:prminus_gauss_skel_fixed}
\end{figure}

\subsubsection{Performance: FCM, Fixed Sparsity Runs}

Figures \ref{fig:auc_roc_fcm_dag_fixed}-\ref{fig:auc_pr_fcm_skel_fixed} show the ROC AUC and PR AUC (for both DAGs and skeletons) for FCM generated data with fixed sparsity capping. Using FCM data generation yields similar results to Gaussian SEM data generation when analyzing the impact of initialization with fixed sparsity capping: in settings with posteriors that are expected to be disperse, BROOD tends to be robust to recovering from poorly initialized search spaces, but is more sensitive to the initialization in settings where posteriors are expected to be peakier.

\begin{figure}
    \centering
    \includegraphics[width=0.9\linewidth]{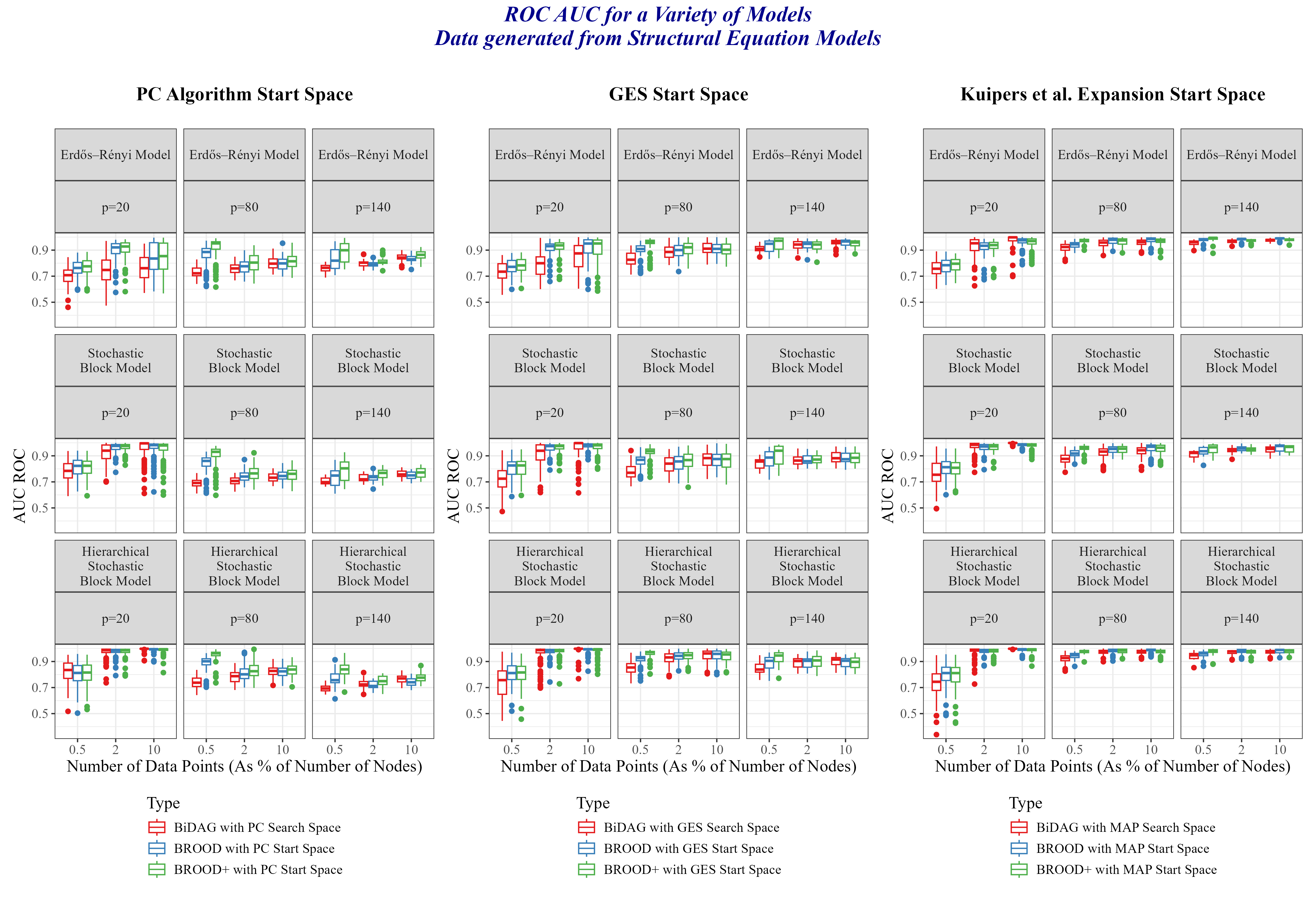}
    \caption{ROC AUC results with FCM data, using fixed sparsity.}
    \label{fig:auc_roc_fcm_dag_fixed}
\end{figure}

\begin{figure}
    \centering
    \includegraphics[width=0.9\linewidth]{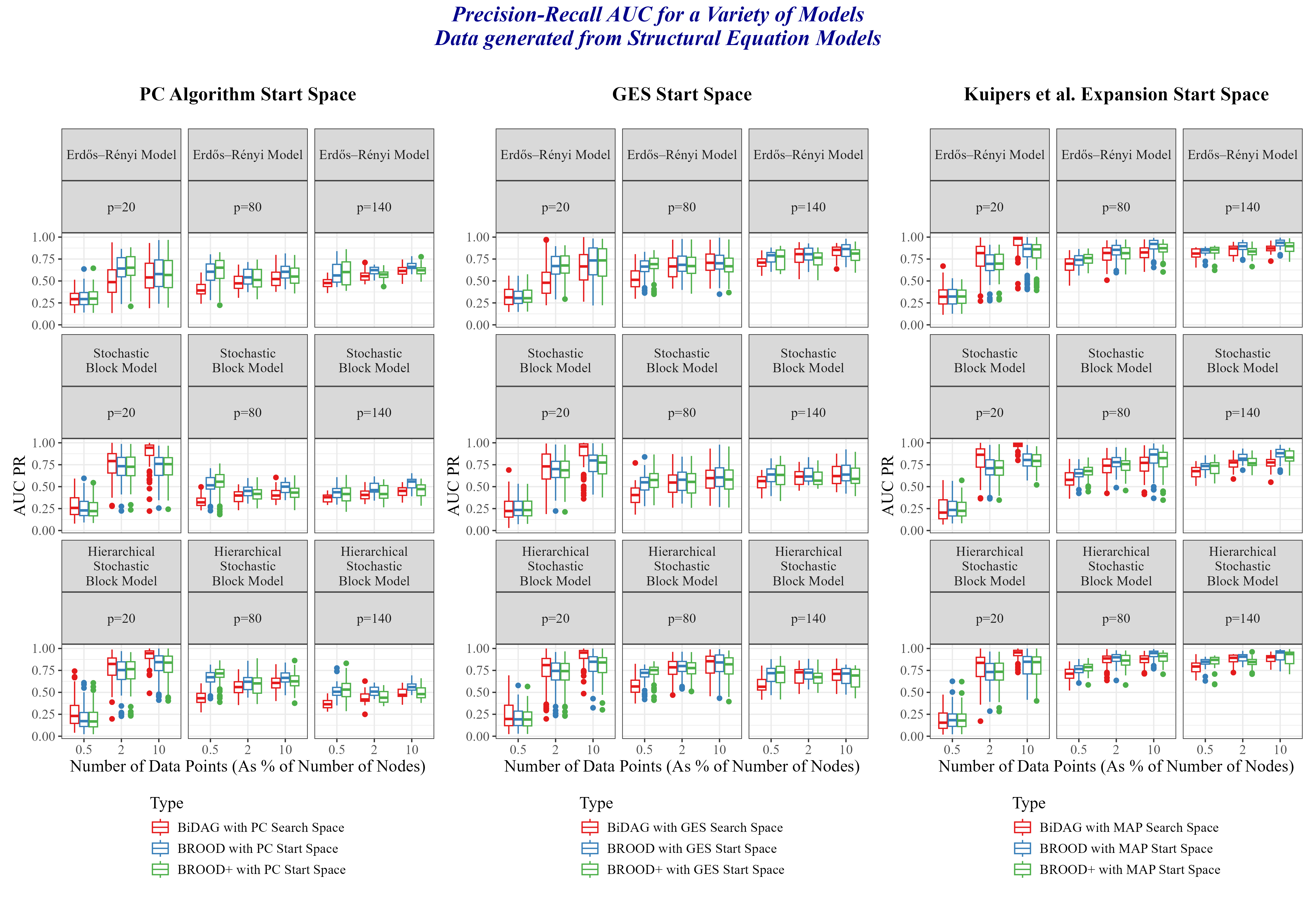}
    \caption{PR AUC results with FCM data, using fixed sparsity.}
    \label{fig:auc_pr_fcm_dag_fixed}
\end{figure}

\begin{figure}
    \centering
    \includegraphics[width=0.9\linewidth]{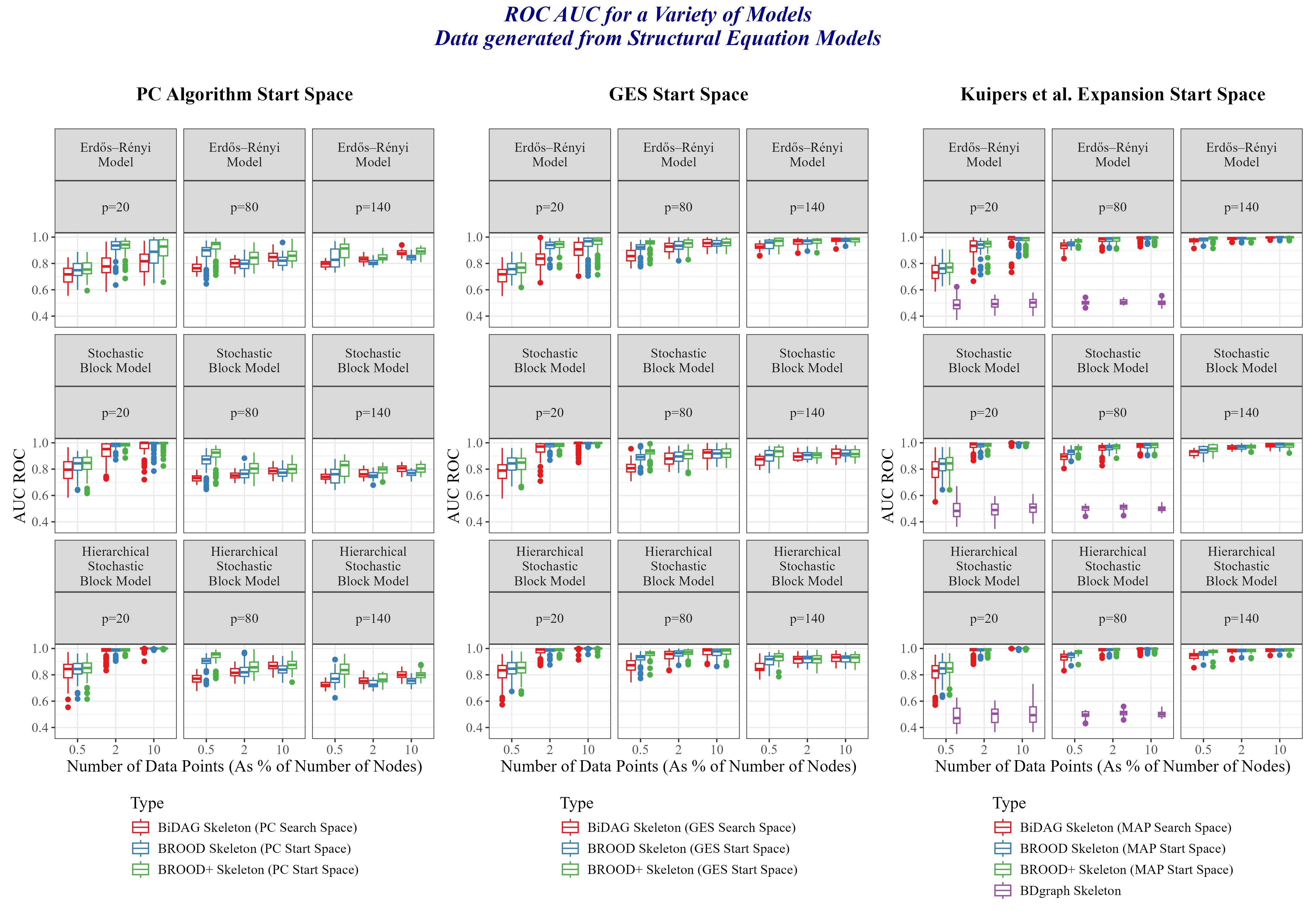}
    \caption{Skeleton version ROC AUC results with FCM data, using fixed sparsity.}
    \label{fig:auc_roc_fcm_skel_fixed}
\end{figure}

\begin{figure}
    \centering
    \includegraphics[width=0.9\linewidth]{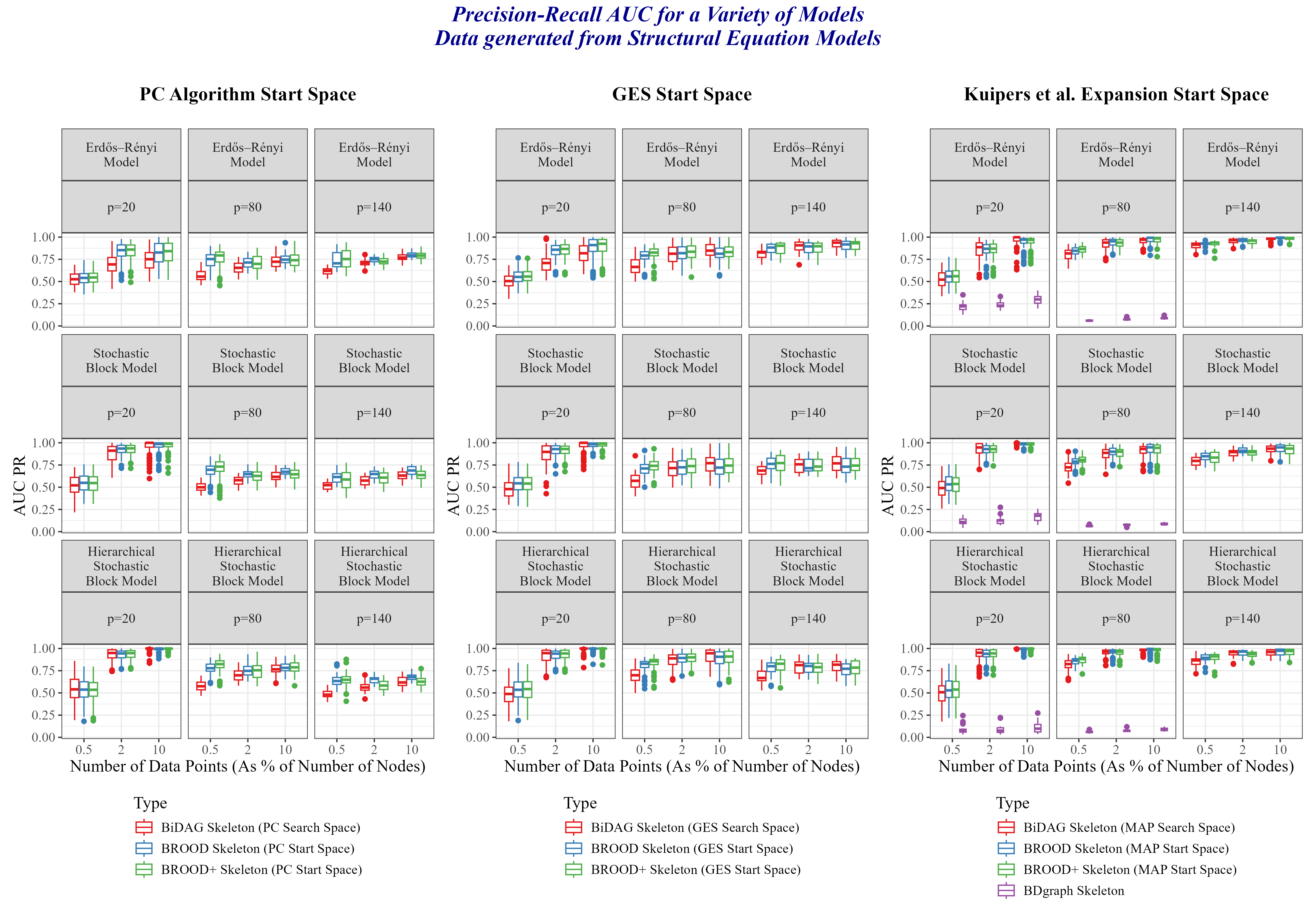}
    \caption{Skeleton version of PR AUC results with FCM data, using fixed sparsity.}
    \label{fig:auc_pr_fcm_skel_fixed}
\end{figure}

\subsubsection{Computation Time}

In the main text, we showed the downside of BROOD's flexibility in terms of log-time of the sampler finishing in Figure \ref{fig:log_time_gauss} on Gaussian SEM data, with plus-one sparsity. Figures \ref{fig:log_time_fcm_plus1}, \ref{fig:log_time_gauss_fixed}, \ref{fig:log_time_fcm_fixed} show the results for FCM data with plus-one sparsity, Gaussian SEM data with fixed sparsity, and FCM data with fixed sparsity, respectively.

Results for these are quite similar to those shown in the main text: BROOD is slower than BiDAG, but the gap closes as the problem gets more complex. This suggests that while there is a baked-in cost to running BROOD over BiDAG, if BROOD can be robust to a poorly initialized starting search space, for sufficiently difficult problems, it can be a better solution in both performance and computational time than existing methods for MCMC-based structure inference.

\begin{figure}
    \centering
    \includegraphics[width=0.9\linewidth]{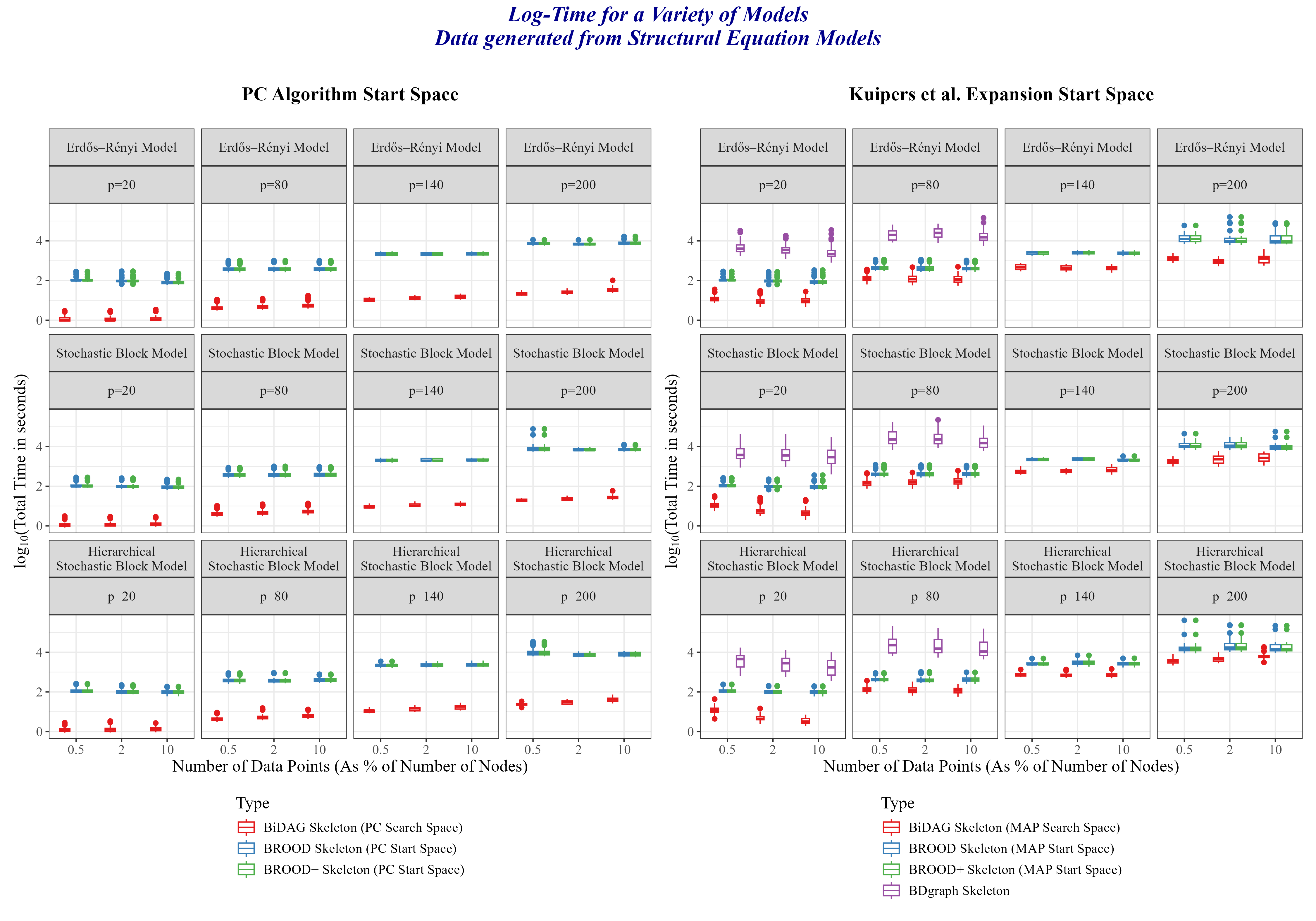}
    \caption{$\log_{10}(\text{Run Time})$ results for FCM data, using plus-one sparsity}
    \label{fig:log_time_fcm_plus1}
\end{figure}

\begin{figure}
    \centering
    \includegraphics[width=0.9\linewidth]{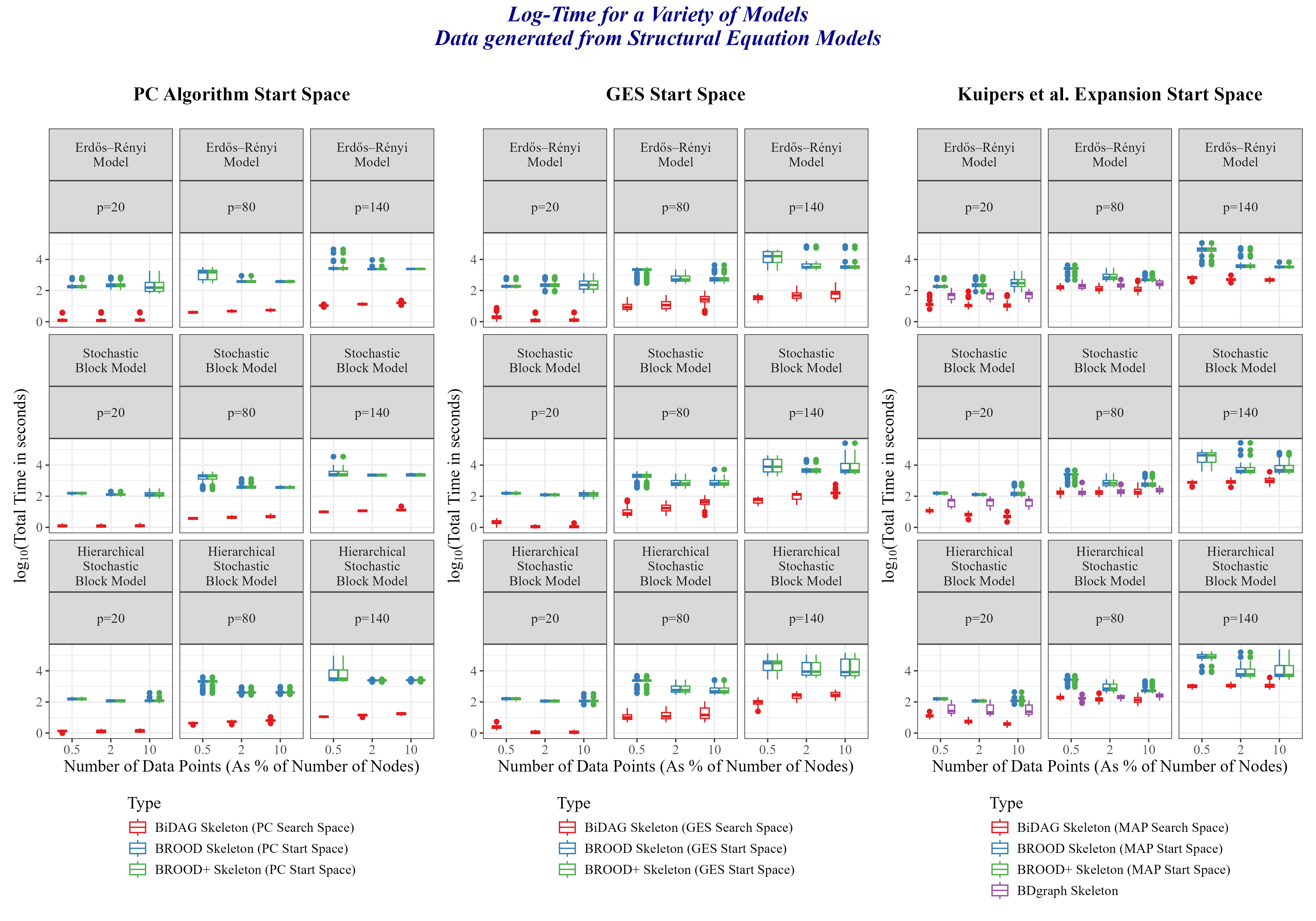}
    \caption{$\log_{10}(\text{Run Time})$ results for Gaussian SEM data, using fixed sparsity}
    \label{fig:log_time_gauss_fixed}
\end{figure}

\begin{figure}
    \centering
    \includegraphics[width=0.9\linewidth]{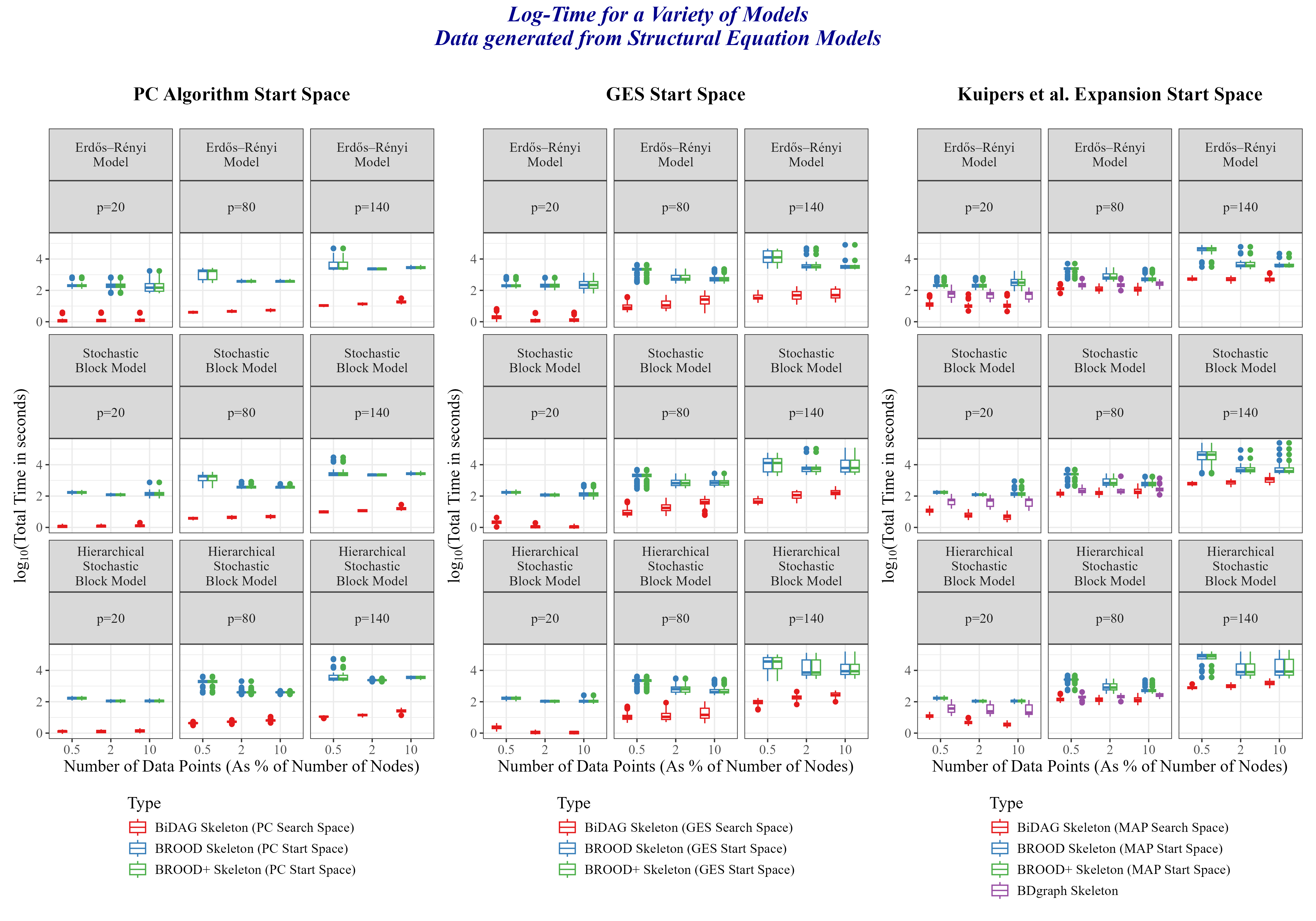}
    \caption{$\log_{10}(\text{Run Time})$ results for FCM data, using fixed sparsity}
    \label{fig:log_time_fcm_fixed}
\end{figure}

\end{document}

%% file: preamble/preamble.tex
\usepackage[T1]{fontenc}
\usepackage{tgtermes}

\usepackage{amssymb}

\usepackage{scalefnt,letltxmacro}
\LetLtxMacro{\oldtextsc}{\textsc}
\renewcommand{\textsc}[1]{\oldtextsc{\scalefont{1.10}#1}}
\usepackage[ttdefault=true]{AnonymousPro}

\usepackage{amsmath}
\usepackage{amsthm}
\usepackage{bbm}
\usepackage{bm}

\usepackage[
  paper  = letterpaper,
  left   = 1.25in,
  right  = 1.25in,
  top    = 1.0in,
  bottom = 1.0in,
  ]{geometry}
\usepackage[english]{babel}
\usepackage[parfill]{parskip}
\usepackage{afterpage}
\usepackage{enumitem}
\usepackage{framed}
\usepackage{xspace}

\usepackage{lineno}

\usepackage{ragged2e}

\newcounter{parcount}

\usepackage[dvipsnames]{xcolor}
\definecolor{shadecolor}{gray}{0.9}

\usepackage{graphicx}
\usepackage[labelfont=bf]{caption}
\usepackage[format=hang]{subcaption}
\usepackage{nicematrix}

\usepackage{booktabs, array, adjustbox}

\usepackage[algoruled,ruled,vlined]{algorithm2e}
\usepackage{listings}
\usepackage{fancyvrb}
\fvset{fontsize=\small}

\usepackage{natbib}
\usepackage[colorlinks,linktoc=all]{hyperref}
\usepackage[all]{hypcap}
\hypersetup{citecolor=MidnightBlue}
\hypersetup{linkcolor=MidnightBlue}
\hypersetup{urlcolor=MidnightBlue}
\usepackage[nameinlink]{cleveref}
\creflabelformat{equation}{#2\textup{#1}#3}  

\usepackage
[acronym,nowarn,section,nogroupskip,nonumberlist]{glossaries}
\glsdisablehyper{}

\usepackage[english]{babel}

\newtheorem{theorem}{Theorem}[section]
\newtheorem{corollary}{Corollary}[theorem]
\newtheorem{lemma}[theorem]{Lemma}

\newtheorem{remark}{Remark}

\lstdefinestyle{mystyle}{
    commentstyle=\color{OliveGreen},
    numberstyle=\tiny\color{black!60},
    stringstyle=\color{BrickRed},
    basicstyle=\ttfamily\scriptsize,
    breakatwhitespace=false,
    breaklines=true,
    captionpos=b,
    keepspaces=true,
    numbers=none,
    numbersep=5pt,
    showspaces=false,
    showstringspaces=false,
    showtabs=false,
    tabsize=2
}
\usepackage{dsfont}

\usepackage[sc,compact]{titlesec}
\titleformat{\section}
  {\large\normalfont\scshape}{\thesection}{1em}{}
\titleformat{\subsection}
  {\normalfont\scshape}{\thesubsection}{1em}{}

%% file: preamble/preamble_tikz.tex
\usepackage{tikz}

\usepackage{subcaption}
\usetikzlibrary{arrows.meta, positioning, decorations.pathmorphing}

\usetikzlibrary{bayesnet} 

\pgfdeclarelayer{edgelayer}
\pgfdeclarelayer{nodelayer}
\pgfsetlayers{edgelayer,nodelayer,main}

\definecolor{hexcolor0xbfbfbf}{rgb}{0.749,0.749,0.749}

\tikzset{>=latex}
\tikzstyle{none}   = [inner sep=0pt]
\tikzstyle{line}   = [ -, thick, shorten <=1pt, shorten >=1pt ]
\tikzstyle{arrow}  = [ ->, thick, shorten <=1pt, shorten >=1pt ]
\tikzstyle{ardash} = [ dashed, ->, thick, shorten <=1pt, shorten >=1pt ]

\tikzstyle{empty}=[circle,opacity=0.0,text opacity=1.0,inner sep=0pt]
\tikzstyle{box}=[rectangle,fill=White,draw=Black]
\tikzstyle{filled}=[circle,thick,fill=hexcolor0xbfbfbf,draw=Black]
\tikzstyle{hollow}=[circle,thick,fill=White,draw=Black]
\tikzstyle{param}=[rectangle,fill=Black,draw=Black,inner sep=0pt,minimum width=4pt,minimum height=4pt]
\tikzstyle{paramhollow}=[rectangle,thick,fill=White,draw=Black,inner sep=0pt,minimum
width=4pt,minimum height=4pt]

\usepackage{pgfplots}                               
\pgfplotsset{compat=newest}
\pgfplotsset{plot coordinates/math parser=false}
\newlength\figureheight
\newlength\figurewidth
\newlength\figureheightsmall
\newlength\figurewidthsmall
\definecolor{POSTcolor}{rgb}{0.48, 0.20, 0.58} 
\definecolor{Qcolor}{rgb}{0.00, 0.53, 0.22} 